\documentclass[11pt,letterpaper]{article}
\usepackage{amssymb}

\usepackage{amsmath}
\usepackage{float}
\usepackage[margin=.9in]{geometry}
\usepackage[hang]{caption}
\usepackage{color}
\usepackage[dvips, pdftex]{graphicx}
\usepackage{psfrag}
\usepackage[english]{babel}
\usepackage{wasysym}
\usepackage{textcomp}
\usepackage{bm}
\usepackage{enumitem}
\usepackage{adjustbox}
\usepackage{amsthm,array}
\usepackage{tabularx}
\usepackage{booktabs}
\usepackage{rotating}
\usepackage{bigstrut}
\usepackage{multirow}
\usepackage{subfigure}
\RequirePackage[round,author year,longnamesfirst]{natbib}
\usepackage{hyperref}
\usepackage[svgnames]{xcolor}
\hypersetup{
	colorlinks=true,
	linkcolor=DarkRed,
	citecolor=NavyBlue}

\usepackage[toc,title,page]{appendix}

\newtheorem{theorem}{Theorem}[section]

\newtheorem{assumption}{Assumption}

\newtheorem{lemma}{Lemma}[section]

\newtheorem{proposition}{Proposition}[section]

\usepackage[english]{babel}
\newtheoremstyle{boldremark}
    {\dimexpr\topsep/2\relax} 
    {\dimexpr\topsep/2\relax} 
    {}          
    {}          
    {\bfseries} 
    {.}         
    {.5em}      
    {}          

\theoremstyle{boldremark}
\newtheorem{example}{Example}
\newtheorem{remark}{Remark}[section]

\newcommand{\diag}{\text{diag}}

\newcommand{\indep}{\perp\!\!\!\perp}

\newcommand{\convP}{\stackrel{p}{\longrightarrow}}
\newcommand{\convD}{\rightsquigarrow}

\newcommand{\N}{\mathcal{N}}
\newcommand{\eps}{\varepsilon}

\renewcommand{\epsilon}{\varepsilon}

\DeclareMathOperator*{\argmin}{arg\,min}

\usepackage{xr}
\newlength{\bibitemsep}
\newlength{\bibparskip}
\let\oldthebibliography\thebibliography
\renewcommand\thebibliography[1]{%
  \oldthebibliography{#1}%
  \setlength{\parskip}{\bibitemsep}%
  \setlength{\itemsep}{\bibparskip}%
}

\makeatletter
\renewcommand\footnotesize{%
   \@setfontsize\footnotesize\@ixpt{11}%
   \abovedisplayskip 8\p@ \@plus2\p@ \@minus4\p@
   \abovedisplayshortskip \z@ \@plus\p@
   \belowdisplayshortskip 4\p@ \@plus2\p@ \@minus2\p@
   \def\@listi{\leftmargin\leftmargini
               \topsep 4\p@ \@plus2\p@ \@minus2\p@
               \parsep 2\p@ \@plus\p@ \@minus\p@
               \itemsep \parsep}%
   \belowdisplayskip \abovedisplayskip
}
\makeatother
\allowdisplaybreaks

\begin{document}
\title{\huge{Gradient Wild Bootstrap for Instrumental Variable Quantile Regressions with Weak and Few Clusters}\thanks{We would like to thank Aureo de Paula, Firmin Doko Tchatoka, Kirill Evdokimov, S\'ilvia Gon\c calves,  Christian Hansen, Jungbin Hwang, Qingfeng Liu, Morten Ørregaard Nielsen, Ryo Okui, Kyungchul (Kevin) Song, Naoya Sueishi, Yoshimasa Uematsu, and participants at the 2021 Annual Conference of the International Association for Applied Econometrics, the 2021 Asian Meeting of the Econometric Society, the 2021 China Meeting of the Econometric Society, the 2021 Australasian Meeting of the Econometric Society, the 2021 Econometric Society European Meeting, the 37th Canadian Econometric Study Group Meetings, the 16th International Symposium on Econometric Theory and Applications, UCONN Econometrics Seminar, NUS Econometrics Seminar, the 2023 North American Winter Meeting of the Econometric Society for their valuable comments. Special thanks to James MacKinnon for very insightful discussions and advice. Wang acknowledges the financial support from the Singapore Ministry of Education Tier 1 grants RG104/21, RG51/24, and NTU CoHASS Research Support Grant. 
Zhang acknowledges the financial support from the NSFC under grant No. 72133002. Any possible errors are our own.}}

\author{
Wenjie Wang\footnote{Division of Economics, School of Social Sciences, Nanyang Technological University.
HSS-04-65, 14 Nanyang Drive, Singapore 637332. 
E-mail address: wang.wj@ntu.edu.sg.}  
\ and Yichong Zhang\footnote{School of Economics, Singapore Management University. E-mail address: yczhang@smu.edu.sg. Corresponding author.}
}

\date{\today}

\maketitle

\begin{abstract}

We study the gradient wild bootstrap-based inference for instrumental variable quantile regressions in the framework of a small number of large clusters in which the number of clusters is viewed as fixed, and the number of observations for each cluster diverges to infinity. For the Wald inference, we show that our wild bootstrap Wald test, with or without studentization using the cluster-robust covariance estimator (CRVE), controls size asymptotically up to a small error as long as the parameter of endogenous variable is strongly identified in at least one of the clusters. We further show that the wild bootstrap Wald test with CRVE studentization is more powerful for distant local alternatives than that without. Last,  we develop a wild bootstrap Anderson-Rubin (AR) test for the weak-identification-robust inference. We show it controls size asymptotically up to a small error, even under weak or partial identification for all clusters. We illustrate the good finite-sample performance of the new inference methods using simulations and provide an empirical application to a well-known dataset about US local labor markets.\\

     \bigskip
   
\end{abstract}


\setlength{\baselineskip}{18pt}
\setlength{\abovedisplayskip}{10pt}
\belowdisplayskip\abovedisplayskip
\setlength{\abovedisplayshortskip }{5pt}
\abovedisplayshortskip \belowdisplayshortskip%
\setlength{\abovedisplayskip}{8pt} \belowdisplayskip\abovedisplayskip%
\setlength{\abovedisplayshortskip }{4pt} %
\linespread{1.3}
\large

 \section{Introduction}

The instrumental variable (IV) regression is one of the five most widely used methods for causal inference, as highlighted by \cite{Angrist-Pischke(2008)}, and it is often employed in analyses involving clustered data. For instance, \cite{young2022consistency} examines 1,359 IV regressions across 31 papers published by the American Economic Association (AEA), with 24 of these papers accounting for the clustering of observations. In the context of quantile regression models, where endogeneity may be present, \cite{Chernozhukov-Hansen(2004), Chernozhukov-Hansen(2005), Chernozhukov-Hansen(2006), Chernozhukov-Hansen(2008a)} (hereafter referred to as CH) developed an instrumental variables quantile regression (IVQR) method. This approach offers a general IV procedure to address the endogeneity of regressors in quantile regressions, and it has been widely adopted by empirical researchers to capture the distributional effects of endogenous variables.


However, three difficulties arise when running IVQR with clustered data. 
First, the number of clusters is small in many empirical applications with IVs.
For instance, \cite{Acemoglu2011} cluster the standard errors at the country/polity level, resulting in 12-19 clusters, \cite{Glitz2020} cluster at the sectoral level with 16 sectors, and \cite{Rogall2021} clusters at the province (district) level with 11 provinces (30 districts), respectively.  \cite{ADH2013} study the effects of Chinese imports on local labor markets in the US by clustering at the state level. However, if we focus on the estimation and inference for the effects on specific regions such as the South region designated by the US Census Bureau, there are only 16 states. Furthermore, \cite{Bester-Conley-Hansen(2011)} and \cite{L22} partition spatial and network data into clusters, respectively. Both papers consider the asymptotic setting in which the number of clusters is small and, thus, treated as fixed. When the number of clusters is small, conventional cluster-robust inference procedures may be unreliable for IVQR. 


Second, in many applications, the strength of IVs may be relatively heterogeneous across clusters, with a few clusters providing the main identification power. For instance, Figure \ref{fig:first_stage} reports the estimated first-stage coefficients for each cluster (state) from the South region in \citeauthor{ADH2013}'s (\citeyear{ADH2013}) dataset, which suggests that there exists substantial variation in the IV strength among states. Specifically, the first-stage coefficients of some states are relatively large compared with the rest in the region. In contrast, some other states have coefficients that are rather close to zero and, thus, potentially subject to weak identification. Some states even have opposite signs for their first-stage coefficients. However, there is no existing proven-valid inference method for IVQR with few clusters, where identification may be weak in some clusters. 

Third, it is also possible that IVs are weak in all clusters, in which case researchers need to use weak-identification-robust inference methods for IVQR \citep{Chernozhukov-Hansen(2008a), chernozhukov2009finite,Andrews-Stock-Sun(2019)}. 

\begin{figure}
          \makebox[\textwidth]{\includegraphics[width=0.5\paperwidth,height=0.3\textheight, trim={0 0 0 1cm}]{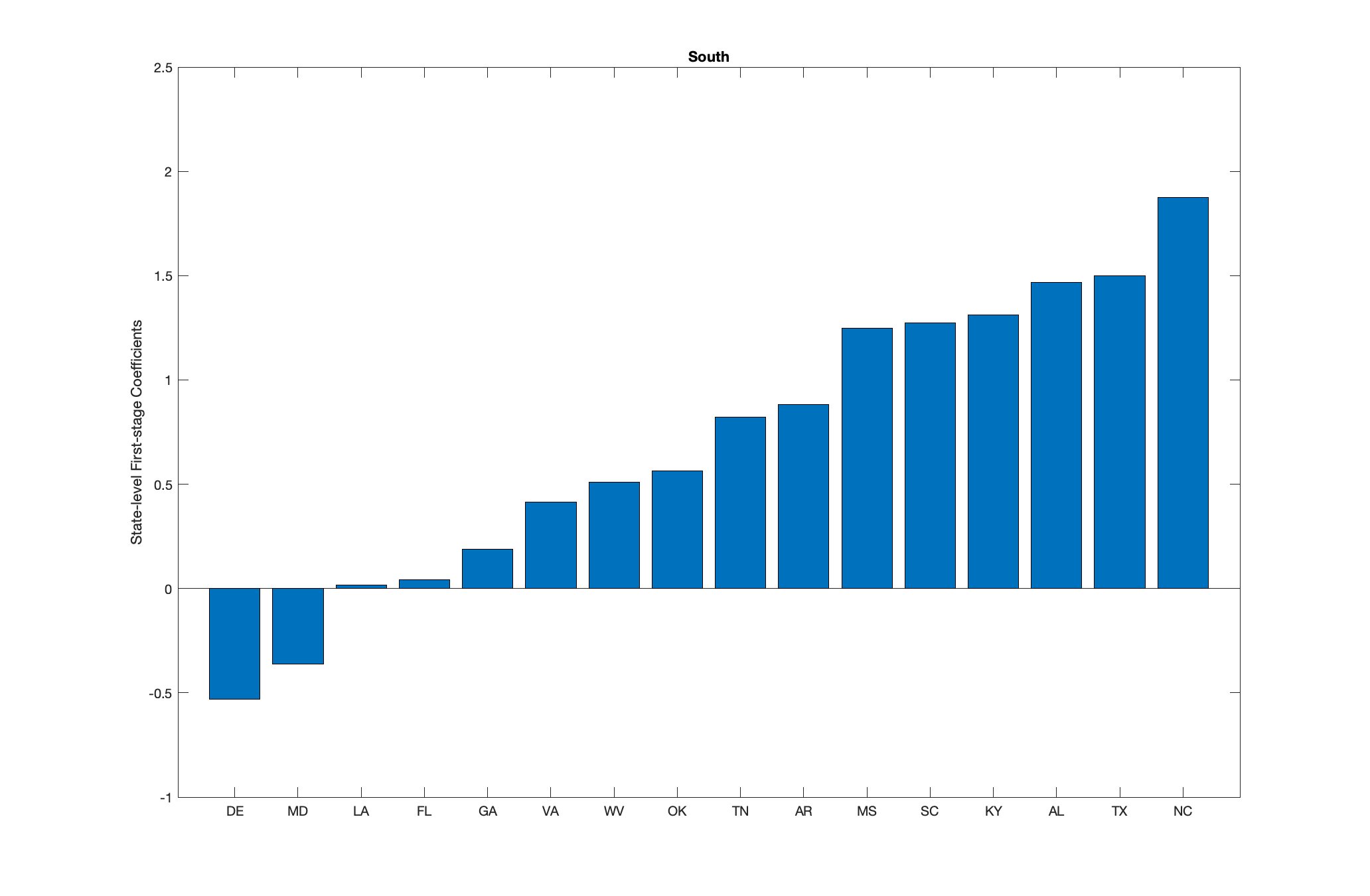}}
     \vspace*{-10mm}
    \caption{State-level First-stage Coefficients for the South Region in \cite{ADH2013}}
    \label{fig:first_stage}
    {\footnotesize{Note: The coefficients are obtained by running first-stage least squares regression at the state level and controlling for the same exogenous variables as those specified in \cite{ADH2013}. 
    }}   
\end{figure}


Motivated by these challenges, this paper investigates inference for IVQR with a small, fixed number of clusters and weak within-cluster error dependence, accommodating significant cluster-level heterogeneity in IV strength. We define clusters where the structural parameter (i.e., for the endogenous variable) is strongly identified as ``strong IV clusters.'' We propose a gradient wild bootstrap procedure for IVQR with clustered data. Our findings show that a bootstrap Wald test, whether studentized by the cluster-robust variance estimator (CRVE) or not, asymptotically controls size, given at least one strong IV cluster. The gradient wild bootstrap tests demonstrate power against local alternatives at the 10\% and 5\% significance levels with at least five and six strong IV clusters, respectively. Additionally, the CRVE-based bootstrap Wald test proves more powerful for distant local alternatives. We also develop a gradient wild bootstrap \citet[AR]{Anderson-Rubin(1949)} test for IVQR that controls size regardless of instrument strength. Compared to analytical methods using HAC estimators, all our bootstrap inference methods are agnostic to within-cluster dependence structure, avoiding complex covariance matrix estimation and making it applicable to datasets with various weak dependence structures, such as cross-sectional, serial, network, and spatial dependence.

The contributions in the present paper relate to several strands of literature. 
First, it is related to the literature on the cluster-robust inference.\footnote{See \cite{Cameron(2008)}, \cite{Conley2011inference}, \cite{Imbens-Kolesar(2016)}, \cite{abadie2022should}, \cite{Hagemann(2017), Hagemann2019placebo, Hagemann2020inference, Hagemann(2019)}, \cite{Mackinnon-Webb(2017)}, \cite{Djogbenou-Mackinnon-Nielsen(2019)}, \cite{Mackinnon-Nielsen-Webb(2019)}, \cite{Ferman2019inference}, \cite{Hansen-Lee(2019)}, \cite{Menzel2021bootstrap}, \cite{Mackinnon2021}, among others, and \cite{mackinnon2022cluster} for a recent survey.}
\cite{Hagemann(2017)}, \cite{Djogbenou-Mackinnon-Nielsen(2019)}, \cite{Mackinnon-Nielsen-Webb(2019)}, and \cite{Menzel2021bootstrap} show bootstrap validity under the asymptotic framework in which the number of clusters diverges to infinity.\footnote{We refer interested readers to \citet[Sections 4.1 and 4.2]{mackinnon2022cluster} for detailed discussions on this asymptotic framework and the alternative asymptotic framework that treats the number of clusters as fixed.} 
\citet[IM]{Ibragimov-Muller(2010), Ibragimov-Muller(2016)},
\citet[BCH]{Bester-Conley-Hansen(2011)},
\citet[CRS]{Canay-Romano-Shaikh(2017)}, \cite{Hagemann2019placebo, Hagemann2020inference, Hagemann(2019), hagemann2024}, and \cite{Hwang(2020)} consider an alternative asymptotic framework in which the number of clusters is treated as fixed, while the number of observations in each cluster is relatively large and the within-cluster dependence is sufficiently weak.
However, the inference methods proposed by BCH and \cite{Hwang(2020)} require an (asymptotically) equal cluster-level sample size,\footnote{See \citet[Assumptions 3 and 4]{Bester-Conley-Hansen(2011)} and \citet[Assumptions 4 and 5]{Hwang(2020)} for details.}
while those proposed by IM, CRS, and \cite{hagemann2024}
would require strong identification for all clusters in the IVQR context. In contrast, our gradient bootstrap Wald tests are more flexible as they do not require an equal cluster size. In addition, they only need one strong IV cluster for size control and five to six for local power, thus allowing for substantial cluster heterogeneity in identification strength for the IVQR model.  
To our knowledge, no alternative method proposed in the literature remains valid in such a context. Furthermore, we provide gradient bootstrap AR tests, which are fully robust to weak identification. 

Second, \cite{Canay-Santos-Shaikh(2020)} and \cite{Wang-Zhang(2024)} study wild bootstrap procedures with a few large clusters. In particular, \cite{Canay-Santos-Shaikh(2020)} first investigates the validity of wild bootstrap by innovatively connecting it with a randomization test with sign changes. 
Our results for IVQR generalize and complement those in \cite{Canay-Santos-Shaikh(2020)} and \cite{Wang-Zhang(2024)} in the following aspects. 
First, \cite{Canay-Santos-Shaikh(2020)} focus on the linear regression with exogenous regressors and then extend their analysis to 
a score bootstrap for the GMM estimator. 
\cite{Wang-Zhang(2024)} focus on the linear IV regression and show the validity of a modified version of the cluster wild restricted efficient (WRE) bootstrap procedure (e.g., \cite{Davidson-Mackinnon(2010)}, \cite{Finlay-Magnusson(2019)}, and \cite{Roodman-Nielsen-MacKinnon-Webb(2019)}, among others)  
in the case with few clusters.  
Instead, we propose gradient wild bootstrap procedures for IVQR inspired by \cite{Hagemann(2017)} and \cite{Jiang-Liu-Phillips-Zhang(2020)}, which avoid the estimation of the Hessian matrix that involves a nonparametric density component. In addition, we obtain the bootstrap estimator from the profiled optimization procedure for IVQR developed by \cite{Chernozhukov-Hansen(2004)}. These set us apart from the score bootstrap in the GMM setting. Second, we study the local power for our bootstrap Wald tests both with and without studentized by CRVE. 
In particular, the power analysis of the Wald test with CRVE is unconventional because, under a fixed number of clusters, the CRVE itself has a random limit. Specifically, we carefully design a gradient bootstrap counterpart for CRVE, which mimics well the original CRVE's randomness under the null and further diverges with the local alternative. The first property leads to the size control while the second leads to an interesting fact that the bootstrap Wald test with CRVE studentization is more powerful than that without in detecting sufficiently distant local alternatives. Such a power advantage is confirmed by our simulation experiments and empirical application. 
In addition, also different from its unstudentized counterpart, the local power of the Wald test with CRVE is established without the assumption that the IVQR first-stage coefficients have the same sign across all clusters, which may not hold in some empirical studies (e.g., see Figure \ref{fig:first_stage}).

Third, our paper is related to the literature on QR and IVQR. See, for example, \cite{Chernozhukov-Hansen(2004), Chernozhukov-Hansen(2005), Chernozhukov-Hansen(2006), Chernozhukov-Hansen(2008a)}, \cite{Hagemann(2017)}, and \cite{KW18}. Furthermore,  \cite{chernozhukov2020} provides a comprehensive overview of IVQR. We differ from them by considering an alternative asymptotic setting with a fixed number of clusters. \cite{hagemann2024} recently proposed a randomization test procedure in the spirit of CRS for inference on entire quantile and (exogenous) regression quantile processes under a small number of large clusters. 
However, as discussed above, a similar randomization test under the current IVQR setting would require strong identification for all clusters to establish validity. In contrast, our gradient wild bootstrap procedures for the Wald inference only need strong identification for at least one of the clusters.


Fourth, our paper is related to the literature on weak-identification-robust inference, in which various normal approximation-based inference approaches are available, among them \cite{Stock-Wright(2000)},
\cite{Kleibergen(2005)}, \cite{Andrews-Cheng(2012)},
\cite{Andrews(2016)}, \cite{Andrews-Mikusheva(2016)}, 
\cite{Moreira-Moreira(2019)}, and \cite{Andrews-Guggenberger(2019)}.
However, these robust inference methods cannot be directly applied in the current context with few clusters. 
On the other hand, it is found in the literature that when implemented appropriately, bootstrap approaches may substantially improve the inference for linear IV models, including the cases where IVs may be rather weak,\footnote{See, for example, 
\cite{Davidson-Mackinnon(2008), Davidson-Mackinnon(2010)}, \cite{Moreira-Porter-Suarez(2009)}, \cite{Wang-Kaffo(2016)},
\cite{Kaffo-Wang(2017)}, \cite{Wang-Doko(2018)}, \cite{Finlay-Magnusson(2019)}, \cite{young2022consistency}, and \cite{Wang-Zhang(2024)}, among others.
In addition, \cite{tuvaandorj2021robust} develops permutation versions of weak-IV-robust tests with (non-clustered) heteroskedastic errors.}
although the related literature for IVQR inference remains sparse. 
The bootstrap AR test developed in this paper is a bootstrap counterpart of the analytical AR test proposed by \cite{Chernozhukov-Hansen(2008a)}. We show that it controls asymptotic size for IVQR under both weak/non-identification and a small number of large clusters.

The remainder of this paper is organized as follows. 
Section \ref{sec:IVQR} presents the setup, the IVQR estimation, and our gradient wild bootstrap procedures. 
Section \ref{sec:main} presents assumptions and asymptotic results: Section \ref{subsce:assumptions} gives the main assumptions, Section \ref{subsce:examples} provides several specific examples related to our assumptions, Section \ref{sec:construction_IV} provides three different methods to construct the IVs, Section \ref{subsec:subvector-result} presents the asymptotic results for the Wald inference, while Section \ref{subsec:fullvector-result} presents those for the weak-identification-robust inference.
Simulations in Section \ref{sec: simu} suggest that our procedures have outstanding finite sample size control and, in line with our theoretical analysis, the bootstrap Wald test studentized by CRVE has power advantages compared with the other bootstrap tests. 
The empirical application with \cite{ADH2013}'s dataset is presented in Section \ref{sec: emp}. 

\textbf{Notation.} Throughout the paper, we write $0_{d_1 \times d_2}$, $\iota_d$, and $\mathbb{I}_{d}$ as  a $d_1 \times d_2$ matrix of zeros, $d$-dimensional vector of ones,  and a $d \times d$ identity matrix, respectively. For any positive integer $d$, we denote $[d] = (1,\cdots,d)$. We further denote $||\cdot||_2$ and $||\cdot||_F$ as the $\ell_2$ norm for a vector and the Frobenius norm for a matrix, respectively. 

\section{Setup, Estimation, and Inference Procedure}\label{sec:IVQR}
\subsection{Setup}
Throughout the paper, we observe clustered data where the clusters are indexed by $j \in [J]$ and units in the $j$-th cluster are indexed by $i \in I_{n,j} 
= \{1, ..., n_j \}$. For the $i$-th unit in the $j$-th cluster, we observe $y_{i,j} \in \textbf{R}$, $X_{i,j} \in \textbf{R}$, and $W_{i,j} \in \textbf{R}^{d_w}$ as an outcome of interest, a scalar endogenous regressor, and exogenous regressors, respectively. 
Furthermore, we let $Z_{i,j} \in \textbf{R}^{d_z}$  be the exogenous variables that are excluded from the outcome equation defined through conditional CDF: 
\begin{align}\label{eq:IVQR}
\mathbb{P}(y_{i,j} \leq X_{i,j} \beta_n(\tau) + W_{i,j}^\top\gamma_n(\tau)|W_{i,j},Z_{i,j})=\tau, \quad \tau \in \Upsilon,
\end{align}
where $\Upsilon$ is a compact subset of $(0,1)$. 

Throughout the paper, we focus on the setting with a single endogenous variable, as it is the most common case in empirical applications involving IVs. For instance, 101 out of 230 specifications in \citeauthor{Andrews-Stock-Sun(2019)}'s (\citeyear{Andrews-Stock-Sun(2019)}) sample and 1,087 out of 1,359 in \citeauthor{young2022consistency}'s (\citeyear{young2022consistency}) sample feature one endogenous regressor and one IV. Similarly, \cite{lee2021} find that 61 out of 123 IV papers published in \textit{AER} between 2013 and 2019 use single-IV regressions. While our setting also accommodates multiple IVs, most IVQR applications involve only one endogenous variable and one IV, as seen in studies by \cite{Chernozhukov-Hansen(2004), Chernozhukov-Hansen(2006), Chernozhukov-Hansen(2008a)}, and \cite{chernozhukov2013quantile}. In our empirical application, we revisit the influential study by \cite{ADH2013}, which also employs a single endogenous variable and one IV.

We allow the parameter of interest $\beta_n(\tau)$ (and the coefficient $\gamma_n(\tau)$ for the exogenous controls) to shift with respect to (w.r.t.) the sample size, which incorporates the analyses of size and local power in a concise manner: 
$\beta_n(\tau)  = \beta_0(\tau) + \mu_\beta(\tau)/r_n$, and $\gamma_n(\tau) =\gamma_0(\tau) + \mu_\gamma(\tau)/r_n$, where $\mu_{\beta}(\tau) \in \textbf{R}$ and $\mu_{\gamma}(\tau) \in \textbf{R}^{d_w}$ are the local parameters and $r_n$ is the convergence rate of the score defined later. Throughout the paper, for a generic function $g$ of data $D_{i,j} = (y_{i,j},X_{i,j},W_{i,j},Z_{i,j})$, we let $\mathbb{P}_ng(D_{i,j})=\frac{1}{n}\sum_{j \in [J]} \sum_{i \in I_{n,j}} g(D_{i,j})$, $\overline{\mathbb{P}}_ng(D_{i,j})=\frac{1}{n}\sum_{j \in [J]} \sum_{i \in I_{n,j}} \mathbb{E}g(D_{i,j})$, $\mathbb{P}_{n,j}g(D_{i,j})=\frac{1}{n_j}\sum_{i \in I_{n,j}} g(D_{i,j})$, and $\overline{\mathbb{P}}_{n,j}g(D_{i,j})=\frac{1}{n_j}\sum_{i \in I_{n,j}} \mathbb{E}g(D_{i,j})$. 

\subsection{Estimation}
Following \cite{Chernozhukov-Hansen(2006)}, we construct instrumental variables $\Phi_{i,j}(\tau) \in \textbf{R}^{d_\phi}$ from $(W_{i,j},Z_{i,j})$, that is, $\Phi_{i,j}(\tau) = \Phi(W_{i,j},Z_{i,j},\tau)$. For the validity of our bootstrap inference with a fixed number of clusters, we further require the instruments to be orthogonal to the control variables in the quantile regression context. Therefore, the function $\Phi(\cdot)$ may be unknown but can be estimated as $\hat{\Phi}(\cdot)$. The corresponding feasible IVs are defined as $\hat{\Phi}_{i,j}(\tau) = \hat{\Phi}(W_{i,j},Z_{i,j},\tau)$. We will provide more details about the construction of $\hat{\Phi}(\cdot)$ in Section \ref{sec:construction_IV}. Additionally, a scalar nonnegative weight is defined as $V_{i,j}(\tau)$, which also may be unknown, and its estimator is defined as $\hat{V}_{i,j}(\tau)$. Then, the estimation of $\beta_n(\tau)$ can be implemented via a profiled method described below. For a given value $b$ of $\beta_n(\tau)$, we first compute
\begin{align}
(\hat{\gamma}(b,\tau),\hat{\theta}(b,\tau)) = \arg \inf_{r,t}\sum_{j \in [J]}\sum_{i \in I_{n,j}}\rho_\tau(y_{i,j} - X_{i,j} b - W_{i,j}^\top r - \hat{\Phi}_{i,j}^\top(\tau) t) \hat{V}_{i,j}(\tau),
\label{eq:br1}
\end{align}
where $\rho_\tau(u) = u(\tau - 1\{u \leq 0\})$. 
Under appropriate conditions for strong identification, which will be made clear later, we can estimate $\beta_n(\tau)$ by $\hat{\beta}(\tau)$ defined as 
\begin{align}
\hat{\beta}(\tau) = \arg \inf_{b \in \mathcal{B}}||\hat{\theta}(b,\tau)||_{\hat{A}_1(\tau)},
\label{eq:br2}
\end{align}
where $ \mathcal{B}$ is a compact subset of $\textbf{R}$, $\hat{A}_1(\tau)$ is some $d_{\phi} \times d_{\phi}$ weighting matrix, and the notation $||u||_{A}$ for a compatible vector $u$ and matrix $A$ means $(u^\top A u)^{1/2}$. Last, we define $\hat{\gamma}(\tau) = \hat{\gamma}(\hat{\beta}(\tau),\tau)$ and $\hat{\theta}(\tau) = \hat{\theta}(\hat{\beta}(\tau),\tau)$. Following the lead of \cite{Chernozhukov-Hansen(2006)}, in practice, we suggest setting $\hat V_{i,j}(\tau) = 1$  and $\hat A_1(\tau) = \mathbb I_{d_\phi}$ or $\hat A_1(\tau) = \mathbb P_n \hat \Phi_{i,j}(\tau)\hat \Phi_{i,j}^\top(\tau)$. When the instrument is a scalar so that the model is just identified, the choice of $\hat A_1(\tau)$ becomes irrelevant. 
In addition, we note that the case of just identification is always achievable because even if the original IV $Z_{i,j}$ is multi-dimensional, we can construct $\hat \Phi_{i,j}$ as the prediction of $X_{i,j}$ using $Z_{i,j}$ and $W_{i,j}$ from a first-stage linear regression, which is again a scalar (see Section \ref{sec:construction_IV} for further details on the construction of $\hat \Phi_{i,j}(\tau)$). 

\subsection{Gradient Wild Bootstrap Inference}
\subsubsection{Inference Procedure for Wald Statistics}
In this section, we consider the null and local alternative hypotheses defined as
\begin{align*}
\mathcal{H}_0:  \beta_n(\tau)= \beta_0(\tau),~\forall~\tau \in \Upsilon \quad v.s. \quad \mathcal{H}_{1,n}: \beta_n(\tau) \neq  \beta_0(\tau),~\exists~\tau \in \Upsilon,
\end{align*}
which is equivalent to 
\begin{align}\label{eq: hypothesis}
\mathcal{H}_0:  \mu_\beta(\tau)= 0,~\forall~\tau \in \Upsilon \quad v.s. \quad \mathcal{H}_{1,n}: \mu_\beta(\tau) \neq 0,~\exists~\tau \in \Upsilon,
\end{align}
where $\Upsilon$ is a compact subset of $(0,1)$. 


Consider the test statistic with a normalization factor $\hat{A}_2(\tau)$, and let 
\begin{align}
T_{n} = \sup_{\tau \in \Upsilon}|| ( \hat{\beta}(\tau) - \beta_0(\tau))||_{\hat{A}_2(\tau)}
\label{eq:test_sta}
\end{align}
be the test statistic. In the following, we describe the gradient wild bootstrap procedure.
\begin{enumerate}
	\item[\textbf{Step 1}:] We define the null-restricted estimator 
	$\hat{\gamma}^r(\tau) = \hat{\gamma}(\beta_0(\tau),\tau)$. 
	\item[\textbf{Step 2}:] Let $\textbf{G} = \{-1,1 \}^J$ and for any $g = (g_1,\cdots,g_J) \in \textbf{G}$,
	\begin{align}
	(\hat{\gamma}_g^*(b,\tau),\hat{\theta}_g^*(b,\tau)) = & \arg \inf_{r,t} \biggl[\sum_{j \in [J]}\sum_{i \in I_{n,j}}\rho_\tau(y_{i,j} - X_{i,j} b - W_{i,j}^\top r - \hat{\Phi}_{i,j}^\top(\tau) t) \hat{V}_{i,j}(\tau) \notag \\
	- & \sum_{j \in [J]} g_j\sum_{i \in I_{n,j}} \hat{f}^\top_\tau(D_{i,j},\beta_0(\tau),\hat{\gamma}^r(\tau),0) \begin{pmatrix}
	r \\
	t
	\end{pmatrix}\biggr], \notag \\
	\hat{\beta}_g^*(\tau) = & \arg \inf_{b \in \mathcal{B}}\left[||\hat{\theta}_g^*(b,\tau)||_{\hat{A}_1(\tau)}\right], \quad \text{and} \quad  \hat{\gamma}_g^*(\tau) = \hat{\gamma}_g^*(\hat{\beta}_g^*(\tau),\tau),
	\label{eq:boot_beta}
	\end{align} 
	where the null-restricted estimator $\hat{\gamma}^r(\tau)$ is defined in the previous step,
	\begin{align}\label{eq:fhat}
	\hat{f}_{\tau}(D_{i,j},b,r,t) = (\tau - 1\{y_{i,j} - X_{i,j} b - W_{i,j}^\top r - \hat{\Phi}_{i,j}^\top(\tau) t \leq 0\})\hat{\Psi}_{i,j}(\tau) \hat V_{i,j}(\tau), 
	\end{align}
	and 	$\hat{\Psi}_{i,j}(\tau) = [W_{i,j}^\top,\hat{\Phi}_{i,j}^\top(\tau)]^\top$.

	\item[\textbf{Step 3}:]
	Let 
	$T_{n}^{*}(g) =  \sup_{\tau \in \Upsilon}||  (\hat{\beta}_g^*(\tau) - \hat{\beta}(\tau))||_{\hat{A}^*_{2,g}(\tau)},$ where $\hat{A}^*_{2,g}(\tau)$ is the bootstrap counterpart of the normalization factor $\hat{A}_2(\tau)$. Then, let $\hat{c}_{n}(1-\alpha)$ denote the $1-\alpha$ quantile of $\{ T_{n}^{*}(g) \}_{g \in \textbf{G}}$, 
	and we reject the null hypothesis if $T_{n} > \hat{c}_{n}(1-\alpha).$
	
\end{enumerate}

\begin{remark}
Several remarks regarding the choice of estimators in the above algorithm are in order. Specifically, in \textbf{Step 2}, we impose the null when implementing sign changes on the cluster-level scores by using $\beta_0(\tau)$ and $\hat{\gamma}^r(\tau)$ in $\hat{f}_{\tau}(D_{i,j},b,r,t)$.
    In contrast,  we use $\hat{\beta}(\tau)$, instead of $\beta_0(\tau)$, in \textbf{Step 3} to center the bootstrap IVQR estimator when constructing $T_n^*(g)$. Both choices are essential for the validity of our gradient bootstrap procedure under a small number of large clusters. In particular, we note that \cite{Canay-Santos-Shaikh(2020)} and \cite{Wang-Zhang(2024)} use null-restricted estimators to center their bootstrap estimators for linear (IV) regressions. Instead, we use $\hat{\beta}(\tau)$ in our \textbf{Step 3} because they use residual-based bootstrap procedures while we use the gradient bootstrap procedure.
\end{remark}

Below, we discuss two cases for the normalization factor $\hat{A}_2(\tau)$: (1) it has a deterministic limit, and (2) it involves the cluster-robust variance estimator (CRVE) for score from IVQR. For case (1), given that $\hat{A}_2(\tau)$ has a deterministic limit, we do not need to bootstrap it and just let $\hat{A}_{2,g}^*(\tau) = \hat{A}_2(\tau)$
in \textbf{Step 3}. By an abuse of notation, the corresponding test statistic, bootstrap statistics, and critical value are still denoted as $T_n$, $T_n^{*}(g)$, and $\hat{c}_n(1-\alpha)$, respectively. 

For case (2), we need some extra notation to define the normalization factor formally. Let $\hat G(\tau) \in \Re^{d_\phi}$ have a deterministic limit and 
\begin{align}
    & \hat \Omega(\tau,\tau') = \frac{1}{n}\sum_{j \in [J]} \omega \left[\sum_{i \in I_{n,j}}\hat f_\tau(D_{i,j}, \hat \beta(\tau), \hat \gamma(\tau), 0)\right] \left[\sum_{i \in I_{n,j}}\hat f_{\tau'}(D_{i,j}, \hat \beta(\tau'), \hat \gamma(\tau'), 0)\right]^\top \omega^\top,
\label{eq:Vhat}
\end{align}
where $\omega = (0_{d_\phi \times d_w}, \mathbb{I}_{d_\phi})$. 
Then, we define the normalization factor $\hat{A}_{CR}(\tau)$  in case (2) as 
\begin{align}
\hat{A}_{CR}(\tau) = \left[\hat G^\top(\tau)\hat \Omega(\tau,\tau)\hat G(\tau)\right]^{-1},
\label{eq:A1random}
\end{align}
and the corresponding CRVE-weighted Wald test statistic is defined as 
$$T_{CR,n} = \sup_{\tau \in \Upsilon}|| \hat{\beta}(\tau) - \beta_0(\tau)||_{\hat{A}_{CR}(\tau)}.$$

The normalization factor $\hat A_{CR}(\tau)$ and the test $T_{CR,n}$ take a cluster-robust form because the form of $\hat \Omega(\tau,\tau)$ preserves all within-cluster dependence. Naturally, for the choice of $\hat G(\tau)$, we would like to use a consistent estimator of the Jacobian $\Gamma(\tau)$ defined in \eqref{eq:Gamma} below. However, this would involve a nonparametric conditional density estimation and parameter tuning. Instead, we suggest using 
\begin{align*}
    \hat G^\top(\tau) & = \left[\hat {\mathcal E}_{X,\Phi}(\tau) \hat {\mathcal E}_{\Phi,\Phi}^{-1}(\tau)  \hat A_1(\tau) \hat {\mathcal E}_{\Phi,\Phi}^{-1}(\tau)  \hat {\mathcal E}_{X,\Phi}^\top(\tau)\right]^{-1} \left[\hat {\mathcal E}_{X,\Phi}(\tau) \hat {\mathcal E}_{\Phi,\Phi}^{-1}(\tau)  \hat A_1(\tau) \hat {\mathcal E}_{\Phi,\Phi}^{-1}(\tau)\right],
\end{align*}
where $\hat {\mathcal E}_{X,\Phi}(\tau) = \mathbb P_n X_{i,j}\hat \Phi_{i,j}^\top (\tau) \hat V_{i,j}(\tau)$ and $\hat {\mathcal E}_{\Phi,\Phi}(\tau) = \mathbb P_n \hat \Phi_{i,j}(\tau) \hat \Phi_{i,j}^\top (\tau) \hat V_{i,j}(\tau)$. When $\hat A_1(\tau)$ is set as $\hat {\mathcal E}_{\Phi,\Phi}(\tau)$, we can further simplify $\hat G(\tau)$ as 
\begin{align*}
 \hat G^\top(\tau)  = \left[\hat {\mathcal E}_{X,\Phi}(\tau)  \hat {\mathcal E}_{\Phi,\Phi}^{-1}(\tau)  \hat {\mathcal E}_{X,\Phi}^\top(\tau)\right]^{-1} \left[\hat {\mathcal E}_{X,\Phi}(\tau) \hat {\mathcal E}_{\Phi,\Phi}^{-1}(\tau)\right]. 
\end{align*}

In addition, even if we use a consistent estimator of the Jacobian $\Gamma(\tau)$ as $\hat G(\tau)$, it will not guarantee the consistency of the CRVE because the number of clusters is fixed in our setting. In fact,  $\hat \Omega(\tau,\tau)$, and thus, $\hat A_{CR}(\tau)$ have random limits after a proper normalization. Therefore, in case (2), to construct a valid critical value for $T_{CR,n}$, we also need to bootstrap $\hat A_{CR}(\tau)$ properly to mimic well this randomness in our bootstrap samples.

To define an appropriate gradient bootstrap analogue of $\hat A_{CR}(\tau)$, we let 
\begin{align}
& \hat f_{\tau,g}^*(D_{i,j}) = g_j\hat{f}_\tau(D_{i,j},\beta_0(\tau),\hat{\gamma}^r(\tau),0) + \hat{f}_\tau(D_{i,j},\hat{\beta}_g^*(\tau),\hat{\gamma}_g^*(\tau),0) - \hat{f}_\tau(D_{i,j},\hat{\beta}(\tau),\hat{\gamma}(\tau),0), \notag \\
& \hat{\Omega}_g^*(\tau,\tau') = \frac{1}{n}\sum_{j \in [J]} \omega \left\{\sum_{i \in I_{n,j}}\hat f_{\tau,g}^*(D_{i,j})\right\} \left\{\sum_{i \in I_{n,j}}\hat f_{\tau',g}^*(D_{i,j})\right\} ^\top \omega^\top.
\label{eq:Vhatg}
\end{align} 

Then, we let $\hat{A}_{2,g}^*(\tau)$ in \textbf{Step 3} of the previous bootstrap algorithm equal $\hat{A}_{CR,g}^*(\tau)$, which is defined as 
\begin{align*}
\hat{A}_{CR,g}^*(\tau) = \left[\hat G^\top (\tau)\hat{\Omega}_g^*(\tau,\tau)\hat G(\tau)\right]^{-1},
\end{align*}
where $\hat{\beta}_g^*(\tau) $ and $\hat{\gamma}_g^*(\tau) $ are defined in \eqref{eq:boot_beta}. 
The bootstrap counterpart of $T_{CR,n}$ and the critical value are defined as 
$$T_{CR,n}^{*}(g) = \sup_{\tau \in \Upsilon}||\hat{\beta}_g^*(\tau) -\hat{\beta}(\tau)||_{\hat{A}_{CR,g}^*(\tau)} \quad \text{and} \quad \hat{c}_{CR,n}(1-\alpha), \quad \text{respectively}.$$ 

\begin{remark}\label{rem: boot-CRVE}
Several remarks are in order regarding our design of $\hat{\Omega}_g^*(\tau,\tau')$ in (\ref{eq:Vhatg}). First, 
    when $\mathcal{H}_0$ is true, $\hat{\Omega}_g^*(\tau,\tau')$ has the same limit distribution as $\hat{\Omega}(\tau,\tau')$, which is needed for the asymptotic validity of the bootstrap Wald test with CRVE. Second, we design $\hat{\Omega}_g^*(\tau,\tau')$ in such a way so that under $\mathcal{H}_{1,n}$, the local parameter $\mu_{\beta}(\tau)$ defined in (\ref{eq: hypothesis}) will enter the limit distribution of $\hat{\Omega}_g^*(\tau,\tau')$ in sufficiently many randomization draws.
By contrast, $\mu_{\beta}(\tau)$ does not enter the limit distribution of $\hat{\Omega}(\tau,\tau')$ in the original Wald statistic $T_{CR,n}$.  
This leads to a further power improvement for our bootstrap test studentized by CRVE (more details are given in Theorem \ref{thm:power_comparison} and Remark \ref{rem: power_comparison} below).
\end{remark}

\begin{remark}
    We note that if we are in case (1) and $\Upsilon$ defined in (\ref{eq: hypothesis}) is a singleton, then $\hat A_2(\tau)$ (which has a deterministic limit) shows up in both the test statistic $T_n$ and its bootstrap counterpart $T_n^*(g)$, and thus, the critical value $\hat c_n(1-\alpha)$ so that $\hat A_2(\tau)$ gets canceled out. In this case, our bootstrap test is numerically invariant to the choice of $\hat A_2(\tau)$.
In addition, if we use CRVE to studentize the test statistic (i.e., in case (2)), $\Upsilon$ is a singleton, and the instrument $\hat \Phi_{i,j}(\tau)$ is a scalar, then $\hat G(\tau)$ is also a scalar, which shows up in both the test statistic $T_{CR,n}$ and the critical value $\hat c_{CR,n}(1-\alpha)$, and thus, gets canceled out. Therefore, in this scenario, our bootstrap test with the CRVE-weighted Wald statistic $T_{CR,n}$
is numerically invariant to the choice of $\hat G(\tau)$. 
Such an invariance property is one of the advantages of using the bootstrap tests.
\end{remark}


\subsubsection{Inference Procedure for Weak-instrument-robust Statistics}
This section considers the weak-instrument-robust inference for $\beta_n(\tau)$ when it may be weakly or partially identified. Recall $\beta_n(\tau) = \beta_0(\tau) + \mu_\beta(\tau)/r_n$. Under the null, we have $\mu_\beta(\tau) = 0$, or equivalently, $\beta_n(\tau) = \beta_0(\tau)$. Our test statistic follows the construction by \cite{Chernozhukov-Hansen(2008a)}. Specifically, we let 
\begin{align*}
AR_{n} = \sup_{\tau \in \Upsilon}||\hat{\theta}(\beta_0(\tau),\tau)||_{\hat{A}_3(\tau)},
\end{align*} 
where $\hat{\theta}(b,\tau)$ is defined in \eqref{eq:br1} and $\hat{A}_3(\tau)$ is a $d_\phi \times d_\phi$ weighting matrix, which will be specified later. We differentiate the weighting matrix used here  (denoted as $\hat A_3(\tau)$) with that used for the estimation of $\beta_n(\tau)$ in \eqref{eq:br2} (denoted as $\hat A_1(\tau)$) because of their different usages: the former is for the construction of the weak-instrument-robust test statistic while the later is for the estimation under strong identification. Theoretically, we require $\hat A_3(\tau)$ to have a deterministic limit, same as $\hat A_1(\tau)$.

Next, the bootstrap procedure for the weak-instrument-robust inference is defined as follows. 
\begin{enumerate}
\item[\textbf{Step 1}:] Recall $\textbf{G} = \{-1,1 \}^J$ and for any $g = (g_1,\cdots,g_J) \in \textbf{G}$, the null-imposed bootstrap estimators $(\hat \gamma_g^{*r}(\tau), \hat \theta_g^{*r}(\tau))=
(\hat \gamma_g^*(\beta_0(\tau), \tau),\hat \theta_g^*(\beta_0(\tau), \tau))$ for $(\gamma, \theta)$ are defined in \eqref{eq:boot_beta}.
\item[\textbf{Step 2}:] 
The bootstrap test statistic is then defined as $AR_{n}^{*}(g) = \sup_{\tau \in \Upsilon}||\hat \theta_g^{*r}(\tau) - \hat{\theta}(\beta_0(\tau),\tau)||_{\hat{A}_3(\tau)}$. 
\item[\textbf{Step 3}:] 
Let $\hat{c}_{AR,n}(1-\alpha)$ denote the $1-\alpha$ quantile of $\{ AR_{n}^{*}(g) \}_{g \in \textbf{G}}$,
and we reject the null hypothesis when $AR_{n} > \hat{c}_{AR,n}(1-\alpha)$. 
\end{enumerate}

It is also possible to studentize $\hat{\theta}(\beta_0(\tau),\tau)$ by a (null-imposed) CRVE, i.e., replace $\hat{A}_3(\tau)$ with $\tilde {A}_{CR}(\tau)$, where 
\begin{align*}
    \tilde {A}_{CR}(\tau) =  \left[\hat H (\tau)\tilde \Omega(\tau,\tau) \hat H(\tau)\right]^{-1},
\end{align*}
$\hat H(\tau) \in \Re^{d_\phi \times d_\phi}$ is some symmetric matrix, and 
\begin{align*}
    & \tilde \Omega(\tau,\tau') = \frac{1}{n}\sum_{j \in [J]} \omega \left[\sum_{i \in I_{n,j}}\hat f_\tau(D_{i,j}, \beta_0(\tau), \hat \gamma^r(\tau), 0)\right] \left[\sum_{i \in I_{n,j}}\hat f_{\tau'}(D_{i,j}, \beta_0(\tau'), \hat \gamma^r(\tau'), 0) \right]^\top \omega^\top.
\end{align*}
Note that different from $\hat{\Omega}(\tau, \tau')$ of the Wald test statistic in the previous section, we impose the null hypothesis in the form of $\tilde \Omega(\tau,\tau')$.
This is essential for the validity of the test under weak/non-identification for all clusters. Additionally, note that $ \tilde {A}_{CR}(\tau)$ admits a random limit in our setting, which distinguishes it from $\hat A_3(\tau)$ used for $AR_n$. 

We define the corresponding test statistic, its bootstrap counterparts, and critical value as 
\begin{align*}
AR_{CR,n} = \sup_{\tau \in \Upsilon}||\hat{\theta}(\beta_0(\tau),\tau)||_{\tilde {A}_{CR}(\tau)},
\quad
AR_{CR,n}^{*}(g) = \sup_{\tau \in \Upsilon}||\hat{\theta}_g^{*r}(\tau) - \hat{\theta}(\beta_0(\tau),\tau)||_{\tilde {A}_{CR}(\tau)}, 
\end{align*}
and $\hat{c}_{AR,CR,n}(1-\alpha)$, respectively. Naturally, for $\hat H(\tau)$, we want to use a consistent estimator for the Jacobian of $\hat \theta(\beta_0(\tau),\tau)$, which would again involve a kernel density estimation and parameter tuning. Therefore, similar to the bootstrap Wald inference, we instead suggest setting
\begin{align*}
    \hat H(\tau) =\mathbb P_n \hat \Phi_{i,j}(\tau)\hat \Phi_{i,j}^\top(\tau) \hat V_{i,j}(\tau). 
\end{align*}

Then, we reject the null hypothesis when $AR_{CR,n} > \hat{c}_{AR,CR,n}(1-\alpha)$. 
Note that different from the Wald statistic, we do not need to bootstrap $\tilde{A}_{CR}(\tau)$ here because $\tilde \Omega(\tau,\tau)$ is invariant to sign changes. Also, when the IV $\hat \Phi_{i,j}(\tau)$ is a scalar and $\Upsilon$ is a singleton, the choice of $\hat A_3(\tau)$ becomes irrelevant as it gets canceled in both the test statistic and the bootstrap critical value. Therefore, the bootstrap tests based on $AR_n$ and $AR_{CR,n}$ are numerically equivalent in this case. 

\subsection{Computation}
Given a value of $b$, we can compute $(\hat \gamma(b,\tau),\hat \theta(b,\tau))$ in \eqref{eq:br1} by the standard quantile regression algorithm. Then, we follow the lead of \cite{Chernozhukov-Hansen(2006)} and implement a one-dimensional grid search to compute $\hat \beta(\tau)$ in \eqref{eq:br2}. 

Furthermore, we note that given a value of $b$, the gradient bootstrap estimator $(\hat{\gamma}_g^*(b,\tau),\hat{\theta}_g^*(b,\tau))$ in \eqref{eq:boot_beta} can be formulated as linear programming and solved by well-developed linear optimization solvers. Specifically, we can stack up $y_{i,j} - X_{i,j}b$ and $\hat V_{i,j}$ first within each cluster and then across clusters $j=1,...,J$. Denote them as $\mathcal Y \in \Re^n$ and $\mathcal V \in \Re^n$, respectively. Similarly, we stack up $(W_{i,j}^\top, \hat \Phi_{i,j}^\top(\tau))$ together and construct a $n \times (d_w+d_\phi)$ matrix denoted as $\mathcal X$. Last, denote $\eta^\top = (r^\top,t^\top)$, and
\begin{align*}
    S = \sum_{j \in [J]} g_j\sum_{i \in I_{n,j}} \hat{f}_\tau (D_{i,j},\beta_0(\tau),\hat{\gamma}^r(\tau),0). 
\end{align*}
By letting $u_i = \max(0,\mathcal Y_i - \mathcal X_i^\top \eta)$ and $v_i = \max(0,-\mathcal Y_i + \mathcal X_i^\top \eta)$, we can rewrite \eqref{eq:boot_beta} as 
\begin{align*}
    \sum_{i \in [n]}\rho_\tau(\mathcal Y_i - \mathcal X_i^\top \eta)\mathcal V_i  - S^\top \eta  = \tau \mathcal V^\top u + (1-\tau) \mathcal V^\top v - S^\top \eta, 
\end{align*}
subject to 
\begin{align*}
    \mathcal Y - \mathcal X^\top \eta = u - v,
\end{align*}
where $u = (u_1,\cdots,u_n)^\top$ and $v = (v_1,\cdots,v_n)^\top$. Therefore, the gradient-based wild bootstrap estimator $(\hat{\gamma}_g^*(b,\tau),\hat{\theta}_g^*(b,\tau)) = \hat \eta$, where 
\begin{align*}
    (\hat \eta,\hat u,\hat v) & = \argmin_{\eta, u,v} \tau \mathcal V^\top u + (1-\tau) \mathcal V^\top v - S^\top \eta \\
    & s.t. \quad \mathcal Y - \mathcal X^\top \eta = u - v, \quad u \in \Re^n_+, \quad v\in \Re^n_+, \quad \text{and} \quad \eta \in \Re^{d_w + d_\phi}.
\end{align*}

\section{Assumptions and Asymptotic Results}\label{sec:main}

\subsection{Main Assumptions}\label{subsce:assumptions}

We make the following assumptions to establish the statistical properties of our bootstrap procedures formally. 

\begin{assumption}
	\begin{enumerate}[label=(\roman*)]
		\item Suppose $\mathbb{P}(y_{i,j} \leq X_{i,j} \beta_n(\tau) + W_{i,j}^\top\gamma_n(\tau)|W_{i,j},Z_{i,j})=\tau$ for $\tau \in \Upsilon$, $\beta_n(\tau)  = \beta_0(\tau) + \mu_{\beta}(\tau)/r_n$, and $\gamma_n(\tau) =\gamma_0(\tau) + \mu_\gamma(\tau)/r_n$.
		\item Suppose $\sup_{\tau \in \Upsilon}\left(||\mu_\gamma(\tau)||_2 + ||\mu_\beta(\tau)||_2 + ||\beta_0(\tau)||_2 + ||\gamma_0(\tau)||_2\right)\leq C<\infty$. 
		\item For all $\tau \in \Upsilon$, $\beta_n(\tau) \in \text{int}(\mathcal{B})$, where $\mathcal{B}$ is compact and convex. 
		\item Suppose $\max_{ i \in [n_j],j \in [J]}\sup_{y \in \textbf{R}}f_{y_{i,j}|W_{i,j},X_{i,j},Z_{i,j}}(y) <C$ for some constant $C \in (0,\infty)$, where $f_{y_{i,j}|W_{i,j},X_{i,j},Z_{i,j}}(\cdot)$ denotes the conditional density of $y_{i,j}$ given $W_{i,j},X_{i,j}$, and $Z_{i,j}$. 
		\item Denote the population counterpart of $\hat f_\tau(\cdot)$ as $f_\tau(\cdot)$, which is defined as 
 \begin{align}\label{eq:f}
     	f_{\tau}(D_{i,j},b,r,t) = (\tau - 1\{y_{i,j} - X_{i,j} b - W_{i,j}^\top r - \Phi_{i,j}^\top(\tau) t \leq 0\})\Psi_{i,j}(\tau) V_{i,j}(\tau),
 \end{align} 
 where $\Psi_{i,j}(\tau) = [W_{i,j}^\top,\Phi_{i,j}^\top(\tau)]^\top$.   Further define $\Pi(b,r,t,\tau) = \overline{\mathbb{P}}_n f_{\tau}(D_{i,j},b,r,t)$.  Then, there are compact subsets $\mathcal{R}$ and $\Theta$ of $\textbf{R}^{d_w}$ and $\textbf{R}^{d_\phi}$, respectively, such that Jacobian matrix  $\frac{\partial}{\partial(r^\top,t^\top)}\Pi(b,r,t,\tau)$ is continuous and has full column rank, uniformly in $n$ and over $\mathcal{B} \times \mathcal{R} \times \Theta \times \Upsilon$. 
		\item $\sup_{i \in [n_j], j \in [J],\tau \in \Upsilon}  \mathbb{E}||\Psi_{i,j}(\tau)||^{2+a}<\infty$ for some $a>0$. 
	\end{enumerate}
	\label{ass:id} 
\end{assumption}

\begin{remark}
Several remarks are in order.
First, Assumption \ref{ass:id} allows for the case in which $\beta_n(\tau)$ is partially or weakly identified as we do not require the Jacobian matrix w.r.t. $\beta,\gamma$ (i.e., $ \frac{\partial}{\partial(b^\top,r^\top)}\Pi(b,r,0,\tau)$) to be of full rank. Such a condition is assumed later in Assumption \ref{ass:id2} when we do need strong identification for the Wald inference, but is not required for the weak-identification-robust inference based on $AR_n$ and $AR_{CR,n}$. Second, under Assumption \ref{ass:id}, \cite{Chernozhukov-Hansen(2006)} show that $(\gamma_n^\top(\tau),0_{d_\phi \times 1}^\top)^\top$ is the unique solution to the weighted quantile regression of $y_{i,j} - X_{i,j}\beta_n(\tau)$ on $W_{i,j}$ and $\Phi_{i,j}(\tau)$ at the population level. Again, this condition does not impose strong identification of $\beta_n(\tau)$. 
\end{remark}

\begin{assumption}
	\begin{enumerate}[label=(\roman*)]
		\item Let 
		\begin{align*}
		\hat{\mathcal Q}_n(b,r,t,\tau) = & \mathbb{P}_n \rho_\tau(y_{i,j} - X_{i,j} b - W_{i,j}^\top r - \hat{\Phi}_{i,j}^\top(\tau) t) \hat{V}_{i,j}(\tau), \\
		\mathcal Q_n(b,r,t,\tau) = & \overline{\mathbb{P}}_n\rho_\tau(y_{i,j} - X_{i,j} b - W_{i,j}^\top r - \Phi_{i,j}^\top(\tau) t) V_{i,j}(\tau),
		\end{align*}
		and $\mathcal Q_\infty(b,r,t,\tau) = \lim_{n \rightarrow \infty}\mathcal Q_n(b,r,t,\tau)$. 	Suppose $(\gamma_n(b,\tau),\theta_n(b,\tau))$ and $(\gamma_\infty(b,\tau),\theta_\infty(b,\tau))$ are the unique minimizers of $\mathcal Q_n(b,r,t,\tau)$ and $\mathcal Q_\infty(b,r,t,\tau)$ w.r.t. $(r,t)$, respectively. In addition, suppose $(\gamma_n(b,\tau),\theta_n(b,\tau),\gamma_\infty(b,\tau),\theta_\infty(b,\tau))$ are continuous in $b \in \mathcal{B}$ uniformly over $\tau \in \Upsilon$,  $(\gamma_n(b,\tau),\theta_n(b,\tau)) \in \text{int}(\mathcal{R} \times \Theta)$ for all $(b,\tau) \in \mathcal{B} \times \Upsilon$, where $\mathcal{R}$ and $\Theta$  are defined in Assumption \ref{ass:id}. Also, suppose
		\begin{align*}
&		\sup_{(b,\tau) \in \mathcal{B} \times \Upsilon}|\mathcal Q_\infty(b,r,t,\tau) -\mathcal Q_n(b,r,t,\tau)| = o(1) \quad \text{and}\\ 
&		\sup_{(b,\tau) \in \mathcal{B} \times \Upsilon}|\hat{\mathcal Q}_n(b,r,t,\tau) -\mathcal Q_n(b,r,t,\tau)| = o_p(1). 
		\end{align*}
		
		\item  For any $\eps>0$,
		\begin{align*}
		\lim_{\delta \rightarrow 0}\limsup_{n \rightarrow \infty} \mathbb{P}\begin{pmatrix}
		\sup \biggl\Vert r_n(\mathbb{P}_{n,j}-\overline{\mathbb{P}}_{n,j})\biggl(\hat{f}_{\tau}(D_{i,j},\beta_n(\tau)+v_b,\gamma_n(\tau)+v_r,v_t) \\
		- f_{\tau}(D_{i,j},\beta_n(\tau),\gamma_n(\tau),0)  \biggr)\biggr\Vert_2\geq \eps
		\end{pmatrix} = 0,
		\end{align*}
		where the supremum inside the probability is taken over $\{j \in [J], ||v||_2 \leq \delta,\tau \in \Upsilon\}$ and
		$v = (v_b^\top,v_r^\top,v_t^\top)^\top$.\footnote{ For any function $g(\cdot)$ and its estimator $\hat g(\cdot)$,  $\mathbb{E}\hat{g}(W_{i,j})$ is interpreted as $\mathbb{E}g(W_{i,j})|_{g= \hat{g}}$ following the convention in the empirical processes literature. } 
		\item Denote $\eps_{i,j}(\tau) = y_{i,j} - X_{i,j} \beta_n(\tau) - W_{i,j}^\top\gamma_n(\tau)$, $\pi = (\gamma^\top,\theta^\top)^\top$, and $\delta_{i,j}(v,\tau) = X_{i,j} v_b + W_{i,j}^\top v_r + \hat{\Phi}_{i,j}^\top(\tau) v_t$. 		Then, for any $\eps>0$, we have
		\begin{align*}
		& \lim_{\delta \rightarrow 0}\limsup_{n \rightarrow \infty}\mathbb{P}\left[\sup\left\Vert \overline{\mathbb{P}}_{n,j} f_{\eps_{i,j}(\tau)}(\delta_{i,j}(v,\tau)|W_{i,j},Z_{i,j})\hat{\Psi}_{i,j}(\tau)\hat{\Psi}_{i,j}^\top(\tau)\hat V_{i,j}(\tau) - Q_{\Psi,\Psi,j}(\tau)  \right\Vert_{op} \geq \eps \right] = 0, \\
		& 	\lim_{\delta \rightarrow 0}\limsup_{n \rightarrow \infty}\mathbb{P}\left[\sup\left\Vert \overline{\mathbb{P}}_{n,j}f_{\eps_{i,j}(\tau)}(\delta_{i,j}(v,\tau)|W_{i,j},Z_{i,j})\hat{\Psi}_{i,j}(\tau)X_{i,j}\hat V_{i,j}(\tau) - Q_{\Psi,X,j}(\tau)  \right\Vert_{op} \geq \eps \right] = 0,
		\end{align*}
		where the suprema inside the probability are taken over $\{j \in [J],||v||_2 \leq \delta,\tau \in \Upsilon\}$, $v = (v_b^\top,v_r^\top,v_t^\top)^\top$, 
		\begin{align*}
			& Q_{\Psi,X,j}(\tau) = \lim_{n\rightarrow \infty}\overline{\mathbb{P}}_{n,j} f_{\eps_{i,j}(\tau)}(0|W_{i,j},Z_{i,j})\Psi_{i,j}(\tau)X_{i,j}V_{i,j}(\tau), \quad \text{and}\\	
   & Q_{\Psi,\Psi,j}(\tau) = \lim_{n\rightarrow \infty}\overline{\mathbb{P}}_{n,j} f_{\eps_{i,j}(\tau)}(0|W_{i,j},Z_{i,j})\Psi_{i,j}(\tau)\Psi_{i,j}^\top(\tau)V_{i,j}(\tau).
		\end{align*}
		\item $\sup_{\tau \in \Upsilon} r_n ||\mathbb{P}_nf_{\tau}(D_{i,j},\beta_n(\tau),\gamma_n(\tau),0)||_2 = O_p(1)$ for some $r_n = O(\sqrt n)$. 
			\item Let $n_j$ be the sample size of the $j$-th cluster. Then, we treat the number of clusters $J$ as fixed and $ n_j/n\rightarrow \xi_j$ for $\xi_j >0$ and $j \in [J]$.
		\item We further write $Q_{\Psi,\Psi,j}(\tau)$ as $ \begin{pmatrix}
	Q_{W,W,j}(\tau) & Q_{W,\Phi,j}(\tau) \\
	Q_{W,\Phi,j}^\top(\tau) & Q_{\Phi,\Phi,j}(\tau)
	\end{pmatrix}$,
	where $	Q_{W,W,j}(\tau)$, $	Q_{W,\Phi,j}(\tau)$, and $	Q_{\Phi,\Phi,j}(\tau)$ are $d_w \times d_w$, $d_w \times d_\phi$, and $d_\phi \times d_\phi$ matrices. Then, there exist constants $(c,C)$ such that $$0<c<\inf_{\tau \in \Upsilon}\lambda_{\min}\left(\sum_{j \in [J]} \xi_j Q_{\Psi,\Psi,j}(\tau)\right)\leq \sup_{\tau \in \Upsilon}\lambda_{\max}\left(\sum_{j \in [J]} \xi_j Q_{\Psi,\Psi,j}(\tau)\right)<C<\infty.$$ 
	\end{enumerate}
	\label{ass:asym}
\end{assumption}
\begin{remark}
First, Assumption \ref{ass:asym}(i) ensures $\gamma_n(b,\tau)$ and $\theta_n(b,\tau)$ are uniquely defined in the drifting-parameter setting. Second, Assumption \ref{ass:asym}(ii) is the stochastic equicontinuity of the empirical process 
\begin{align*}
r_n(\mathbb{P}_{n,j}-\overline{\mathbb{P}}_{n,j})\biggl(\hat{f}_{\tau}(D_{i,j},\beta_n(\tau)+v_b,\gamma_n(\tau)+v_r,v_t)- f_{\tau}(D_{i,j},\beta_n(\tau),\gamma_n(\tau),0)  \biggr)
\end{align*}
with respect to $v$. Such a condition is verified by \cite{Chernozhukov-Hansen(2006)} when the data are independent and $\hat{V}_{i,j}(\tau)$ and $\hat{\Phi}_{i,j}(\tau)$ uniformly converge to their population counterparts in probability. Their argument can be extended to data with various forms of weak dependence.  Third, Assumption \ref{ass:asym}(iii) requires the uniform consistency of the Jacobian matrices subject to infinitesimal perturbation of parameters, which holds even when observations are dependent. Fourth, Assumption \ref{ass:asym}(iv) requires the convergence rate of the sample mean of the score function to be $r_n$. We provide more details about $r_n$ in three examples in Section \ref{subsce:examples} below. Fifth, Assumption \ref{ass:asym}(v) implies that we focus on the case with a small number of large clusters.
\end{remark}

\begin{assumption}\label{ass:boot}
\begin{enumerate}[label=(\roman*)]
	\item For $j \in [J]$ and $\tau \in \Upsilon$, $Q_{W,\Phi,j}(\tau) = 0$.
	\item There exist versions of tight Gaussian processes  $\{\mathcal{Z}_j(\tau): \tau \in \Upsilon\}_{j \in [J]}$  such that $ \mathcal{Z}_j(\tau) \in \textbf{R}^{d_\phi}$, $\mathcal{Z}_j(\cdot)$ are independent across $j \in [J]$, $\mathbb{E}\mathcal{Z}_j(\tau)\mathcal{Z}^\top_j(\tau') = \Sigma_j(\tau,\tau')$, 
	$$0<c<\inf_{\tau \in \Upsilon, j \in [J]}\lambda_{\min}(\Sigma_j(\tau,\tau))\leq \sup_{\tau \in \Upsilon, j \in [J]}\lambda_{\max}(\Sigma_j(\tau,\tau)) \leq C<\infty$$ 
	for some constants $(c,C)$ independent of $n$, and 
	\begin{align}\label{eq:normal_approx}
	\sup_{j \in [J],\tau \in \Upsilon}||r_n\mathbb{P}_{n,j}\tilde{f}_{\tau}(D_{i,j},\beta_n(\tau),\gamma_n(\tau),0) - \mathcal{Z}_j(\tau)||_2 \convP 0, 
	\end{align}
	where $
	\tilde{f}_{\tau}(D_{i,j},\beta_n(\tau),\gamma_n(\tau),0) = (\tau - 1\{\eps_{i,j}(\tau)\leq 0\})\Phi_{i,j}(\tau)V_{i,j}(\tau).
	$
\end{enumerate}
\end{assumption}

\begin{remark}\label{rem:10}
Assumption \ref{ass:boot}(i) introduces Neyman orthogonality between the estimators of the coefficients of the endogenous and control variables, which is the key to connecting the gradient bootstrap with the randomization test with sign changes. This is in line with the results in \cite{Canay-Santos-Shaikh(2020)} and \cite{Wang-Zhang(2024)} for their residual-based bootstrap procedures. In Section \ref{sec:construction_IV} below, we propose both parametric and nonparametric approaches to construct IVs that satisfy Assumption \ref{ass:boot}(i) (under some regulatory conditions). 
\end{remark}

\begin{remark}
Consistently estimating $\Sigma_j(\cdot)$ in Assumption \ref{ass:boot}(ii) requires further assumptions on the within-cluster dependence structure and potential tuning parameters; see, for example, \cite{YG20} and \cite{GY23}. Instead, the key benefit of our bootstrap inference approach
is that it is fully agnostic about the expression of the covariance matrices.
\end{remark}

\begin{remark}
We notice that \eqref{eq:normal_approx} holds if the within-cluster dependence is sufficiently weak for some type of CLT to hold. Below, we provide three examples of data structures (Examples \ref{ex:panel}-\ref{ex:network}) that satisfy our requirements. 
\end{remark}

\subsection{Examples}\label{subsce:examples}
This section considers several examples and discusses why our assumptions are satisfied or violated in different scenarios. 

\begin{example} [Serial Dependence] \label{ex:panel}
We use $i$ and $j$ to index time period and clusters, respectively, so that observations have serial dependence over time and are asymptotically independent across clusters. Such settings were considered in BCH (Lemma 1 and Section 4.1), IM (Section 3.1), and CRS (Section S.1) for time series data and IM (Section 3.2) for panel data,\footnote{Specifically, for time series data, they propose to divide the full sample into $J$ (approximately) equal sized consecutive blocks (clusters). 
For panel data, assuming independence across individuals, one may treat the observations for each individual as a cluster.} among others. In this setup, we can verify \eqref{eq:normal_approx} under different levels of serial dependence.
\begin{enumerate}
   \item ($L_q$-Mixingale) Let $\tilde{f}_{\tau}^{(k)}(D_{i,j},\beta_n(\tau),\gamma_n(\tau),0)$ denote the $k$-th element of $\tilde{f}_{\tau}(D_{i,j},\beta_n(\tau),\gamma_n(\tau),0)$. Suppose there exists a filtration $\mathcal{F}_{i,j}$ that satisfies the following conditions: for some $q\geq 3$ and any $l\geq 0$ and $j \in [J]$,
    \begin{align*}
& \left\Vert       \mathbb E (\tilde{f}_{\tau}^{(k)}(D_{i,j},\beta_n(\tau),\gamma_n(\tau),0) \mid \mathcal{F}_{i-l,j}) \right\Vert_q \leq c_{n_j,i}\psi_l, \\
& \left\Vert   \tilde{f}_{\tau}^{(k)}(D_{i,j},\beta_n(\tau),\gamma_n(\tau),0) -    \mathbb E (\tilde{f}_{\tau}^{(k)}(D_{i,j},\beta_n(\tau),\gamma_n(\tau),0) \mid \mathcal{F}_{i+l,j}) \right\Vert_q \leq c_{n_j,i}\psi_{l+1},
    \end{align*}
and $\max_{j \in [J]}\max_{i \in I_{n,j}}c_{n_j,i} = o(n^{1/2})$.     Then, \citet[Theorem 4]{LL20} implies \eqref{eq:normal_approx} holds with $r_n = \sqrt{n}$. In fact, they show that the partial sum process of 
$$\{\tilde{f}_{\tau}(D_{i,j},\beta_n(\tau),\gamma_n(\tau),0)\}_{i \in I_{n,j}}$$ 
can be approximated by a martingale, and thus, is called a mixingale. It forms a very general class of models, including martingale differences, linear processes, and various types of mixing and near-epoch dependence processes as special cases. 
\item (Long Memory) Suppose $\tilde{f}_{\tau}(D_{i,j},\beta_n(\tau),\gamma_n(\tau),0) = \sum_{l=0}^{\infty}\Theta_l a_{i-l,j}$, where the innovations $a_{i,j} = (a_{i,j}^{(1)},\cdots,a_{i,j}^{(d_w + d_\phi)})^\top$ are $(d_w + d_\phi)$-dimensional martingale difference with respect to a filtration $\mathcal{F}_{i,j}$ such that for $k \in [d_w + d_\phi]$, 
\begin{align*}
\max_{i,j,k}\mathbb E (|a_{i,j}^{(k)}|^{2+d}\mid \mathcal{F}_{i-1,j}) < \infty,~a.s. \quad \text{and} \quad \mathbb E (a_{i,j}a_{i,j}^\top \mid \mathcal{F}_{i-1,j}) = \Sigma_a,~ a.s.   
\end{align*}
The $(d_w + d_\phi) \times (d_w + d_\phi)$ matrix coefficient $\Psi_l$ can be approximated by 
\begin{align*}
    \Theta_l \sim \frac{l^{d-1}}{\Gamma(d)} \Pi, \quad \text{as} \quad l \rightarrow \infty,
\end{align*}
where $\Gamma(\cdot)$ is the gamma function, $\Pi$ is a non-singular $(d_w + d_\phi) \times (d_w + d_\phi)$ matrix of constants that are independent of $l$, and $d \in (0,0.5)$ is the memory parameter. Then, \citet[Theorem 1]{C02} implies (\ref{eq:normal_approx}) holds with $r_n = n^{1/2-d}$. 
\end{enumerate}
 
\end{example}

\begin{example}[Spatial Dependence] \label{ex:spatial} This example is proposed by BCH. Suppose we have $n$ individuals indexed by $l$. The location of the $l$-th individual is denoted as $s_l$, an $m$-dimensional integer. The distance between individual $l_1$ and $l_2$ is measured by the maximum coordinatewise metric $\text{dist}(l_1,l_2)  = ||s_{l_1}-s_{l_2}||_\infty$. Observation $D$ is indexed by the location so that $D_l = D_{s_l}$ for $l \in [n]$. The clusters $I_{n,j}$ for $j \in [J]$ are defined as disjoint regions ($\Lambda_1,\cdots,\Lambda_J$). Let $\mathcal{F}_{\Lambda}$ be the $\sigma$-field generated by a given random field $D_s$, $s \in \Lambda$ with $\Lambda$ compact and let $|\Lambda|$ be the number of $s \in \Lambda$. Let $\Upsilon_{\Lambda_1,\Lambda_2}$ denote the minimum distance from an element of $\Lambda_1$ to an element of $\Lambda_2$ where the distance is measured by the maximum coordinatewise metric. The mixing coefficient is then 
\begin{align*}
    \alpha_{k_1,k_2}(l) & = \sup\{ \mathbb P(A \cap B) - \mathbb P(A)\mathbb P( B) \}, \\
    & s.t. \quad A \in \mathcal{F}_{\Lambda_1},~ B \in \mathcal{F}_{\Lambda_2},~ |\Lambda_1| \leq k_1,~ \quad |\Lambda_2| \leq k_2,~ \Upsilon(\Lambda_1,\Lambda_2) \geq l. 
\end{align*}
\cite{Bester-Conley-Hansen(2011)} assume the mixing coefficients satisfy (1) $\sum_{l = 1}^\infty l^{m-1} \alpha_{1,1}(l)^{\delta_1/(2+\delta_1)}<\infty$, (2) $\sum_{l = 1}^\infty l^{m-1} \alpha_{k_1,k_2}(l)<\infty$ for $k_1+k_2\leq 4$, and (3) $\alpha_{1,\infty}(l) = O(l^{-m-\delta_2})$ for some $\delta_1>0$ and $\delta_2>0$. 
Under this assumption and other regularity conditions in their Assumptions 1 and 2, \citet[Lemma 1]{Bester-Conley-Hansen(2011)} verifies \eqref{eq:normal_approx} with a finite number of clusters ($J$ fixed) and $r_n = \sqrt{n}$.
\end{example}


\begin{example}[Network Dependence]\label{ex:network} Suppose we observe $n$ units indexed by $\ell \in [n]$ and an adjacency matrix $\mathcal{A} = \{A_{\ell,\ell'}\}$, where $A_{\ell,\ell'} = 1$ means units $\ell$ and $\ell'$ are linked and $A_{\ell,\ell'} = 0$ means otherwise. We extend the linear-in-means social interaction model studied by \cite{BDF09} to the IVQR model. Specifically, we have 
\begin{align} \label{eq:linear-in-mean}
    y_\ell = \delta_0(U_\ell) + \beta(U_\ell) \frac{\sum_{\ell': A_{\ell,\ell'}=1}y_{\ell'}}{n_\ell} + \delta_1(U_\ell) B_\ell + \delta_2(U_\ell) \frac{\sum_{\ell': A_{\ell,\ell'}=1}B_{\ell'}}{n_\ell},
\end{align}
where $n_\ell = |\ell': A_{\ell,\ell'}=1|$ denotes the $\ell$-th node's number of friends and $B_\ell$ represents the $\ell$-th node's background characteristics, and we assume that $U_\ell$ is independent of $\{B_{\ell'}\}_{\ell' \in [n]}$ and follow the uniform distribution on $(0,1)$. In this setup, we have the endogenous variable $X_\ell = \frac{\sum_{\ell': A_{\ell,\ell'}=1}y_{\ell'}}{n_\ell}$, the control variables $W_\ell = (1, B_\ell, \frac{\sum_{\ell': A_{\ell,\ell'}=1}B_{\ell'}}{n_\ell})^\top$, and $\gamma(U_\ell) = (\delta_0(U_\ell), \delta_1(U_\ell), \delta_2(U_\ell))^\top$. Following the literature, we assume the adjacency matrix is independent of $\{B_\ell,\eps_\ell\}_{\ell \in [n]}$.
Further suppose $X_\ell \beta(u) + W_\ell^\top \gamma(u)$ is monotonically increasing in $u$, then we have
\begin{align*}
\mathbb P\left(y_\ell \leq  X_\ell \beta(\tau) + W_\ell^\top \gamma(\tau)| \{B_{\ell}\}_{\ell \in [n]}   \right) = \mathbb P\left(U_\ell \leq  \tau| \{B_{\ell}\}_{\ell \in [n]}   \right)    = \tau. 
\end{align*}
\cite{BDF09} showed that one can use $Z = \tilde A^2 B$ as the IV, where $\tilde A$ is the $n \times n$ normalized adjacency matrix with a typical entry $\tilde A_{\ell,\ell'} = A_{\ell,\ell'}/n_{\ell}$ and $B$ is a $n \times 1$ vector of $\{B_\ell\}_{\ell \in [n]}$.  Then, by the law of iterated expectation, we have
\begin{align*}
\mathbb P\left(y_\ell \leq  X_\ell \beta(\tau) + W_\ell^\top \gamma(\tau)| W_{\ell},Z_\ell \right)    = \tau 
\end{align*}
so that \eqref{eq:IVQR} holds. 

For inference, we can then follow \cite{L22} to partition the nodes (i.e., $\{I_{n,j}\}_{j \in [J]}$) in the network and construct clusters. Specifically, for a subset of indexes $S \subset [n]$, define the conductance of $S$ as $\phi_{\mathcal{A}}(S) = \frac{|\partial_{\mathcal{A}}(S)|}{vol_A(S)}$, where $|\partial_{\mathcal{A}}(S)| = \sum_{\ell \in S}\sum_{\ell' \in [n]/S}{\mathcal{A}}_{\ell,\ell'}$ is the number of links involving a unit in $S$ and a unit not in $S$ and $vol_{\mathcal{A}}(S) = \sum_{\ell \in S}\sum_{\ell' \in [n]}\mathcal{A}_{\ell,\ell'}$ is the sum of degrees $\sum_{\ell' \in [n]}{\mathcal{A}}_{\ell,\ell'}$ of units $\ell \in S$. Then, \cite{L22} shows \eqref{eq:normal_approx} holds with a finite number of clusters ($J$ fixed) and $r_n = \sqrt{n}$ when $\max_{j \in [J]}\phi_\mathcal{A}(I_{n,j}) (\frac{1}{n}\sum_{\ell \in [n]}\sum_{\ell' \in [n]}\mathcal{A}_{\ell,\ell'})\rightarrow 0$ as $n \rightarrow \infty$ and the observations exhibit weak network dependence in the sense of \citet[Assumption 5]{L22}.\footnote{To be more specific, \cite{L22} shows \eqref{eq:normal_approx} holds when $\Upsilon$ contains a finite and fixed number of quantile indexes.Extending his result to cover a continuum of quantile indexes is plausible but outside the scope of this paper. } As the partition $\{I_{n,j}\}_{j \in [J]}$ are unobserved, \cite{L22} further showed that it is possible to recover the clusters by spectral clustering, a method that clusters the leading $J$ eigenvectors of network graph Laplacian by the k-means algorithm. We provide more details about the spectral clustering in Section \ref{sec: simu}.   
\end{example}

\begin{example}[Factor Structure]
As mentioned in the Introduction, our asymptotic framework, which treats the number of clusters as fixed,  follows previous studies such as IM, BCH, and CRS. We emphasize that the main restriction of such an asymptotic framework is it requires the within-cluster dependence to be sufficiently weak for some CLT to hold within each cluster (as illustrated in Examples \ref{ex:panel}-\ref{ex:network}).
For example, as pointed out by \citet[Sections 3.1, 3.2, and 4.2]{mackinnon2022cluster}, this requirement rules out the case where the error follows a factor structure, i.e., for $\eps_{i,j}(\tau) = y_{i,j} - X_{i,j} \beta_n(\tau) - W_{i,j}^\top \gamma_n(\tau)$ in the current IVQR model, 
\begin{align}\label{eq: factor-structure}
    \eps_{i,j}(\tau)  = \lambda_{i,j}(\tau) f^s_{j}(\tau) + u_{i,j}(\tau),
\end{align}
where $u_{i,j}(\cdot)$ denotes the idiosyncratic error, $f^s_j(\cdot)$ denotes the cluster-wide shock, and $\lambda_{i,j}(\cdot)$ is the factor loading. By contrast, such a dependence structure can be handled under the asymptotic framework that lets the number of clusters $J$ diverge to infinity. However, we conjecture that our gradient wild bootstrap procedure is also valid under the alternative asymptotic framework with a large number of small clusters (more discussions are provided in Remark \ref{rem:Hagemann}).

\end{example}

\begin{example}[Cluster Fixed Effects]
If cluster fixed effects exist in the IVQR model, we can add cluster dummies into the control variables $W$. In our setting, the number of clusters is fixed so that even $W$ includes cluster dummies, it still has a fixed dimension, and all our assumptions can still hold. We also note that in the case with linear regressions, adding cluster dummies is equivalent to first projecting out the fixed effects so that $(y_{i,j}, X_{i,j}, W_{i,j}, Z_{i,j})$ is expressed as deviations from cluster means, which is also recommended by \citet[Section 2.2]{Djogbenou-Mackinnon-Nielsen(2019)} and \citet[Section 3.2]{mackinnon2022cluster}. However, we emphasize that in the current setting with quantile regressions, cluster-level demeaning and adding cluster dummies are not equivalent to each other.  
\end{example}

\begin{example}[Heterogeneous IV Strength Across Clusters]\label{ex:Hetero_IV_strength}
As mentioned in the Introduction, we allow for cluster-level heterogeneity with regard to IV strength. For instance, we can consider the following first-stage regression:
\begin{align*}
X_{i,j} =     Z_{i,j}^\top \Pi_{z,j,n} + W_{i,j}^\top \Pi_{w,j,n} + E_{i,j},
\end{align*}
Then, our model \eqref{eq:IVQR} allows for both $\Pi_{z,j,n}$ and $\Pi_{w,j,n}$ to vary across clusters. 
In this case, the Jacobian $Q_{\Phi, X,j}(\tau)$ for the $j$-th cluster takes the form of 
\begin{align*}
Q_{\Phi, X,j}(\tau) = \lim_{n \rightarrow \infty} \overline{\mathbb{P}}_{n,j}f_{\eps_{i,j}(\tau)}(0|W_{i,j},Z_{i,j})\Phi_{i,j}(\tau)Z_{i,j}^\top \Pi_{z,j,n}  V_{i,j}(\tau).
\end{align*}
In particular, our bootstrap Wald tests (i.e., $T_n$ and $T_{CR,n}$) are still valid even when $\Pi_{z,j,n}$, and thus, $Q_{\Phi, X,j}(\tau)$ decay to or are zero for some of the clusters. Our bootstrap AR tests (i.e., $AR_n$ and $AR_{CR,n}$) control asymptotic size even when $\Pi_{z,j,n}$ decay to or are zero for all clusters.  
\end{example}

\begin{example}[Heterogeneous Slope for the Endogenous Variable]\label{ex:Hetero_slop}
Similar to \citet[Example 2]{Canay-Santos-Shaikh(2020)}, we cannot allow for $\beta_n(\tau)$ in \eqref{eq:IVQR} to be heterogeneous across clusters, denoted as $(\beta_{n,j}(\tau))_{j \in [J]}$. More specifically, in this case, the IVQR estimator based on the full sample will estimate $\tilde \beta(\tau)$, a certain weighted average of $(\beta_{n,j}(\tau))_{j \in [J]}$. Then, in Assumption \ref{ass:boot}, we have
\begin{align*}
    \tilde{f}_{\tau}(D_{i,j},\beta_n(\tau),\gamma_n(\tau),0) = (\tau - 1\{X_{i,j}(\beta_j(\tau) - \tilde \beta(\tau))+\eps_{i,j}(\tau)\leq 0\})\Phi_{i,j}(\tau)V_{i,j}(\tau),
\end{align*}
so that Assumption \ref{ass:boot}(ii) is violated. 
\end{example}

\begin{example}[Cluster-level Endogenous Variable]\label{ex:cluster_IV}
If $X_{i,j}$ is a cluster-level variable (say, $X_j$), then the within-cluster limiting Jacobian $Q_{\Psi, X,j}(\tau)$ may be random and potentially correlated with the within-cluster score component $\mathcal{Z}_j$ 
(as $X_{j}$ is endogenous), which violates Assumption \ref{ass:asym}(iii). We notice that similar issues can arise with the approaches of BCH, IM, and CRS. On the other hand, our bootstrap weak-instrument-robust tests remain valid in this case as they do not depend on $Q_{\Psi, X,j}(\tau)$.
\end{example}

\subsection{Construction of Instruments}\label{sec:construction_IV}
In Assumption \ref{ass:boot}(i) above, we require the IVs to satisfy the following condition: for $j \in [J]$, 
\begin{align*}
    Q_{W,\Phi,j}(\tau) = \lim_{n\rightarrow \infty}\overline{\mathbb{P}}_{n,j} f_{\eps_{i,j}(\tau)}(0|W_{i,j},Z_{i,j})W_{i,j}\Phi_{i,j}^\top(\tau) V_{i,j}(\tau)= 0, 
\end{align*}
where $\eps_{i,j}(\tau) = y_{i,j} - X_{i,j}\beta_n(\tau) - W_{i,j}^\top \gamma_n(\tau)$ and $f_{\eps_{i,j}(\tau)}(0|W_{i,j},Z_{i,j})$ is the conditional PDF of $\eps_{i,j}(\tau)$ given $(W_{i,j},Z_{i,j})$ and evaluated at $0$. This section proposes three ways to construct IVs that satisfy this requirement. 

\begin{remark}[Parametric Approach]\label{rem:par} Suppose $Z_{i,j} = G(W_{i,j},\pi) + U_{i,j} \in \Re^{d_z}$ such that $\pi$ is a finite dimensional  parameter, $G(\cdot)$ is a known $d_z$-dimensional function (e.g., $G(w,\pi) = w^\top \pi$), and $U_{i,j}$ is a random shock such that $U_{i,j} \indep \eps_{i,j}| W_{i,j}$ and $\mathbb E(U_{i,j}|W_{i,j}) = 0$. Furthermore, let $\hat \lambda$ be the regression coefficient of $Z_{i,j}$ in the linear (first-stage) regression of $X_{i,j}$ on $Z_{i,j}$ and $W_{i,j}$ using the full sample and $\lambda$ be the probability limit of $\hat \lambda$. Then,  we can let $\Phi_{i,j} = \lambda^\top U_{i,j} \in \Re$ and its feasible version be 
\begin{align*}
    \hat \Phi_{i,j} = \hat \lambda^\top (Z_{i,j} - G(W_{i,j},\hat \pi)), 
\end{align*}
where $\hat \pi$ is a consistent estimator of $\pi$.\footnote{When $G(w,\pi) = w^\top \pi$, we can compute $\hat \pi$ as the coefficient of $W_{i,j}$ in the linear regression of $Z_{i,j}$ on $W_{i,j}$ using observations in the full sample.} We can see that, if $V_{i,j}(\tau) = 1$, then  
\begin{align*}
   \overline{\mathbb{P}}_{n,j} f_{\eps_{i,j}(\tau)}(0|W_{i,j}, Z_{i,j})W_{i,j} \Phi_{i,j}^\top & =    \overline{\mathbb{P}}_{n,j} f_{\eps_{i,j}(\tau)}(0|W_{i,j},Z_{i,j})W_{i,j} U_{i,j}^\top \lambda \\
   & = \overline{\mathbb{P}}_{n,j} f_{\eps_{i,j}(\tau)}(0|W_{i,j}) W_{i,j} \mathbb E(U_{i,j}^\top|W_{i,j}) \lambda = 0. 
\end{align*}    

The key benefit of this approach is that it does not require nonparametric estimation of the conditional density and, thus, the tuning parameters. The way to convert the potentially multi-dimensional IVs into one via the least squares estimator $\hat \lambda$ is also recommended by \cite{Chernozhukov-Hansen(2006)}. For example, in their application of IVQR to \cite{Angrist-Krueger(1991)}'s dataset for the study of returns to schooling, \cite{Chernozhukov-Hansen(2006)} used the linear projection of years of schooling, the endogenous variable, onto covariates and three quarter-of-birth dummy variables (i.e., three-dimensional IVs).
However, efficiency may instead be improved by choosing $\hat \Phi_{i,j}(\tau)$ and $\hat V_{i,j}(\tau)$ appropriately.
\end{remark}

\begin{remark}[Nonparametric Approach]\label{rem:nonpar}
To enforce Neyman-orthogonality between IVs and control variables, we follow \cite{CHW20} and partial out the effect of $W_{i,j}$ from $Z_{i,j}$.  Specifically, for the current nonparametric approach, we construct $\hat{\Phi}_{i,j}(\tau)$ as $\hat{\Phi}_{i,j}(\tau) =  \hat Z_{i,j}-  \hat{\chi}^\top(\tau)W_{i,j}$, where $\hat Z_{i,j} =  (Z_{i,j}^\top,W_{i,j}^\top) \hat \lambda$ and $\hat \lambda$ contains the regression coefficients of both $Z_{i,j}$ and $W_{i,j}$ in the linear (first-stage) regression of $X_{i,j}$ on $Z_{i,j}$ and $W_{i,j}$ using the full sample. To compute $\hat{\chi}(\tau)$, we first need to compute the residual $\hat{\underline{\eps}}_{i,j}(\tau)$. Specifically, let 
\begin{align}
\hat{\underline{\eps}}_{i,j}(\tau) = Y_{i,j} - X_{i,j} \beta_0(\tau) - W_{i,j}^\top \hat{\gamma}(\beta_0(\tau),\tau),
\label{eq:epsunderline_2}
\end{align}
where $\hat{\gamma}(\beta_0(\tau),\tau)$ is defined in \eqref{eq:br1} with $\hat \Phi_{i,j}(\tau) = \hat Z_{i,j}$ and $\beta_0(\tau)$ is the null hypothesis. Under the null and local alternative, we have $\sup_{\tau \in \Upsilon}||\hat{\gamma}(\beta_0(\tau),\tau) - \gamma_n(\tau)||_2 = O_p(r_n^{-1})$ so that $\hat{\underline{\eps}}_{i,j}(\tau)$ can approximate the true error $\eps_{i,j}(\tau)$ well. In addition, let $K(\cdot)$ be a symmetric kernel function and $(h_1,h_2)$ be bandwidths. 
Then, we compute $\hat \chi$ as 
\begin{align}
\hat{\chi}(\tau) = \underline{\hat{Q}}_{W,W}(\tau) \underline{\hat{Q}}^{-}_{W,W}(\tau) \underline{\hat{Q}}^{-}_{W,W}(\tau)\hat{\underline{Q}}_{W,Z}(\tau)
\label{eq:chi}
\end{align}
where 
\begin{align}
\underline{\hat{Q}}_{W,W}(\tau) = \mathbb{P}_{n} \left( \frac{1}{h_1}K\left(\frac{\hat{\underline{\eps}}_{i,j}(\tau)}{h_1}\right)V_{i,j}(\tau)W_{i,j}W_{i,j}^\top\right), \quad \text{and}
\label{eq:Jhat_rr}
\end{align}
\begin{align}
\hat{\underline{Q}}_{W,Z}(\tau) = \mathbb{P}_{n} \left( \frac{1}{h_2}K\left(\frac{\hat{\underline{\eps}}_{i,j}(\tau)}{h_2}\right)V_{i,j}(\tau)W_{i,j} \hat Z_{i,j} \right).
\label{eq:Jhatunder_rt}
\end{align}
For a symmetric and positive semidefinite matrix $A$, $A^-$ is its generalized inverse. 

We suggest using the uniform kernel $K(u) = 1\{|u| \leq 1\}/2 $. \cite{Kato12} has derived the rule-of-thumb bandwidths for both independent and weakly dependent data: 
\begin{align*}
h_1 & =  \hat{s} \left[\frac{4.5 \mathbb P_n \hat V_{i,j}(\tau) ||W_{i,j}||_2^4 }{q(\tau) \left\Vert \mathbb P_n \hat V_{i,j}(\tau) W_{i,j}W_{i,j}^\top \right\Vert_F^2} \right]^{1/5}n^{-1/5}, \\
h_2 & =  \hat{s} \left[\frac{4.5 \mathbb P_n \hat V_{i,j}(\tau) ||W_{i,j}||_2^2 \hat Z_{i,j}^2 }{q(\tau) \left\Vert \mathbb P_n \hat V_{i,j}(\tau) W_{i,j}\hat Z_{i,j} \right\Vert_F^2} \right]^{1/5}n^{-1/5},
\end{align*}
where $q(\tau) = (1-F_N^{-1}(\tau))^2 f_N(F_N^{-1}(\tau))$, $F_N(\cdot)$ and $f_N(\cdot)$ are the distribution and density functions of the standard normal distribution, respectively, and $\hat{s}$ is the sample standard error of $\{\underline{\hat{\eps}}_{i,j}\}_{i \in I_{n,j}},j\in [J]\}$. 

This approach is valid given that there exists $\chi_{n,j}(\tau)$ such that
\begin{align}\label{eq:chi_condition}
    Q_{W,Z,j}(\tau) = Q_{W,W,j}(\tau)\chi_{n,j}(\tau) \quad \text{and} \quad \overline{\mathbb{P}}_{n,j}||W_{i,j}^\top(\chi(\tau) - \chi_{n,j}(\tau))||_{op}^2 = o(1),
\end{align}
where $Q_{W,W,j}(\tau)$ is defined in Assumption \ref{ass:asym}(vi) and $Q_{W,Z,j}(\tau)$ is defined in the same manner, while $Q_{W,W}(\tau)$, $Q_{W,Z}(\tau)$, and $ \chi(\tau) = Q_{W,W}^{-1}(\tau)Q_{W,Z}(\tau)$ are the probability limits of $\hat{\underline{Q}}_{W,W}(\tau)$, $\hat{\underline{Q}}_{W,Z}(\tau)$, and  $\hat \chi(\tau)$, respectively. The requirement in \eqref{eq:chi_condition} is similar in spirit to \citet[Assumption 2(iv) in Section A]{Canay-Santos-Shaikh(2020)}. \cite{Canay-Santos-Shaikh(2020)} further pointed out that one sufficient but not necessary condition for \eqref{eq:chi_condition} is that the distributions of $(Z^{\top}_{i,j}, W^{\top}_{i,j})_{i \in I_{n,j}}$ are the same across clusters. Note this still allows for heterogeneous IV strength in the first stage. In Section \ref{sec:proj} of the Online Supplement, we further provide regularity conditions, which, along with \eqref{eq:chi_condition}, imply $\hat{\Phi}_{i,j}(\tau)$ satisfies Assumption \ref{ass:boot}(i).

Furthermore, we note that the concerns about the presence of tuning parameters are mitigated for two reasons. First, we do not suffer from the curse of dimensionality because $\underline{\hat{\eps}}_{i,j}$ is a scalar. Second, we aim to estimate consistently, rather than make inferences of,  $Q_{W,W}(\tau)$ and $Q_{W,Z}(\tau)$, and thus, other automatic bandwidths such as cross validation can be well integrated into our bootstrap method. We also emphasize that unlike the estimation of $\Sigma_j$ defined in Assumption \ref{ass:boot}(ii), the consistency of $\underline{\hat{Q}}_{W,W}(\tau)$ and $\underline{\hat{Q}}_{W,\Phi}(\tau)$ holds under general weak dependence of observations within clusters, and importantly, does not require us to specify this dependence structure. 
\end{remark}

\begin{remark}[Cluster-level Estimation]\label{rem:cluster}
In the two above examples, we estimate the parameters $(\pi,\chi(\tau))$ using all the observations. To allow for the case where the coefficients may be heterogeneous across clusters, we can estimate them at the cluster level instead. 

Specifically, for the parametric approach, if $G(w,\pi_j) = w^\top \pi_j$, then we can estimate $\pi_j$ by the OLS regression of $Z_{i,j}$ on $W_{i,j}$ using observations in the $j$-th cluster. 

For the nonparametric approach, we can let $\hat{\Phi}_{i,j}(\tau) = \hat Z_{i,j}-  \hat{\chi}_j^\top(\tau)W_{i,j}$, where $\hat \chi_j$ is defined as 
\begin{align}
\hat{\chi}_j(\tau) = \underline{\hat{Q}}_{W,W,j}(\tau) \underline{\hat{Q}}^{-}_{W,W,j}(\tau) \underline{\hat{Q}}^{-}_{W,W,j}(\tau) \hat{\underline{Q}}_{W,Z,j}(\tau),
\label{eq:chi_j}
\end{align}
where 
\begin{align}
\underline{\hat{Q}}_{W,W,j}(\tau) = \mathbb{P}_{n,j} \left( \frac{1}{h_{3,j}}K\left(\frac{\hat{\underline{\eps}}_{i,j}(\tau)}{h_{3,j}}\right)V_{i,j}(\tau)W_{i,j}W_{i,j}^\top\right), \quad \text{and}
\label{eq:Jhat_rrj}
\end{align}
\begin{align}
\hat{\underline{Q}}_{W,Z,j}(\tau) = \mathbb{P}_{n,j} \left( \frac{1}{h_{4,j}}K\left(\frac{\hat{\underline{\eps}}_{i,j}(\tau)}{h_{4,j}}\right)V_{i,j}(\tau)W_{i,j} \hat Z_{i,j}^\top \right).
\label{eq:Jhatunder_rtj}
\end{align}
This definition allows for $\underline{\hat{Q}}_{W,W,j}(\tau)$ to be non-invertible for some but not all clusters. Then, following \cite{Kato12}, we can use the uniform kernel $K(u) = 1\{|u| \leq 1\}/2 $ and the rule of thumb bandwidths : 
\begin{align*}
h_{3,j} & =  \hat{s}_j \left[\frac{4.5 \mathbb P_{n,j} \hat V_{i,j}(\tau) ||W_{i,j}||_2^4 }{q(\tau) \left\Vert \mathbb P_{n,j} \hat V_{i,j}(\tau) W_{i,j}W_{i,j}^\top \right\Vert_F^2} \right]^{1/5}n_j^{-1/5} \\
h_{4,j} & =  \hat{s}_j \left[\frac{4.5 \mathbb P_{n,j} \hat V_{i,j}(\tau) ||W_{i,j}||_2^2 ||\hat Z_{i,j}||_2^2 }{q(\tau) \left\Vert \mathbb P_{n,j} \hat V_{i,j}(\tau) W_{i,j}\hat Z_{i,j}^\top \right\Vert_F^2} \right]^{1/5}n_j^{-1/5},
\end{align*}
where $\hat{s}_j$ is the sample standard error of $\{\underline{\hat{\eps}}_{i,j}\}_{i \in I_{n,j}}$.
In Section \ref{sec:proj} of the Online Supplement, we also provide the regularity conditions that imply $\hat{\Phi}_{i,j}(\tau)$ constructed from the cluster-level estimation satisfies Assumption \ref{ass:boot}(i). 

We note that in the dataset, if there exist some clusters with rather few numbers of observations, then the finite-sample performance of the cluster-level estimation may be negatively affected. In this case, we recommend first merging such small clusters into larger ones or using the full-sample estimation approaches described in Remarks \ref{rem:par} and \ref{rem:nonpar} instead.

\end{remark}

\subsection{Inference for Wald Statistics}
\label{subsec:subvector-result}

Denote $Q_{\Psi,\Psi}(\tau) = \sum_{j \in [J]}\xi_jQ_{\Psi,\Psi,j}(\tau)$, $Q_{\Psi,X}(\tau) = \sum_{j \in [J]} \xi_j Q_{\Psi,X,j}(\tau)$, $Q_{\Phi,X,j}(\tau) = \omega Q_{\Psi,X,j}(\tau)$, and $Q_{\Phi,X}(\tau) = \sum_{ j \in [J]} \xi_j Q_{\Phi,X,j}(\tau)$, where $\omega = (0_{d_\phi \times d_w}, \mathbb{I}_{d_\phi})$.

\begin{assumption}
	\begin{enumerate}[label=(\roman*)]
		\item 	There are compact subsets $\mathcal{R}$ and $\Theta$ of $\textbf{R}^{d_w}$ and $\textbf{R}^{d_\phi}$, respectively, such that Jacobian matrix $\frac{\partial}{\partial(b^\top,r^\top)}\Pi(b,r,0,\tau)$ is continuous and has full column rank, uniformly in $n$ and over $\mathcal{B} \times \mathcal{R} \times \Theta \times \Upsilon$.\footnote{For a sequence of matrices $A_n(v)$ indexed by $v \in \mathcal{V}$ and $n$, we say that $A_n(v)$ is of full column rank uniformly over $v \in \mathcal{V}$ and $n$ if 
			$\inf_{v \in \mathcal{V}, n \rightarrow \infty}\lambda_{\min}(A_n^\top(v)A_n(v)) \geq \underline{c}>0,
$ for some constant $\underline{c}$.}  
		\item The image of $\mathcal{B} \times \mathcal{R}$ under the mapping $(b,r) \mapsto \Pi(b,r,0,\tau)$ is simply connected.	
		\item Suppose $\sup_{\tau \in \Upsilon}||\hat{A}_1(\tau) - A_1 (\tau)||_{op} = o_p(1)$, where $A_1(\tau)$ is a symmetric $d_\phi \times d_\phi$ deterministic matrix such that 
		$0<c\leq \inf_{\tau \in \Upsilon}\lambda_{\min}(A_1 (\tau)) \leq \sup_{\tau \in \Upsilon}\lambda_{\max}(A_1 (\tau)) \leq C<\infty, \text{and}$ 
		\begin{align*}
		0<c\leq & \inf_{\tau \in \Upsilon}\left(Q_{\Phi,X}^\top(\tau) Q_{\Phi,\Phi}^{-1}(\tau) A_1 (\tau)  Q_{\Phi,\Phi}^{-1} Q_{\Phi,X}(\tau) \right) \\
		\leq & \sup_{\tau \in \Upsilon}\left(Q_{\Phi,X}^\top(\tau) Q_{\Phi,\Phi}^{-1}(\tau) A_1 (\tau)  Q_{\Phi,\Phi}^{-1} Q_{\Phi,X}(\tau) \right) \leq C < \infty
		\end{align*}
		for some constants $c,C$.
	\end{enumerate}
	\label{ass:id2}
\end{assumption}

\begin{remark}
    Assumptions \ref{ass:id2}(i) and \ref{ass:id2}(ii) are Assumptions R$5^*$ and R$6^*$ in \cite{Chernozhukov-Hansen(2008a)}. They, along with Assumption \ref{ass:id}, imply that $\beta_n(\tau)$ is uniquely defined. Second, by \citet[Theorem 2]{Chernozhukov-Hansen(2006)}, under Assumptions \ref{ass:id} and \ref{ass:id2}(i)--\ref{ass:id2}(iii), $(\beta_n(\tau),\gamma_n(\tau))$ uniquely solves the system of equations $\mathbb{E}(\tau - 1\{y_{i,j} \leq X_{i,j} b + W_{i,j}^\top r \})\Psi_{i,j}(\tau)V_{i,j}(\tau) = 0$. Third, Assumption \ref{ass:id2}(iii) implies $Q_{\Phi,X}(\tau)$ is of full column rank, and thus, $\beta_n(\tau)$ is strongly identified. 
    However, it allows for the presence of weak IV clusters. Specifically, let us define 
    \begin{align}\label{eq:a_j}
    a_j(\tau) = \Gamma(\tau) Q_{\Psi,X,j}(\tau) = \tilde \Gamma(\tau)Q_{\Phi,X,j}(\tau), 
    \end{align}    
    and
		\begin{align}
	\Gamma(\tau) &= \left[Q_{\Psi,X}^\top(\tau) Q_{\Psi,\Psi}^{-1}(\tau) \omega^\top A_1 (\tau) \omega  Q_{\Psi,\Psi}^{-1}(\tau) Q_{\Psi,X}(\tau)\right]^{-1}Q_{\Psi,X}^\top(\tau) Q_{\Psi,\Psi}^{-1}(\tau) \omega^\top  A_1 (\tau) \omega Q_{\Psi,\Psi}^{-1}(\tau), \notag \\
 \tilde\Gamma(\tau) &= \left[Q_{\Phi,X}^\top(\tau) Q_{\Phi,\Phi}^{-1}(\tau) A_1(\tau)  Q_{\Phi,\Phi}^{-1}(\tau) Q_{\Phi,X}(\tau)\right]^{-1}Q_{\Phi,X}^\top(\tau) Q_{\Phi,\Phi}^{-1}(\tau) A_1(\tau)  Q_{\Phi,\Phi}^{-1}(\tau).
		\label{eq:Gamma}
		\end{align}
Here, $a_j(\tau)$ measures the identification strength of the $j$-th cluster and $\sum_j \xi_j a_j(\tau) = 1$ by construction. We say the $j$-th cluster is a weak IV cluster if $||Q_{\Phi,X,j}(\tau)||_2 = 0$, which implies $a_j(\tau) = 0$. In contrast, the $j$-th cluster is a strong IV cluster if $a_j(\tau) \neq 0$. 
When there exist $j \in [J]$ such that $a_j(\tau)=0$, the inference procedures that are based on cluster-level IVQR estimators of $\beta_n(\tau)$ (e.g., $\hat{\beta}_{j}(\tau)$ for $j \in [J]$) can become invalid,\footnote{For example, the identification for the $j$-th cluster may be too weak for $\hat{\beta}_j(\tau)$, the IVQR estimator of the $j$-th cluster, to retain consistency. In this case, the inference methods based on cluster-level estimators will become invalid.} while our gradient bootstrap procedure remains valid, provided that the overall identification, captured by $Q_{\Phi,X}(\tau)$, is strong.
\end{remark}

\begin{assumption}\label{ass:TQR_alternative}
\begin{enumerate}[label=(\roman*)]
\item 
Suppose $\sup_{\tau \in \Upsilon}|\hat{A}_2(\tau) - A_2(\tau)| = o_p(1)$, where $A_2(\tau)$ is deterministic and $0 < c \leq A_2(\tau) \leq C <\infty,$ for some constants $c,C$.
\item Suppose there exists a subset $\mathcal J_s$ of $[J]$ such that $\inf_{j \in \mathcal J_s,\tau \in \Upsilon}a_j(\tau)\geq c_0>0$ and $a_{j}(\tau) = 0$ for $(j,\tau) \in ([J] \backslash \mathcal J_s) \times \Upsilon$, where $a_j(\tau)$ is defined in (\ref{eq:a_j}). Further denote $J_s = |\mathcal J_s|$, which satisfies  $$\lceil |\textbf{G}|(1-\alpha) \rceil \leq |\textbf{G}|-2^{J-J_s+1}.$$ 
\end{enumerate}
\end{assumption}

\begin{theorem}\label{thm:unstudentizedsize}
	Suppose Assumptions \ref{ass:id}-\ref{ass:id2} and \ref{ass:TQR_alternative}(i) hold. 	Then under $\mathcal{H}_0$ defined in (\ref{eq: hypothesis}), that is, $\mu_{\beta}(\tau) = 0$ for $\tau \in \Upsilon$, 
	\begin{align*}
	\alpha - \frac{1}{2^{J-1}} \leq \liminf_{n \rightarrow \infty} \mathbb{P}(T_{n} > \hat{c}_{n}(1-\alpha)) \leq \limsup_{n \rightarrow \infty} \mathbb{P}(T_{n} > \hat{c}_{n}(1-\alpha)) \leq \alpha + \frac{1}{2^{J-1}}. 
	\end{align*}
	In addition, if Assumption \ref{ass:TQR_alternative}(ii) holds, then under $\mathcal{H}_{1,n}$ defined in (\ref{eq: hypothesis}), 
	\begin{align*}
\lim_{\sup_{\tau \in \Upsilon}|\mu_{\beta}(\tau)| \rightarrow \infty} \liminf_{n \rightarrow \infty}\mathbb{P}(T_{n} > \hat{c}_{n}(1-\alpha)) = 1. 
	\end{align*}
\end{theorem}

\begin{remark}
Theorem \ref{thm:unstudentizedsize} shows that the $T_n$-based gradient wild bootstrap test controls size asymptotically when at least one of the clusters is strong. The error $1/2^{J-1}$ can be viewed as the upper bound for the asymptotic size distortion, which vanishes exponentially with the total number of clusters rather than the number of strong IV clusters. Intuitively, although the weak IV clusters do not contribute to the identification of $\beta_n$, 
the scores of such clusters, i.e., $r_n\mathbb{P}_{n,j}\tilde{f}_{\tau}(D,\beta_n(\tau),\gamma_n(\tau),0)$  for $j \in ([J] \backslash \mathcal J_s)$, still contribute to the limiting distribution of the IVQR estimator, which in turn determines the total number of possible sign changes in the bootstrap Wald statistic.  
\end{remark}

\begin{remark}
Furthermore, Theorem \ref{thm:unstudentizedsize} shows that the gradient wild bootstrap test has power against $r_n^{-1}$-local alternatives if further Assumption \ref{ass:TQR_alternative}(ii) holds. 
To see why Assumption \ref{ass:TQR_alternative}(ii) is needed,  note that our procedure compares the test statistic $T_n$ with the critical value $\hat c_n(1-\alpha)$, where $T_n$ is asymptotically equivalent to $T^*_n (\iota_J)$, i.e., the bootstrap test statistic with $g$ equal to a $J \times 1$ vector of ones, and the critical value $\hat c_n(1-\alpha)$ is just the $\lceil |\textbf{G}|(1-\alpha) \rceil$-th order statistic of $\{T^*_n (g)\}_{ g \in \textbf{G}}$. 
In the proof of Theorem \ref{thm:unstudentizedsize}, we show that when $|\mu_{\beta}(\tau)| \rightarrow \infty$ and the signs of $g_j$ for all strong IV clusters are the same (only weak IV clusters have different signs),
$T^*_n (g)$ is equivalent to $T^*_n (\iota_J)$, and thus, $T_n$, even under the alternative. 
Intuitively, the effect of $|\mu_{\beta}(\tau)|$  on the asymptotic behaviour of $T_n^*(g)$ is only manifested through sign changes on those strong IV clusters (the effect of sign changes from the weak IV clusters becomes negligible).

Let us denote the set of $g$'s that only flip the sign of weak IV clusters as $\textbf{G}_w = \{g: g_j = g_{j'}, \forall j, j' \in \mathcal J_s\}$. Given there are $J_s$ strong IV clusters, the cardinality of $\textbf{G}_w$ is $2^{J-J_s+1}$. 
To establish the power against $\mathcal{H}_{1,n}$ in Theorem \ref{thm:unstudentizedsize}, we request that our bootstrap critical value $\hat c_n(1-\alpha)$ does not take values of $T^*_n(g)$ for $g \in \textbf{G}_w$ because otherwise the test statistic $T_n$ and the critical value are asymptotically equivalent even under the alternative. 
This implies
\begin{align*}
\lceil |\textbf{G}|(1-\alpha) \rceil \leq |\textbf{G}|-2^{J-J_s+1}.
\end{align*}
Therefore, we need a sufficient number of strong IV clusters to establish the power result. 
For instance, the condition $\lceil |\textbf{G}|(1-\alpha) \rceil \leq |\textbf{G}|-2^{J-J_s+1}$ requires that 
$J_s \geq 5$ and $J_s \geq 6$ for $\alpha=10\%$ and $5\%$, respectively. Theorem \ref{thm:unstudentizedsize} suggests that although the size of the gradient wild bootstrap test is well controlled even with only one strong IV cluster, its power depends on the number of strong IV clusters.
\end{remark}

\begin{remark}\label{rem:Hagemann}
We conjecture that when $\beta_n(\tau)$ is strongly identified, our gradient wild bootstrap-based Wald inference procedure is also valid for IVQR under the alternative asymptotic framework with a large number of small clusters (e.g., see \cite{Hagemann(2017)} in the QR context).
Specifically, to establish bootstrap validity in this case, we need to impose regularity conditions similar to those in \cite{Hagemann(2017)} and show that conditional on the data, as the number of small clusters diverges, the distribution of the resampling process for the bootstrap IVQR estimator $r_n(\hat \beta^* (\tau) - \hat \beta (\tau) )$ is approximately the same as that of the sampling process $r_n(\hat \beta (\tau) - \beta_0(\tau))$ under the null. However, such arguments for the consistency of the bootstrap distribution are rather different from the randomization test perspective underlying the proof of Theorem \ref{thm:unstudentizedsize}. For the conciseness of the paper, we leave this direction of investigation for future research. 

\end{remark}

We need the following assumption for the Wald inference of IVQR with CRVE. 
\begin{assumption}\label{ass:J}
\begin{enumerate}[label=(\roman*)]
\item Suppose
\begin{align*}
\sup_{\tau \in \Upsilon}||    \hat G(\tau) - G(\tau)||_{op} = o_p(1),
\end{align*}
where $\inf_{\tau \in \Upsilon}(G^\top(\tau) G(\tau)) \geq c>0$ for some constant $c$ and 
$\hat G(\tau)$ is defined in \eqref{eq:A1random}.  
\item Let $b_j(\tau)  =  G^\top (\tau)Q_{\Phi,X,j}(\tau)$ and 
\begin{align*}
    v_{j}^\top(\tau) = -[\xi_1,\cdots,\xi_J] \otimes  (b_j(\tau) \tilde \Gamma(\tau)) + [0_{d_\phi(j-1)}^\top,  G^\top (\tau),0_{d_\phi(J-j)}^\top] \in \Re^{1 \times d_\phi J},
\end{align*}
where $\tilde \Gamma(\tau) = \omega \Gamma(\tau)$. Then, the rank of $[v_{1}(\tau),\cdots,v_{J}(\tau)] \in \Re^{(d_\phi J) \times J}$ is strictly greater than 1 uniformly over $\tau \in \Upsilon$. 


\end{enumerate}
\end{assumption}

\begin{remark}
Assumption \ref{ass:J}(i) guarantees that the CRVE $\hat{A}_{CR}(\tau)$ is invertible and the corresponding test statistic $T_{CR,n}$ does not degenerate. Assumption \ref{ass:J}(ii) is a rank condition, which holds in general when $J>1$. 
\end{remark}


\begin{theorem}\label{thm:studentized}
	Suppose Assumptions \ref{ass:id}--\ref{ass:id2} and \ref{ass:J} hold.	Then under $\mathcal{H}_0$ defined in (\ref{eq: hypothesis}), that is, $\mu_{\beta}(\tau) = 0$ for $\tau \in \Upsilon$, 
	\begin{align*}
	\alpha - \frac{1}{2^{J-1}} \leq \liminf_{n \rightarrow \infty} \mathbb{P}(T_{CR,n} > \hat{c}_{CR,n}(1-\alpha)) \leq \limsup_{n \rightarrow \infty} \mathbb{P}(T_{CR,n} > \hat{c}_{CR,n}(1-\alpha)) \leq \alpha + \frac{1}{2^{J-1}}. 
	\end{align*}
	In addition, suppose there exists a subset $\mathcal J_s$ of $[J]$ such that $\inf_{j \in \mathcal J_s, \tau \in \Upsilon}\min(|a_j(\tau)|,|b_j(\tau)|)\geq c_0>0$, $a_{j}(\tau) = b_{j}(\tau) = 0$ for $(j,\tau) \in [J] \backslash \mathcal J_s \times \Upsilon$, and $\lceil |\textbf{G}|(1-\alpha) \rceil \leq |\textbf{G}|-2^{J-J_s+1}$, where $J_s = |\mathcal J_s|$. Then under $\mathcal{H}_{1,n}$ defined in (\ref{eq: hypothesis}),
	\begin{align*}
\lim_{\sup_{\tau \in \Upsilon}|\mu_{\beta}(\tau)| \rightarrow \infty} \liminf_{n \rightarrow \infty}\mathbb{P}(T_{CR,n} > \hat{c}_{CR,n}(1-\alpha)) = 1.
	\end{align*}
\end{theorem}

\begin{theorem}\label{thm:power_comparison}
Suppose the assumptions in Theorem \ref{thm:studentized} hold and $\Upsilon$ is a singleton. 
Let 
$$\textbf{G}(c_0) = \{ g \in \textbf{G}: |\sum_{j \in [J]} g_j \xi_j a_j(\tau)| \geq c_0\}$$ 
for some positive constant $c_0$. We assume there exists a constant $c_0>0$ such that $|\textbf{G}(c_0)| > |\textbf{G}| - \lceil |\textbf{G}|(1-\alpha) \rceil$. Then under $\mathcal{H}_{1,n}$ in \eqref{eq: hypothesis}, 
for any $\delta>0$, there exists a constant $c_{\mu}>0$ such that when $|\mu_{\beta}(\tau)| >c_\mu$, 
\begin{align*}
\liminf_{n \rightarrow \infty}\mathbb{P}(\phi^{cr}_n \geq \phi_n) \geq 1-\delta,
\end{align*}
where $\phi^{cr}_n = 1 \{ T_{CR,n} > \hat{c}_{CR,n}(1-\alpha)\}$ and $\phi_n = 1 \{ T_{n} > \hat{c}_{n}(1-\alpha) \}$.
\end{theorem}


\begin{remark}\label{rem: power_comparison}
Theorem \ref{thm:studentized} shows that similar to the $T_n$-based bootstrap test, the bootstrap test with CRVE controls size asymptotically (with a small error) as long as there exists at least one strong IV cluster, and it has power against $r_n^{-1}$-local alternative if there are a sufficient number of strong IV clusters.  However, the power result in Theorem \ref{thm:studentized} is derived using arguments very different from those for the one without CRVE. In particular, distinct from the $T_n$-based test, the local power of the $T_{CR,n}$-based test is established without the assumption that the IVQR first-stage coefficients have the same sign for all clusters, which may not hold in some empirical studies (e.g., Figure \ref{fig:first_stage}). Therefore, the power result for the bootstrap test with CRVE allows for substantially more heterogeneity in the first stage of IVQR.


Furthermore, we establish in Theorem \ref{thm:power_comparison} that in the case where $\Upsilon$ is a singleton (i.e., testing $\mathcal{H}_0:  \beta_n(\tau)= \beta_0(\tau) \; v.s. \; \mathcal{H}_{1,n}: \beta_n(\tau) \neq  \beta_0(\tau)$ for a certain quantile index $\tau$), the power of the $T_{CR,n}$-based bootstrap test dominates that based on $T_n$ with a large probability when the local parameter $\mu_{\beta}(\tau)$ is sufficiently different from zero. Specifically, in this case, we have
\begin{align*}
    1\{T_n > \hat{c}_{n}(1-\alpha)\} = 1\{T_{CR,n} >  \tilde{c}_{CR,n}(1-\alpha)\},
\end{align*}
where $\tilde{c}_{CR,n}(1-\alpha)$ denotes the $(1-\alpha)$ quantile of $\left\{ ||\hat{\beta}_g^*(\tau) -\hat{\beta}(\tau)||_{\hat{A}_{CR}(\tau)}: g \in \textbf{G}\right\}$, and $\hat{A}_{CR}(\tau)$ is the inverse of the original CRVE instead of its bootstrap analogue. Then, Theorem \ref{thm:power_comparison} follows because $\tilde{c}_{CR,n}(1-\alpha) > \hat{c}_{CR,n}(1-\alpha)$ with large probability as $|\mu_{\beta}(\tau)|$ becomes sufficiently large.
Intuitively, $\hat \Omega(\tau,\tau)$, and thus, $\hat A_{CR}(\tau)$ have random limits under a fixed number of clusters. 
We show that when $\mathcal{H}_0$ is true, the original $\hat{A}_{CR}(\tau)$ and its bootstrap counterpart $\hat{A}_{CR,g}^*(\tau)$ have the same limit distribution. By contrast, under $\mathcal{H}_{1,n}$, although the local parameter does not enter the limit distribution of $\hat{A}_{CR}(\tau)$, it does enter that of $\hat{A}_{CR,g}^*(\tau)$ because of our design of the formula in (\ref{eq:Vhatg}). 
Then, when $\mu_{\beta}(\tau)$ is sufficiently different from zero, 
it becomes dominant in $\hat{A}_{CR,g}^*(\tau)$, helping to ``drag down" the value of each $T^*_n(g)$ and, as a consequence, the bootstrap critical value $\hat{c}_{CR,n}(1-\alpha)$.
This gives the $T_{CR,n}$-based bootstrap test a power advantage over its $T_n$-based counterpart. The assumption in Theorem \ref{thm:power_comparison} that there exists a constant $c_0>0$ such that $|\textbf{G}(c_0)| > |\textbf{G}| - \lceil |\textbf{G}|(1-\alpha) \rceil$ is also mild. For example, it can hold even with one strong IV cluster as $|\textbf{G}(c_0)|$ is at least as large as the number of all possible sign changes of the weak IV clusters given that $\sum_{j \in [J]}\xi_ja_j(\tau)=\sum_{j \in \mathcal{J}_s}\xi_ja_j(\tau)=1$.

\end{remark}

\subsection{Inference for Weak-instrument-robust Statistics}\label{subsec:fullvector-result}
\begin{assumption}
	Suppose one of the conditions below holds. 
	\begin{enumerate}[label=(\roman*)]
		\item There exists a symmetric $d_\phi \times d_\phi$ matrix $A_3 (\tau)$ such that $\sup_{\tau \in \Upsilon}||\hat{A}_3(\tau) - A_3 (\tau)||_{op} = o_p(1)$, where $A_3(\tau)$ is some deterministic matrix and for some constants $c$ and $C$,
		$$0<c\leq \inf_{\tau \in \Upsilon}\lambda_{\min}(A_3 (\tau)) \leq \sup_{\tau \in \Upsilon}\lambda_{\max}(A_3 (\tau)) \leq C<\infty.$$ 
		\item Suppose $\hat{A}_3(\tau)$ is set as $\tilde {A}_{CR}(\tau) =  \left[\hat H  (\tau)\tilde \Omega(\tau,\tau)\hat H(\tau)\right]^{-1}$, where  $\sup_{\tau \in \Upsilon}||\hat H (\tau) - H(\tau)||_{op} = o_p(1)$ for some deterministic $H(\tau)$ such that  
  $$0<c\leq \inf_{\tau \in \Upsilon}\lambda_{\min}( H^\top (\tau)  H(\tau)) \leq \sup_{\tau \in \Upsilon}\lambda_{\max}( H^\top (\tau) H(\tau)) \leq C<\infty.$$ 
  Furthermore, we require $J > d_\phi$ in this case. 
	\end{enumerate}
	\label{ass:full}
\end{assumption}

\begin{theorem}\label{thm:full}
	Suppose Assumptions \ref{ass:id}--\ref{ass:boot} and \ref{ass:full} hold and $\beta_n(\tau) = \beta_0(\tau)$ for $\tau \in \Upsilon$. Then, 
	\begin{align*}
	\alpha - \frac{1}{2^{J-1}} \leq & \liminf_{n \rightarrow \infty} \mathbb{P}(AR_{n}> \hat{c}_{AR,n}(1-\alpha)) \\
 &	\leq  \limsup_{n \rightarrow \infty} \mathbb{P}(AR_{n}> \hat{c}_{AR,n}(1-\alpha)) \leq \alpha + \frac{1}{2^{J-1}}, \,\,\text{and}
	\end{align*}
	\begin{align*}
	\alpha - \frac{1}{2^{J-1}} \leq & \liminf_{n \rightarrow \infty} \mathbb{P}(AR_{CR,n}> \hat{c}_{AR,CR,n}(1-\alpha)) \\
&	\leq \limsup_{n \rightarrow \infty} \mathbb{P}(AR_{CR,n}> \hat{c}_{AR,CR,n}(1-\alpha)) \leq \alpha + \frac{1}{2^{J-1}}.
	\end{align*}
\end{theorem}

\begin{remark}
Theorem \ref{thm:full} holds without assuming strong identification (i.e., Assumption \ref{ass:id2}). The asymptotic size of the $AR_{n}$ and $AR_{CR,n}$-based bootstrap inference is therefore controlled up to an error $2^{1-J}$, even when $\beta_n(\tau)$ is weakly or partially identified. This aligns with the robust inference approach proposed by \cite{Chernozhukov-Hansen(2008a)} for i.i.d. data, which is based on chi-squared critical values.  
\end{remark}

\section{Monte Carlo Simulation}\label{sec: simu}
To examine the performance of the gradient wild bootstrap inference for IVQR, we consider two data-generating processes (DGPs). The first DGP is inspired by those designed by \cite{KS17} and \cite{KW18} and considers the setting with clustered data and within-cluster dependence. The second DGP is inspired by that designed by \cite{BDF09} with network dependence, and then, the clusters are obtained by spectral clustering as proposed by \cite{L22}. 
For both designs, we set the number of observations, simulation repetitions, and bootstrap repetitions as 500, 500, and 300, respectively. 

\subsection{Simulation Designs}
\noindent \textbf{DGP 1.} Let $a_{k,j} = (a_{k,1,j},\cdots,a_{k,n_j,j})^\top$ for $k=1,\cdots,d_z$, where $d_z$ is the dimension of $Z$, and $u_{j} = (u_{l,1,j},\cdots,u_{l,n_j,j})^\top$ for $l=1,2$. Then,  $(a_{1,j},\cdots,a_{d_z,j},u_{1,j},u_{2,j})$ are independent, and each of them is independent across $j \in [J]$ and follows an $n_j \times 1$ multivariate normal distribution with mean zero and covariance $\Sigma(\rho_j)$, where $\Sigma(\rho_j)$ is an $n_j \times n_j$ toeplitz matrix with coefficient $\rho_j$ and $\rho_j = 0.2+0.5j/J$ for $j \in [J]$. We define the original instruments as
\begin{align*}
Z_{k,i,j} =  1\{F_N(a_{k,i,j})>0.5\}, \quad \forall k \in [d_z],
\end{align*}
where $F_N(\cdot)$ is the standard normal CDF. Then, the endogenous variable $X_{i,j}$ is generated as 
\begin{align*}
& X_{i,j} = 1\left\{0.1+\sum_{k=1}^{d_z} \Pi_{k,j} (F_N(a_{k,i,j}) -0.5) + 0.5(F_N(u_{i,j})-0.5) >0  \right\},    
\end{align*}
where $\Pi_{k,j}$ determines the identification strength. To allow for first-stage heterogeneity in the identification strength across clusters, for all $k \in [d_z]$, we let 
\begin{align*}
 \Pi_{k,j} = 
 \begin{cases}
 &     \pi \quad \text{if} \quad 1\leq j \leq J/3\\
 &    0 \quad \text{if} \quad J/3< j \leq 2J/3\\
 &    2\pi \quad \text{if} \quad 2J/3< j \leq J
  \end{cases},   
\end{align*}
where $\pi \in (1,1/2,1/4)$. 

The outcome $Y_{i,j}$ is generated as 
\begin{align*}
& Y_{i,j} = X_{i,j}\beta(u_{1,i,j}) + \underline{W}_{i,j} \gamma + \sqrt{0.1}u_{i,j},
\end{align*}
where $\underline{W}_{i,j}$ is 
distributed following $\frac{1}{2}\chi^2_{1}$ and independent of $(a_{1,j},\cdots,a_{d_z,j},u_{1j}, u_{2j})$, and $W_{i,j} = (1,\underline{W}_{i,j})^\top$. We set $\beta(u) = 1 + F_N(u)$, $\gamma = 0$, and the total number of observations $n=500$. The correlation between $X_{i,j}$ and $u_{i,j}$ is about 0.77, indicating the endogeneity level. 

To allow for unbalanced clusters, we follow \cite{Djogbenou-Mackinnon-Nielsen(2019)} and \cite{Mackinnon2021} and set the cluster sizes as
\begin{align*}
    n_j = \left[ n \frac{\exp( r \cdot j /J)}{\sum_{j \in [J]} \exp(r \cdot j /J) }  \right], \;\; \text{for} \;\; j = 1, ..., J-1,  
\end{align*}
and $n_J = n - \sum_{j \in [j]} n_j$. 
We let $r=4$ when $J=9$ and $J=18$ to generate substantial heterogeneity in cluster sizes. The parameter $r$ corresponds to the heterogeneity parameter considered in the simulations of \cite{Djogbenou-Mackinnon-Nielsen(2019)} and \cite{mackinnon2022cluster}, and $r=4$ is the most heterogeneous setting considered in those papers. 
 When $J=9$, the cluster sizes are $(5,8,12,19,30,48,75,117,186)$. When $J=18$, the resulting cluster sizes are $(2,2,3,4,5,7,8,10,13,17,21,26,33,41,52,65,81,110)$. We  test the following hypothesis for $d_z \in \{1,3\}$ and $\tau \in \{0.1,0.25,0.5,0.75,0.9\}$:
\begin{align*}
H_0: \beta(\tau) = 1+ \tau \quad v.s. \quad H_1: \beta(\tau) = 0.5+ \tau.
\end{align*}

\noindent \textbf{DGP 2.} We consider the linear-in-mean social interaction model detailed in Example \ref{ex:network}. Specifically, following \cite{L22}, we generate the network $\mathcal A$ with $n=500$ nodes as $\mathcal A_{\ell,\ell'} = 1\{||\eta_{\ell} - \eta_{\ell'}||_2 \geq (7/(\pi n))^{1/2}\}$, where $\eta_\ell \stackrel{i.i.d.}{\sim} \text{Uniform}[0,1]^2$. 


Then, we generate $B_\ell$ in  \eqref{eq:linear-in-mean}, the $\ell$-th node's background characteristic, as follows:  
\begin{align*}
    B_\ell = \begin{cases}
0 & \text{ with probability $1/2$} \\
\exp( - \log(2) + (\log(4))^{1/2}\eps_\ell') & \text{ with probability $1/2$} 
    \end{cases},
\end{align*}
where $\eps_\ell' \stackrel{i.i.d.}{\sim}\N(0,1)$. Then, the control variables $W$ and IV $Z$ can be constructed as in Example \ref{ex:network}. 

Further denote $\tilde {\mathcal A}$ a normalized adjacency matrix with a typical entry $\tilde {\mathcal A}_{\ell,\ell'} = \mathcal  A_{\ell,\ell'}/d_\ell$, where $d_\ell = \sum_{\ell' \neq \ell} {\mathcal A}_{\ell,\ell'}$ is the degree of node $\ell$. Then, we generate the outcome variable as  \eqref{eq:linear-in-mean} in which $U = \{U_\ell\}_{\ell \in [n]}$ is a sequence of i.i.d. uniform (0,1) random variables and 
\begin{align*}
    y = \left(\mathbb I_n - \diag(\beta(U)) \tilde{\mathcal A}\right)^{-1} \left(\delta_0(U) + \diag(\delta_1(U)) B + \diag(\delta_2(U))\tilde{\mathcal A} B\right),
\end{align*}
where $\diag(v)$ for a $n \times 1$ vector $v$ denotes an $n \times n$ diagonal matrix with $v$ as the diagonal, $y = (y_1,\cdots,y_n)^\top$, $B = (B_1,\cdots,B_n)^\top$, $\beta(U) = (\beta(U_1),\cdots,\beta(U_n))^\top$, and $(\delta_0(U), \delta_1(U), \delta_2(U))$ are similarly defined. In addition, we set $\delta_0(u) = 0.7683+0.25(u-0.5)$, $\beta(u) = 0.4666+0.2(u-0.5)$, $\delta_1(u) = 0.0834+0.1(u-0.5)$, and $\delta_2(u) = 0.1507+0.2(u-0.5)$. For comparison, the values  $(\delta_0(0.5),\beta(0.5),\delta_1(0.5),\delta_2(0.5)) = (0.7683,0.4666,0.0834,0.1507)$ are set according to the calibration study by \cite{BDF09}. Given the outcome variable, we can set the endogenous variable as $X = \tilde{\mathcal A}y$, where $X = (X_1,\cdots,X_n)^\top$. Finally, we note that \eqref{eq:IVQR} holds because coefficients $(\delta_0(\cdot),\beta(\cdot))$ are monotone increasing and their corresponding regressors are positive and $\{U_\ell\}_{\ell \in [n]} \indep \{B_\ell\}_{\ell \in [n]}$. The correlation between $X_\ell$ and $U_\ell$ is about 0.16, which indicates the endogeneity level. As discussed in Example \ref{ex:network}, we let $W = (\iota_n, B, \tilde A B)$ and $Z = \tilde A^2 B$. 

We follow the classification procedure proposed by \cite{L22} to obtain the clusters. 
\begin{enumerate}
    \item Input: a positive integer $L$ and network $\mathcal A$. 
    \item Compute all separated components (no links between two components) of the network denoted as  $\{\mathcal V_{h}\}_{h \in \{0\} \cup [H]}$, where $\{\mathcal V_{h}\}_{h \in \{0\} \cup [H]}$ is a partition of $[n]$ (n vertexes) and they are sorted in ascending order according to their sizes. We keep all the components whose sizes are greater than 5. Denote the number of components left as $L'+1$ for some $L' \geq 0$.  
    \item Suppose the biggest component $\mathcal V_0$ has size $\tilde n_0$. By permuting labels, we suppose  $\mathcal V_0 = [\tilde n_0]$ and denote its adjacency matrix as $\mathcal A_0$. Then, we compute the graph Laplacian as 
\begin{align*}
\mathcal L_0 = \mathbb I_{\tilde n_0} - D_0^{-1/2} \mathcal A_0 D_0^{-1/2},    
\end{align*}
where   $D_0 = \diag(\sum_{i \in [\tilde n_0]} A_{0,1,i},\cdots,\sum_{i \in [\tilde n_0]} A_{0,\tilde n_0,i})$ is an $\tilde n_0 \times \tilde n_0$ diagonal matrix of degrees.  
\item Obtain the top $L$ eigenvector matrix of $\mathcal L_0$ corresponding to its $L$ largest eigenvalues and denote it as $ V = [V_1^\top,\cdots,V_{\tilde n_0}^\top]^\top,$
where $V_{i} \in \Re^{L}$ for $i \in [\tilde n_0]$. 
\item Apply the k-means algorithm to $V$ and divide $\mathcal V_0 = [\tilde n_0]$ into $L$ groups, denoted as $\mathcal V_{0,1},\cdots,\mathcal V_{0,L}$. 
\item Output: we obtain $J = L+L'$ clusters $(\mathcal V_{0,1},\cdots,\mathcal V_{0,L},\mathcal V_{1},\cdots,\mathcal V_{L'})$. 
\end{enumerate}
We let $L$ be $(10,20)$. The number of clusters ($J$) depends on $L'$, which varies across simulation replications. Note that different from DGP 1, there is no cluster heterogeneity in the identification strength under DGP 2. 
We  test the following hypothesis for $\tau \in \{0.1,0.25,0.5,0.75,0.9\}$:
\begin{align*}
H_0: \beta(\tau) = 0.4666+0.2(\tau-0.5) \quad v.s. \quad H_1: \beta(\tau) = 1.2166+0.2(\tau-0.5).
\end{align*}

\subsection{Inference Procedure}
We compare the performance of our three bootstrap inference methods  $T_{CR,n}$, $T_n$, and $AR_n$ (denoted as T\_CR, T, and AR in Tables \ref{tab:1}-\ref{tab:5}) with the conventional asymptotic Wald inference based on CRVE and two alternative methods based on cluster-level IVQR estimators available in the literature.
\begin{enumerate}
\item T\_STD: This inference method is based on the same test statistic as $T_{CR,n}$ but with the conventional normal critical value. 
    \item IM: This inference method is proposed by \cite{Ibragimov-Muller(2010)}, which compares their group-based $t$-test statistic with the critical value of a $t$-distribution with $J-1$ degrees of freedom. The test statistic is constructed by separately running IVQR using the samples in each cluster.  
    \item CRS: The approximate randomization inference method proposed by \cite{Canay-Romano-Shaikh(2017)}, which compares IM's group $t$-test statistic with the critical value of the sign changes-based randomization distribution of the statistic.
\end{enumerate}

We construct the instrument using the nonparametric approach outlined in Remark \ref{rem:nonpar}. The tuning parameters are set based on the rule-of-thumb provided in Remark \ref{rem:nonpar}. Because this approach eventually produces a scalar $\hat \Phi_{i,j}(\tau)$ regardless of the dimension of the original IVs $Z_{i,j}$, our $AR_n$ and $AR_{CR,n}$ tests are numerically equivalent. Therefore, we only report the performance of $AR_n$. In addition, because we focus on inference point-wise in $\tau$ and $\hat \Phi_{i,j}(\tau)$ is a scalar, the choices of $\hat A_1(\tau)$ in the estimation of $\hat \beta(\tau)$, $\hat A_2(\tau)$ in the Wald inference based on $T_n(\tau)$, $\hat G(\tau)$ in the Wald inference based on $T_{CR,n}(\tau)$, and $\hat A_3(\tau)$ in the AR inference based on $AR_n(\tau)$ will not affect the performance of corresponding inference methods, and are set to 1 without loss of generality. 

\subsection{Simulation Results}

Tables \ref{tab:1}--\ref{tab:4} and \ref{tab:5} collect the simulation results for DGPs 1 and 2, respectively, when the nominal null rejection rate is set at $10\%$.\footnote{The rejection probabilities using $AR_n$ and $AR_{CR,n}$ are numerically the same with one IV.} Several key observations emerge from the results. First, T\_CR, T, and AR effectively control size across various settings, while T\_STD, IM, and CRS have size distortions, particularly in DGP 1, which includes both weak and strong IV clusters (i.e., there exists substantial cluster heterogeneity in the IV strength under DGP 1). The size distortions of T\_STD, IM, and CRS also typically increase when $\pi$ becomes small in Tables \ref{tab:1}-\ref{tab:4}. 
Second, for both DGPs, T\_CR demonstrates greater power than T across various settings, aligning with our theoretical expectations. 
Third, we observe that T\_STD typically has larger size distortion when the number of clusters is small, while the size distortions of IM and CRS tend to increase with the number of clusters. 
Fourth, the power of CRS and IM is comparable, and in DGP 2, their power is similar to that of T\_CR, although they display slightly higher size distortion.
Overall, T\_CR has the best size and power performance among the six methods, and is thus our recommended inference procedure.

\begin{table}[htbp]
  \centering
    \begin{tabular}{l|rrrrr|r|rrrrr}
          & \multicolumn{5}{c|}{H0}                &       & \multicolumn{5}{c}{H1} \\ \hline 
    $\tau$   & 0.1   & 0.25  & 0.5   & 0.75  & 0.9   & \multicolumn{1}{l|}{$\tau$} & 0.1   & 0.25  & 0.5   & 0.75  & 0.9 \\ \hline 
          & \multicolumn{11}{c}{$\pi = 1$} \\ \hline 
    T\_CR & 9.4   & 9.4   & 11.6  & 7.4   & 4.4   &       & 66.6  & 62.0  & 69.0  & 67.4  & 52.4 \\
    T     & 6.8   & 8.6   & 8.6   & 5.4   & 3.2   &       & 60.2  & 51.4  & 70.4  & 63.6  & 48.4 \\
    AR    & 5.4   & 6.4   & 8.8   & 7.2   & 5.6   &       & 26.4  & 64.8  & 69.2  & 50.2  & 21.0 \\
    T\_STD & 18.8  & 20.0  & 16.4  & 15.8  & 24.2  &       & 7.0   & 48.2  & 99.6  & 99.0  & 93.2 \\
    IM    & 48.0  & 37.2  & 21.0  & 23.0  & 22.6  &       & 92.8  & 89.6  & 84.8  & 86.6  & 81.8 \\
    CRS   & 47.4  & 37.8  & 22.2  & 23.2  & 23.2  &       & 92.6  & 89.8  & 85.2  & 86.8  & 81.2 \\ \hline 
          & \multicolumn{11}{c}{$\pi = 0.5$} \\ \hline 
    T\_CR & 8.8   & 11.0  & 11.6  & 8.2   & 10.6  &       & 64.8  & 59.4  & 64.0  & 47.2  & 33.4 \\
    T     & 6.6   & 9.2   & 8.4   & 5.8   & 10.2  &       & 62.4  & 48.6  & 62.6  & 45.6  & 31.8 \\
    AR    & 6.2   & 8.0   & 9.8   & 6.6   & 5.4   &       & 21.0  & 65.2  & 61.2  & 23.2  & 11.6 \\
    T\_STD & 20.2  & 17.6  & 17.8  & 23.6  & 43.8  &       & 3.0   & 11.8  & 97.8  & 86.6  & 57.4 \\
    IM    & 48.0  & 33.2  & 28.2  & 27.6  & 29.8  &       & 87.6  & 85.8  & 80.6  & 79.4  & 78.6 \\
    CRS   & 47.6  & 35.6  & 30.0  & 27.8  & 29.8  &       & 86.8  & 85.4  & 80.2  & 79.2  & 78.4 \\ \hline 
          & \multicolumn{11}{c}{$\pi = 0.25$} \\ \hline 
    T\_CR & 10.8  & 8.4   & 5.8   & 9.2   & 10.8  &       & 46.4  & 74.0  & 38.6  & 16.8  & 18.6 \\
    T     & 9.2   & 5.8   & 4.8   & 7.6   & 10.6  &       & 44.6  & 67.8  & 36.6  & 15.6  & 17.0 \\
    AR    & 5.6   & 6.4   & 9.6   & 6.6   & 8.4   &       & 9.8   & 72.2  & 22.2  & 6.6   & 8.0 \\
    T\_STD & 17.0  & 9.4   & 9.4   & 36.2  & 58.0  &       & 1.0   & 0.2   & 60.4  & 33.8  & 29.0 \\
    IM    & 49.4  & 51.6  & 29.2  & 30.4  & 37.2  &       & 89.2  & 88.6  & 79.2  & 77.8  & 82.6 \\
    CRS   & 50.0  & 52.2  & 30.6  & 30.8  & 39.0  &       & 89.4  & 87.2  & 78.6  & 77.8  & 83.0 \\
    \end{tabular}%
      \caption{Rejection Probability of DGP 1 with 1 IV and $J=9$}
  \label{tab:1}%
\end{table}%

\begin{table}[htbp]
  \centering
    \begin{tabular}{l|rrrrr|r|rrrrr}
          & \multicolumn{5}{c|}{H0}                &       & \multicolumn{5}{c}{H1} \\ \hline 
    $\tau$   & 0.1   & 0.25  & 0.5   & 0.75  & 0.9   & \multicolumn{1}{l|}{$\tau$} & 0.1   & 0.25  & 0.5   & 0.75  & 0.9 \\ \hline 
          & \multicolumn{11}{c}{$\pi = 1$} \\ \hline 
    T\_CR & 8.0   & 10.2  & 10.6  & 6.6   & 4.4   &       & 84.8  & 96.2  & 99.0  & 90.0  & 67.8 \\
    T     & 6.6   & 9.2   & 9.4   & 4.6   & 4.2   &       & 81.6  & 89.6  & 97.2  & 86.4  & 63.6 \\
    AR    & 7.2   & 7.8   & 9.4   & 8.0   & 5.0   &       & 42.0  & 92.8  & 97.4  & 77.0  & 40.4 \\
    T\_STD & 12.6  & 10.2  & 9.6   & 9.8   & 18.0  &       & 1.2   & 7.2   & 100.0 & 96.2  & 82.8 \\
    IM    & 60.6  & 52.4  & 32.8  & 33.2  & 31.8  &       & 99.8  & 99.6  & 99.0  & 98.6  & 99.2 \\
    CRS   & 60.4  & 50.4  & 32.4  & 32.8  & 32.0  &       & 99.4  & 99.6  & 98.8  & 98.6  & 99.2 \\ \hline 
          & \multicolumn{11}{c}{$\pi = 0.5$} \\ \hline 
    T\_CR & 6.8   & 11.0  & 9.4   & 7.8   & 12.6  &       & 76.0  & 95.8  & 94.6  & 55.2  & 31.2 \\
    T     & 5.2   & 10.0  & 7.4   & 7.0   & 9.6   &       & 71.0  & 88.6  & 90.2  & 52.6  & 29.4 \\
    AR    & 7.6   & 9.2   & 9.4   & 9.0   & 9.6   &       & 32.8  & 90.4  & 86.0  & 33.2  & 12.8 \\
    T\_STD & 10.6  & 7.2   & 5.8   & 16.4  & 37.4  &       & 0.4   & 0.2   & 100.0 & 71.6  & 40.4 \\
    IM    & 62.2  & 54.0  & 33.2  & 39.8  & 49.8  &       & 100.0 & 100.0 & 99.2  & 97.0  & 99.2 \\
    CRS   & 62.0  & 53.4  & 33.6  & 38.6  & 49.6  &       & 100.0 & 100.0 & 99.4  & 97.0  & 99.2 \\ \hline 
          & \multicolumn{11}{c}{$\pi = 0.25$} \\ \hline 
    T\_CR & 12.6  & 9.0   & 5.4   & 9.6   & 11.6  &       & 43.2  & 90.6  & 47.8  & 16.6  & 19.6 \\
    T     & 11.0  & 6.0   & 3.0   & 8.0   & 11.8  &       & 40.4  & 83.4  & 46.0  & 16.4  & 17.8 \\
    AR    & 7.6   & 6.8   & 8.0   & 9.2   & 9.2   &       & 17.0  & 80.8  & 25.6  & 7.4   & 9.2 \\
    T\_STD & 15.2  & 5.0   & 8.0   & 29.0  & 55.4  &       & 1.8   & 0.0   & 100.0 & 33.6  & 39.2 \\
    IM    & 61.6  & 60.8  & 46.6  & 55.4  & 75.4  &       & 99.8  & 99.4  & 98.0  & 98.8  & 99.8 \\
    CRS   & 61.8  & 60.8  & 46.2  & 55.4  & 74.2  &       & 99.4  & 99.6  & 97.4  & 98.6  & 99.8 \\
    \end{tabular}%
      \caption{Rejection Probability of DGP 1 with 1 IV and $J=18$}
  \label{tab:2}%
\end{table}%

\begin{table}[htbp]
  \centering
    \begin{tabular}{l|rrrrr|r|rrrrr}
          & \multicolumn{5}{c|}{H0}                &       & \multicolumn{5}{c}{H1} \\ \hline 
    $\tau$   & 0.1   & 0.25  & 0.5   & 0.75  & 0.9   & \multicolumn{1}{l|}{$\tau$} & 0.1   & 0.25  & 0.5   & 0.75  & 0.9 \\ \hline 
          & \multicolumn{11}{c}{$\pi = 1$} \\ \hline 
    T\_CR & 10.0  & 10.2  & 9.6   & 10.2  & 5.2   &       & 67.0  & 55.4  & 71.3  & 67.4  & 61.2 \\
    T     & 6.6   & 9.0   & 9.6   & 7.6   & 5.2   &       & 63.4  & 47.4  & 66.7  & 64.2  & 54.2 \\
    AR    & 7.8   & 6.4   & 9.2   & 7.8   & 6.4   &       & 27.0  & 36.2  & 63.6  & 51.8  & 26.6 \\
    T\_STD & 19.2  & 19.0  & 14.8  & 18.4  & 21.0  &       & 6.4   & 40.4  & 96.8  & 96.8  & 93.4 \\
    IM    & 40.8  & 28.6  & 21.8  & 21.8  & 19.6  &       & 88.8  & 85.0  & 84.4  & 81.6  & 75.4 \\
    CRS   & 42.4  & 30.4  & 22.4  & 22.4  & 19.6  &       & 88.0  & 85.6  & 82.9  & 80.0  & 74.2 \\ \hline 
          & \multicolumn{11}{c}{$\pi = 0.5$} \\ \hline 
    T\_CR & 9.0   & 11.0  & 10.4  & 8.4   & 6.8   &       & 63.4  & 58.0  & 64.5  & 67.2  & 55.0 \\
    T     & 7.8   & 11.0  & 10.8  & 6.2   & 4.8   &       & 64.0  & 49.2  & 64.3  & 67.4  & 54.6 \\
    AR    & 7.6   & 9.2   & 10.0  & 7.8   & 7.0   &       & 27.6  & 32.2  & 63.6  & 51.0  & 19.2 \\
    T\_STD & 17.0  & 16.6  & 16.8  & 19.0  & 26.4  &       & 3.0   & 28.4  & 95.6  & 95.8  & 87.4 \\
    IM    & 46.0  & 38.2  & 30.4  & 31.2  & 25.4  &       & 87.2  & 86.8  & 82.0  & 84.0  & 80.0 \\
    CRS   & 47.0  & 38.6  & 31.6  & 32.2  & 26.2  &       & 87.6  & 87.2  & 82.0  & 84.6  & 79.6 \\\hline  
          & \multicolumn{11}{c}{$\pi = 0.25$} \\\hline  
    T\_CR & 8.2   & 11.0  & 10.2  & 8.0   & 11.6  &       & 57.4  & 57.0  & 61.6  & 42.8  & 30.4 \\
    T     & 7.2   & 8.4   & 7.6   & 7.0   & 9.2   &       & 58.8  & 50.6  & 55.5  & 43.0  & 29.6 \\
    AR    & 7.6   & 7.2   & 10.0  & 7.8   & 8.2   &       & 22.0  & 35.6  & 51.5  & 25.0  & 9.8 \\
    T\_STD & 16.0  & 13.4  & 12.6  & 21.0  & 44.6  &       & 2.2   & 9.2   & 91.2  & 76.4  & 45.2 \\
    IM    & 49.6  & 42.6  & 39.4  & 36.6  & 32.8  &       & 87.6  & 86.6  & 85.1  & 83.0  & 81.6 \\
    CRS   & 50.0  & 44.0  & 41.6  & 38.6  & 33.4  &       & 87.2  & 88.0  & 84.2  & 82.8  & 81.2 \\ 
    \end{tabular}%
      \caption{Rejection Probability of DGP 1 with 3 IVs and $J=9$}
  \label{tab:3}%
\end{table}%

\begin{table}[htbp]
  \centering
    \begin{tabular}{l|rrrrr|r|rrrrr}
          & \multicolumn{5}{c|}{H0}                &       & \multicolumn{5}{c}{H1} \\ \hline 
    $\tau$   & 0.1   & 0.25  & 0.5   & 0.75  & 0.9   & \multicolumn{1}{l|}{$\tau$} & 0.1   & 0.25  & 0.5   & 0.75  & 0.9 \\ \hline 
          & \multicolumn{11}{c}{$\pi = 1$} \\ \hline 
    T\_CR & 10.2  & 10.2  & 10.8  & 10.4  & 6.2   &       & 83.8  & 94.1  & 94.7  & 91.8  & 84.0 \\
    T     & 6.8   & 8.8   & 10.0  & 9.4   & 5.2   &       & 79.2  & 83.1  & 92.4  & 89.8  & 78.0 \\
    AR    & 8.0   & 7.8   & 10.8  & 9.8   & 6.6   &       & 43.2  & 75.4  & 92.0  & 85.4  & 54.2 \\
    T\_STD & 12.2  & 14.4  & 12.6  & 10.8  & 16.6  &       & 1.4   & 29.8  & 97.2  & 98.0  & 95.8 \\
    IM    & 26.4  & 20.0  & 13.8  & 10.6  & 9.8   &       & 93.0  & 90.1  & 84.9  & 79.2  & 74.6 \\
    CRS   & 26.0  & 19.8  & 14.0  & 10.0  & 9.8   &       & 92.8  & 89.0  & 84.0  & 77.8  & 74.8 \\ \hline 
          & \multicolumn{11}{c}{$\pi = 0.5$} \\ \hline 
    T\_CR & 8.8   & 9.4   & 11.0  & 8.8   & 6.8   &       & 84.6  & 93.4  & 96.6  & 89.4  & 67.0 \\
    T     & 6.0   & 8.6   & 8.4   & 8.2   & 6.2   &       & 80.0  & 82.5  & 91.8  & 87.2  & 67.8 \\
    AR    & 6.0   & 6.8   & 8.0   & 8.6   & 7.8   &       & 46.4  & 74.3  & 90.8  & 77.6  & 38.2 \\
    T\_STD & 12.8  & 11.0  & 8.6   & 12.0  & 21.0  &       & 1.0   & 16.0  & 96.6  & 97.0  & 87.2 \\
    IM    & 39.2  & 31.2  & 19.2  & 17.0  & 14.0  &       & 94.8  & 91.4  & 89.1  & 87.2  & 80.4 \\
    CRS   & 39.2  & 30.4  & 19.2  & 16.4  & 13.6  &       & 94.8  & 91.0  & 88.7  & 86.6  & 80.2 \\ \hline 
          & \multicolumn{11}{c}{$\pi = 0.25$} \\ \hline 
    T\_CR & 8.8   & 11.6  & 8.4   & 5.0   & 13.2  &       & 79.0  & 89.5  & 92.2  & 64.8  & 35.6 \\
    T     & 8.8   & 10.4  & 7.6   & 4.0   & 10.6  &       & 76.0  & 77.9  & 87.8  & 64.2  & 34.4 \\
    AR    & 9.0   & 10.6  & 9.0   & 7.2   & 7.4   &       & 41.2  & 68.0  & 83.8  & 41.0  & 12.4 \\
    T\_STD & 10.4  & 9.6   & 6.0   & 12.0  & 36.2  &       & 0.6   & 2.4   & 90.4  & 79.4  & 47.8 \\
    IM    & 42.0  & 36.0  & 26.8  & 18.8  & 16.0  &       & 95.2  & 95.2  & 91.0  & 88.6  & 82.4 \\
    CRS   & 41.0  & 35.8  & 24.8  & 19.2  & 16.2  &       & 94.8  & 95.8  & 89.9  & 87.8  & 82.6 \\ 
    \end{tabular}%
      \caption{Rejection Probability of DGP 1 with 3 IVs and $J=18$}
  \label{tab:4}%
\end{table}%

\begin{table}[htbp]
  \centering
    \begin{tabular}{l|rrrrr|r|rrrrr}
          & \multicolumn{5}{c|}{H0}                &       & \multicolumn{5}{c}{H1} \\ \hline 
   $\tau$    & 0.1   & 0.25  & 0.5   & 0.75  & 0.9   & \multicolumn{1}{l|}{$\tau$} & 0.1   & 0.25  & 0.5   & 0.75  & 0.9 \\ \hline 
          & \multicolumn{11}{c}{$L= 10$} \\ \hline 
    T\_CR & 10.0  & 8.3   & 11.3  & 13.0  & 7.7   &       & 78.7  & 69.7  & 69.3  & 77.0  & 82.0 \\
    T     & 5.0   & 6.7   & 9.7   & 6.7   & 3.3   &       & 37.7  & 50.0  & 61.3  & 70.3  & 65.3 \\
    AR    & 6.3   & 8.7   & 9.3   & 6.7   & 4.7   &       & 37.3  & 59.7  & 67.0  & 73.7  & 63.0 \\
    T\_STD & 36.0  & 32.0  & 26.7  & 34.0  & 16.7  &       & 80.7  & 68.0  & 63.0  & 52.7  & 59.3 \\
    IM    & 13.7  & 8.0   & 8.3   & 10.7  & 8.7   &       & 74.7  & 69.0  & 71.0  & 76.0  & 78.0 \\
    CRS   & 13.7  & 8.3   & 8.7   & 10.0  & 9.0   &       & 74.7  & 69.7  & 71.0  & 75.0  & 78.0 \\ \hline 
          & \multicolumn{11}{c}{$L=20$} \\ \hline 
    T\_CR & 6.3   & 8.7   & 12.0  & 12.7  & 5.0   &   & 86.0  & 75.0  & 77.0  & 83.0  & 84.0 \\
    T     & 3.0   & 7.0   & 10.7  & 7.0   & 2.7   &    & 42.7  & 53.7  & 65.0  & 74.0  & 66.0 \\
    AR    & 3.3   & 6.7   & 9.3   & 9.7   & 5.7   &    & 42.3  & 64.3  & 72.3  & 77.7  & 65.7 \\
    T\_STD & 30.7  & 27.3  & 29.3  & 33.7  & 17.7  &    & 74.3  & 63.7  & 66.0  & 52.3  & 57.3 \\
    IM    & 14.3  & 13.0  & 13.7  & 13.3  & 12.7  &    & 81.3  & 82.3  & 82.0  & 87.0  & 93.3 \\
    CRS   & 15.0  & 13.7  & 12.7  & 13.0  & 13.0  &    & 81.0  & 81.0  & 81.3  & 86.0  & 92.7 \\
    \end{tabular}%
    \caption{Rejection Probability with DGP 2}
  \label{tab:5}%
\end{table}%

\section{Empirical Applications}\label{sec: emp}
In an influential study, \cite{ADH2013} analyzes the effect of rising Chinese import competition on wages and employment in US local labor markets between 1990 and 2007, when the share of total US spending on Chinese goods increased substantially from 0.6\%  to 4.6\%. In this section, 
we further analyze the region-wise distributional effects of such import exposure by applying IVQR and the proposed gradient wild bootstrap procedures to the Census Bureau-designated South region with 16 states and total number of observations equal to 578.\footnote{We focus on the South region because it has the highest IV strength.}

For the IVQR model, we let the outcome variable $(y_{i,j})$ denote the decadal change in the average individual log weekly wage in a given CZ.
The endogenous variable $(X_{i,j})$ is the change in Chinese import exposure per worker in a CZ, which is instrumented by $(Z_{i,j})$ Chinese import growth in other high-income countries.\footnote{See Sections I.B and III.A in \cite{ADH2013} for a detailed definition of these variables.} We follow the nonparametric approach in Remark \ref{rem:nonpar} to construct $\hat \Phi_{i,j}$ used in IVQR. 
In addition, the exogenous variables $(W_{i,j})$ include the characteristic variables of commuting zones (CZs) and decade specified in \cite{ADH2013} as well as state fixed effects.
Our IV quantile regressions are based on the CZ samples in the South region, and the samples are clustered at the state level, following \cite{ADH2013}.
Besides the results for the full sample, we also report those for male and female samples separately. 

The main results are given in Table \ref{tab:app-IV}. Specifically, for $\tau \in \{0.1, 0.25, 0.5, 0.75, 0.9\}$, we report the point IVQR estimate and the 90\% bootstrap confidence sets (CSs) constructed by inverting the corresponding $AR_n$, $T_n$, and $T_{CR,n}$-based tests. The computation of the bootstrap CSs was conducted over the parameter space $[-3,1]$ with a step size of 0.01, and the number of bootstrap draws is set at 300 for each step. We can draw three key observations from the results in Table \ref{tab:app-IV}. First, the impact of Chinese imports on wages shows distributional heterogeneity. Specifically, the three types of bootstrap CSs reveal that the effects are relatively significant at the high quantiles, followed by the median quantile, but not at the lower quantiles, across the full, male, and female samples. Second, the $T_{CR,n}$-based CSs are generally shorter than those of $AR_n$ and $T_n$. This aligns with our theory and simulation results, which suggest that the $T_{CR,n}$-based bootstrap test is more effective at detecting distant local alternatives. Third, the distributional effects of Chinese imports on wages are fairly consistent between males and females. In addition, we observe that the effects of Chinese imports are relatively more substantial for the male samples.

\begin{table}[htbp]
  \centering
    \begin{tabular}{r|c|c|c|c|c|c}
Gender          & $\tau$ & \multicolumn{1}{c|}{0.1} & \multicolumn{1}{c|}{0.25} & \multicolumn{1}{c|}{0.5} & \multicolumn{1}{c|}{0.75} & \multicolumn{1}{c}{0.9} \\ \hline 
    \multicolumn{1}{l|}{All} & Point Est. & \multicolumn{1}{c|}{-0.67} & \multicolumn{1}{c|}{-0.56} & \multicolumn{1}{c|}{-0.72} & \multicolumn{1}{c|}{-1.12} & \multicolumn{1}{c}{-1.17} \\    
          & AR    & [-2.86,1.00] & [-2.07,0.76] & [-3.00,-0.04] & [-2.52,-0.26] & [-3.00,1.00] \\
          & T     & [-2.82,1.00] & [-2.02,1.00] & [-2.65,0.92] & [-2.37,-0.25] & [-2.99,0.66] \\
          & T\_CR  & [-2.15,0.49] & [-1.47,0.45] & [-2.26,0.24] & [-2.20,-0.23] & [-1.81,-0.39] \\ \hline 
    \multicolumn{1}{l|}{Male} & Point Est. & \multicolumn{1}{c|}{-0.27} & \multicolumn{1}{c|}{-0.46} & \multicolumn{1}{c|}{-0.81} & \multicolumn{1}{c|}{-1.42} & \multicolumn{1}{c}{-1.38}  \\
          & AR    & [-3.00,1.00] & [-1.99,0.84] & [-3.00,-0.16] & [-3.00,-0.47] & [-3.00,0.92] \\
          & T     & [-2.99,1.00] & [-1.99,1.00] & [-2.88,0.64] & [-2.99,-0.46] & [-3.00,0.24] \\
          & T\_CR & [-1.16,0.66] & [-1.45,0.60] & [-2.21,0.01] & [-2.83,-0.35] & [-2.92,-0.18] \\ \hline
    \multicolumn{1}{l|}{Female} & Point Est. & \multicolumn{1}{c|}{-0.40} & \multicolumn{1}{c|}{-0.27} & \multicolumn{1}{c|}{-0.52} & \multicolumn{1}{c|}{-1.12} & \multicolumn{1}{c}{-0.90} \\
          & AR    & [-1.37,1.00] & [-1.96,0.99] & [-3.00,0.28] & [-2.60,-0.17] & [-3.00,1.00] \\
          & T     & [-1.37,1.00] & [-1.72,1.00] & [-2.82,0.99] & [-2.35,-0.18] & [-2.80,1.00] \\
          & T\_CR & [-0.91,0.60] & [-1.50,0.97] & [-2.24,0.89] & [-2.60,-0.36] & [-2.09,-0.27] \\
    \end{tabular}%
    \caption{IVQR Point Estimates and Confidence Intervals for South Region}
  \label{tab:app-IV}%
\end{table}%

\newpage

\appendix

\section{Constructing the IVs}
\label{sec:proj0}
In this section, we discuss how to implement projections at both the full-sample and the cluster levels for IVQR to construct IVs that satisfy the requirement in Assumption \ref{ass:boot}.
\label{sec:proj_ivqr}
\subsection{Full-Sample Projection}
We first consider the full-sample projection mentioned in Remark \ref{rem:nonpar} in the main text. 
\begin{assumption}
	Recall $\underline{\hat{Q}}_{W,W,j}$ and $\underline{\hat{Q}}_{W,Z,j}$ defined in \eqref{eq:Jhat_rrj} and \eqref{eq:Jhatunder_rtj} in the main text. 
	Define 
	\begin{align*}
	& Q_{W,W,j}(\tau) =  \lim_{n \rightarrow \infty}\overline{\mathbb{P}}_{n,j}f_{\eps_{i,j}(\tau)}(0|W_{i,j},Z_{i,j})V_{i,j}(\tau)W_{i,j}W_{i,j}^\top,\\
	&Q_{W,Z,j}(\tau) =  \lim_{n \rightarrow \infty}\overline{\mathbb{P}}_{n,j}f_{\eps_{i,j}(\tau)}(0|W_{i,j},Z_{i,j})V_{i,j}(\tau)W_{i,j}\hat Z_{i,j}^\top  ,\\
	& Q_{W,W}(\tau) = \sum_{j \in [J]} \xi_j Q_{W,W,j}(\tau),  \quad Q_{W,Z}(\tau) = \sum_{j \in [J]} \xi_j Q_{W,Z,j}(\tau), \quad \text{and} \quad \chi(\tau) = Q^{-1}_{W,W}(\tau)Q_{W,Z}(\tau).
	\end{align*} 
	\begin{enumerate}[label=(\roman*)]
		\item Suppose 
		\begin{align*}
		\sup_{\tau \in \Upsilon}\left[||\hat{\underline{Q}}_{W,W,j}(\tau) - Q_{W,W,j}(\tau)||_{op}+||\hat{\underline{Q}}_{W,Z,j}(\tau) - Q_{W,Z,j}(\tau)||_{op}\right] = o_p(1).
		\end{align*}  
		\item Recall $\delta_{i,j}(v,\tau) = X_{i,j}^\top v_b + W_{i,j}^\top v_r + \hat{\Phi}_{i,j}^\top(\tau) v_t.$ Then, 
		\begin{align*}
		\sup \left\Vert \overline{\mathbb{P}}_{n,j} f_{\eps_{i,j}(\tau)}(\delta_{i,j}(v,\tau)|W_{i,j},Z_{i,j})V_{i,j}(\tau)W_{i,j} \hat Z_{i,j}^\top - Q_{W,Z,j}(\tau)\right\Vert_{op} \convP 0,
		\end{align*}
		where the supremum is taken over $\{j \in [J],||v||_2 \leq \delta,\tau \in \Upsilon\}$, $v = (v_b^\top,v_r^\top,v_t^\top)^\top$.
		\item For $j \in [J]$, there exists $\chi_{n,j}(\tau)$ such that 
		\begin{align*}
		Q_{W,Z,j}(\tau) = Q_{W,W,j}(\tau)\chi_{n,j}(\tau) \quad \text{and} \quad \overline{\mathbb{P}}_{n,j}||W_{i,j}^\top(\chi(\tau) - \chi_{n,j}(\tau))||_{op}^2 = o(1). 
		\end{align*}
		\item There exist constants $c,C$ such that 
		\begin{align*}
		0<c<\inf_{\tau \in \Upsilon }\lambda_{\min}(Q_{W,W}(\tau))\leq \sup_{\tau \in \Upsilon }\lambda_{\max}(Q_{W,W}(\tau)) \leq C<\infty.   
		\end{align*}
	\end{enumerate}
	\label{ass:boot'}
\end{assumption}

Assumption \ref{ass:boot'}(i) holds under mild regularity conditions because $\beta_n(\tau)$ converges to $\beta_0(\tau)$ under local alternatives. Second, Assumption \ref{ass:boot'} requires the nonparametric estimation of $Q_{W,W,j}(\tau) $ and $Q_{W,Z,j}(\tau)$ via kernel smoothing. We discuss the choice of kernel function and bandwidth in Remark \ref{rem:nonpar}. Third, Assumption \ref{ass:boot'}(ii) holds under mild smoothness conditions. Fourth, Assumption \ref{ass:boot'}(iii) is discussed in Remark \ref{rem:nonpar}.

\begin{proposition}
	Suppose Assumption \ref{ass:boot'} holds and $\hat{\Phi}_{i,j}(\tau) = \hat{Z}_{i,j} -\hat{\chi}^\top(\tau) W_{i,j}$, where $\hat{\chi}(\tau)$ is defined in \eqref{eq:chi} in the main text.  Then, $Q_{W,\Phi,j}(\tau)$ defined in Assumption \ref{ass:boot} is zero for $j \in [J]$ and $\tau \in \Upsilon$, i.e., by letting $n\rightarrow \infty$ followed by $\delta \rightarrow 0$, we have
	\begin{align*}
	\sup \left\Vert \overline{\mathbb{P}}_{n,j} f_{\eps_{i,j}(\tau)}(\delta_{i,j}(v,\tau)|W_{i,j},Z_{i,j})V_{i,j}(\tau)W_{i,j}\hat{\Phi}_{i,j}^\top(\tau) \right\Vert_{op} \convP 0,
	\end{align*}
	where the supremum is taken over $\{j \in [J],||v||_2 \leq \delta,\tau \in \Upsilon\}$ for $v = (v_b^\top,v_r^\top,v_t^\top)^\top$. 
	\label{prop:boot_full}
\end{proposition}

\subsection{Cluster-Level Projection}
\label{sec:cluster_proj}

In this section, we consider the cluster-level projection mentioned in Remark \ref{rem:cluster} in the main text. 
\begin{assumption}
	\begin{enumerate}[label=(\roman*)]
		\item Suppose 
		\begin{align*}
		\sup_{\tau \in \Upsilon}\left[||\underline{\hat{Q}}_{W,W,j}(\tau) - Q_{W,W,j}(\tau)||_{op}+||\underline{\hat{Q}}_{W,Z,j}(\tau) - Q_{W,Z,j}(\tau)||_{op}\right] = o_p(1).
		\end{align*}  
		\item Recall $\delta_{i,j}(v,\tau) = X_{i,j}^\top v_b + W_{i,j}^\top v_r + \hat{\Phi}_{i,j}^\top(\tau) v_t.$ Then, 
		\begin{align*}
		\sup\left\Vert \overline{\mathbb{P}}_{n,j} f_{\eps_{i,j}(\tau)}(\delta_{i,j}(v,\tau)|W_{i,j},Z_{i,j})V_{i,j}(\tau)W_{i,j}\hat{Z}_{i,j}^\top - Q_{W,Z,j}(\tau)\right\Vert_{op} \convP 0,
		\end{align*}
		where the supremum is taken over $\{j \in [J],||v||_2 \leq \delta,\tau \in \Upsilon\}$, $v = (v_b^\top,v_r^\top,v_t^\top)^\top$.
		\item Define the $k$-th largest singular value of $Q_{W,W,j}(\tau)$ as $\sigma_{k}(Q_{W,W,j}(\tau))$ for $k \in [d_w]$. Then there exist constants $c,C$ and integer $R \in [1,d_w]$ such that 
		\begin{align*}
		0<c \leq \inf_{\tau \in \Upsilon}\sigma_{R}(Q_{W,W,j}(\tau)) \leq \sup_{\tau \in \Upsilon}\sigma_{1}(Q_{W,W,j}(\tau))\leq C<\infty. 
		\end{align*}
	\end{enumerate}
	\label{ass:boot''}
\end{assumption}

We note that the use of generalized inverse in \eqref{eq:chi_j} in the main text accommodates the case that $Q_{W,W,j}(\tau)$ is not invertible. Assumption \ref{ass:boot'}(iii) only requires that the minimum nonzero singular value of $Q_{W,W,j}(\tau)$ is bounded away from zero uniformly over $\tau \in \Upsilon$. 

\begin{proposition}
	Suppose Assumption \ref{ass:boot''} holds and $\hat{\Phi}_{i,j}(\tau) = \hat{Z}_{i,j}(\tau) -\hat{\chi}_j^\top(\tau) W_{i,j}$, where $\hat{\chi}_j^\top(\tau)$ is defined in \eqref{eq:chi} in the main text. Then, $Q_{W,\Phi,j}(\tau)$ defined in Assumption \ref{ass:boot} is zero for $j \in [J]$ and $\tau \in \Upsilon$, i.e., by letting $n\rightarrow \infty$ followed by $\delta \rightarrow 0$, we have
	\begin{align*}
	\sup\left\Vert \overline{\mathbb{P}}_{n,j} f_{\eps_{i,j}(\tau)}(\delta_{i,j}(v,\tau)|W_{i,j},Z_{i,j})V_{i,j}(\tau)W_{i,j}\hat{\Phi}_{i,j}^\top(\tau) \right\Vert_{op} \convP 0,
	\end{align*}
	where the suprema are taken over $\{j \in [J],||v||_2 \leq \delta,\tau \in \Upsilon\}$ for $v = (v_b^\top,v_r^\top,v_t^\top)^\top$. 
	\label{prop:boot}
\end{proposition}

\section{Proof of Theorem \ref{thm:unstudentizedsize}}
\label{sec:pf_sec3}
Recall $\Gamma(\tau)$ and $\tilde{\Gamma}(\tau)$
defined in \eqref{eq:Gamma} in the main text and $\omega = (0_{d_\phi \times d_w}, \mathbb{I}_{d_\phi})$. 
Because $Q_{\Psi,\Psi}(\tau)$ is a block matrix, so is $Q_{\Psi,\Psi}^{-1}(\tau)$, and thus
\begin{align*}
\omega Q_{\Psi,\Psi}^{-1}(\tau) = (0_{d_{\phi} \times d_w},    Q^{-1}_{\Phi,\Phi}(\tau)). 
\end{align*}

Then, we have
\begin{align*}
Q_{\Psi,X}^\top(\tau) Q_{\Psi,\Psi}^{-1}(\tau) \omega^\top A_1(\tau) \omega Q_{\Psi,\Psi}^{-1}(\tau) Q_{\Psi,X}(\tau) = Q_{\Phi,X}^\top(\tau) Q_{\Phi,\Phi}^{-1}(\tau) A_1(\tau)  Q_{\Phi,\Phi}^{-1}(\tau) Q_{\Phi,X}(\tau)
\end{align*}
and 
\begin{align}\label{eq:Jftilde}
& \omega Q^{-1}_{\Psi,\Psi}(\tau)f_\tau(D_{i,j},\beta_n(\tau),\gamma_n(\tau),0) = Q_{\Phi,\Phi}^{-1}(\tau)\tilde{f}_\tau(D_{i,j},\beta_n(\tau),\gamma_n(\tau),0) \quad \text{and} \notag \\
& \Gamma(\tau)f_\tau(D_{i,j},\beta_n(\tau),\gamma_n(\tau),0) = \tilde \Gamma(\tau)\tilde f_\tau(D_{i,j},\beta_n(\tau),\gamma_n(\tau),0).
\end{align}

Therefore, by Lemma \ref{lem:expand}, we have
\begin{align*}
r_n \left(\hat{\beta}(\tau) - \beta_n(\tau)\right) & = \tilde{\Gamma}(\tau)\left[\sum_{j \in [J]} \xi_j r_n(\mathbb{P}_{n,j} - \overline{\mathbb{P}}_{n,j})\tilde{f}_{\tau}(D_{i,j},\beta_n(\tau),\gamma_n(\tau),0)\right] + o_p(1),
\end{align*}
where the $o_p(1)$ term holds uniformly over $\tau \in \Upsilon$. 

Then we have
\begin{align*}
r_n(\hat{\beta}(\tau)  - \beta_0(\tau)) & = r_n \left(\hat{\beta}(\tau) - \beta_n(\tau)\right) + \mu_\beta(\tau)\\
& =  \tilde{\Gamma}(\tau)\left[\sum_{j \in [J]} \xi_j \mathcal{Z}_j(\tau)\right]+ \mu_\beta(\tau) + o_p(1),
\end{align*}
where the $o_p(1)$ term holds uniformly over $\tau \in \Upsilon$. 

By Lemma \ref{lem:bootexpand}, we have
\begin{align}\label{eq:hatbeta^*_g-hatbeta}
& r_n (\hat{\beta}_g^*(\tau) - \hat{\beta}(\tau)) =\tilde{\Gamma}(\tau)\left[\sum_{j \in [J]} g_j \xi_j \mathcal{Z}_j(\tau)\right] + \overline{a}_g^*(\tau)\mu_\beta(\tau) + o_p(1),
\end{align}
where $o_p(1)$ term holds uniformly over $\tau \in \Upsilon$, $\overline{a}_g^*(\tau) = \sum_{ j \in [J]}\xi_j g_j a_j(\tau)$, 
\begin{align*}
a_j(\tau) = \Gamma(\tau) Q_{\Psi,X,j}(\tau) = \tilde \Gamma(\tau)Q_{\Phi,X,j}(\tau).
\end{align*}

Let $T_\infty(g) = \sup_{\tau \in \Upsilon} \left\Vert \tilde{\Gamma}(\tau)\left[\sum_{j \in [J]} g_j \xi_j \mathcal{Z}_j(\tau)\right] + \overline{a}_g^*(\tau)\mu_\beta(\tau)\right\Vert_{A_2(\tau)}$. Then, we have, uniformly over $\tau \in \Upsilon$
\begin{align*}
\{r_n T_{n},\{r_n T_{n}^{*}(g) \}_{g \in \textbf{G}}\} \convD (T_\infty(\iota_J), \{T_\infty(g)\}_{g \in \textbf{G}}), 
\end{align*}
where $\iota_J$ is a $J \times 1$ vector of ones and we use the fact that $\overline{a}_{\iota_J}^*(\tau) = \sum_{j \in [J]} \xi_j a_j(\tau) = 1$. In addition, under $\mathcal{H}_0$, $T_\infty(g) \stackrel{ d}{=} T_\infty(g')$ for $g,g' \in \textbf{G}$. Let $k^* = \lceil |\textbf{G}|(1-\alpha) \rceil$. We order $\{T_n^{*}(g)\}_{g \in \textbf{G}}$ and $\{T_\infty(g)\}_{g \in \textbf{G}}$ in ascending order: 
\begin{align*}
(T_n^{*})^{(1)} \leq \cdots \leq (T_n^{*})^{(|\textbf{G}|)} \quad \text{and} \quad (T_\infty)^{(1)} \leq \cdots \leq (T_\infty)^{|\textbf{G}|}.  
\end{align*}
Then, we have
\begin{align*}
\hat{c}_{n}(1-\alpha) \convP (T_\infty)^{(k^*)}
\end{align*}
and 
\begin{align*}
\limsup_{n \rightarrow \infty}\mathbb{P} \{ T_n > \hat{c}_n(1-\alpha) \} & \leq \limsup_{n \rightarrow \infty}\mathbb{P}\{T_n \geq \hat{c}_n(1-\alpha)\} \\
& \leq \mathbb{P}\left\{ T_\infty(\iota_J) \geq (T_\infty)^{(k^*)}  \right\} \leq \alpha + \frac{1}{2^{J-1}}, 
\end{align*}
where the second inequality is due to  Portmanteau's theorem (see, e.g., \citep{VW96}, Theorem 1.3.4(iii)), and the third inequality is due to the properties of randomization tests (see, e.g., \citep{LR06}, Theorem 15.2.1) and the facts that the distribution of $T_\infty(g)$ is invariant w.r.t. $g$, $T_\infty(g) = T_\infty(-g)$, and $T_\infty(g) \neq T_\infty(g')$ if $g \notin \{g', -g'\}$.  Similarly, we have
\begin{align*}
\liminf_{n \rightarrow \infty}\mathbb{P} \{ T_n > \hat{c}_n(1-\alpha) \} = \liminf_{n \rightarrow \infty}\mathbb{P} \{ T_n > (T_n^{*})^{(k^*)} \} \geq  \mathbb{P}\left\{ T_\infty(\iota_J) > (T_\infty)^{(k^*)}  \right\} \geq \alpha - \frac{1}{2^{J-1}}. 
\end{align*}
To see the last inequality, we note that  $T_\infty(g) = T_\infty(-g)$, and $T_\infty(g) \neq T_\infty(g')$ if $g \notin \{g', -g'\}$. Therefore, 
\begin{align*}
\sum_{g \in \textbf{G}}1\{T_\infty(g) \leq (T_\infty)^{(k^*)}\} \leq k^* + 1.
\end{align*}
Then, with $|\textbf{G}| = 2^J$, we have 
\begin{align*}
|\textbf{G}| \mathbb{E}1\{T_\infty(\iota_J) > (T_\infty)^{(k^*)}\} & = \mathbb{E}\sum_{g \in \textbf{G}}1\{T_\infty(g) > (T_\infty)^{(k^*)}\} \\
& = |\textbf{G}|-\mathbb{E}\sum_{g \in \textbf{G}}1\{T_\infty(g) \leq (T_\infty)^{(k^*)}\} \geq |\textbf{G}| - (k^*+1) \\
& \geq \lfloor |\textbf{G}|\alpha \rfloor-1 \geq |\textbf{G}|\alpha  -2,
\end{align*}
where the first equality holds because $T_\infty(\iota_J) \stackrel{d}{=} T_\infty(g)$ for $g \in \textbf{G}$. 

Under $\mathcal{H}_{1,n}$, we still have 
\begin{align*}
\liminf_{n \rightarrow \infty}\mathbb{P} \{ T_{n} > \hat{c}_n(1-\alpha) \} & \geq  \mathbb{P}\left\{ T_\infty(\iota_J) > (T_\infty)^{(k^*)}  \right\}.
\end{align*}
Let $\textbf{G}_s = \textbf{G} \backslash \textbf{G}_w$, where $\textbf{G}_w = \{g \in \textbf{G}: g_{j} = g_{j'}, \forall j,j' \in \mathcal J_s\}$.  We aim to show that, as $\sup_{\tau \in \Upsilon}||\mu_{\beta}(\tau)||_2 \rightarrow \infty$, 
\begin{align}
\mathbb{P} \{ T_\infty(\iota_J) > \max_{g \in \textbf{G}_s}T_\infty(g) \} \rightarrow 1.     
\label{eq:Gs}
\end{align}
In addition, we note that $|\textbf{G}_s| = |\textbf{G}| - 2^{J-J_s+1} \geq k^*$. This implies 
as $\sup_{\tau \in \Upsilon}||\mu_\beta(\tau)||_2 \rightarrow \infty$, 
\begin{align*}
\mathbb{P} \{ T_\infty(\iota_J) > (T_\infty)^{(k^*)} \} \geq \mathbb{P}\{T_\infty(\iota_J) > \max_{g \in \textbf{G}_s}T_\infty(g)\} \rightarrow 1.     
\end{align*}
Therefore, it suffices to establish \eqref{eq:Gs}. Note $T_\infty(\iota_J) \geq \sup_{\tau \in \Upsilon}||\mu_\beta(\tau)||_{A_2(\tau)} - O_p(1)$ and 
\begin{align*}
\max_{g \in \textbf{G}_s}T_\infty(g) \leq \sup_{\tau \in \Upsilon,g \in \textbf{G}_s}|\sum_{j \in [J]}g_j \xi_j a_j(\tau)|  \sup_{\tau \in \Upsilon}||\mu_\beta(\tau)||_{A_2(\tau)} + O_p(1). 
\end{align*}
Because when $g \in \textbf{G}_s$, the signs of $\{g_j\}_{j \in \mathcal J_s}$ cannot be the same.  In addition, we have $\sum_{j \in [J]} \xi_j a_j(\tau) = \sum_{j \in \mathcal J_s} \xi_j a_j(\tau) = 1$, and $\inf_{\tau \in \Upsilon,j \in \mathcal J_s} a_j(\tau) \geq c_0>0$. These imply 
\begin{align*}
\max_{g \in \textbf{G}_s, \tau \in \Upsilon}|\sum_{j \in [J]}g_j \xi_j a_j(\tau)| \leq 1-2\min_{j \in [J]}\xi_jc_0  <1. 
\end{align*}
Then, as $\sup_{\tau \in \Upsilon}||\mu_\beta(\tau)||_{2} \rightarrow \infty$, we have 
\begin{align*}
& \sup_{\tau \in \Upsilon}||\mu_\beta(\tau)||_{A_2(\tau)} - (1-2\min_{j \in [J]}\xi_jc_0)\sup_{\tau \in \Upsilon}||\mu_\beta(\tau)||_{A_2(\tau)} \\
& \geq 2\min_{j \in [J]}\xi_jc_0 \inf_{\tau \in \Upsilon}\lambda_{\min}(A_2(\tau)) \sup_{\tau \in \Upsilon}||\mu_\beta(\tau)||_{2} \rightarrow \infty.
\end{align*}
This concludes the proof. 
$\blacksquare$

\section{Proof of Theorem \ref{thm:studentized}}
Recall $\tilde{\Gamma}(\tau)$ defined in \eqref{eq:Gamma},  
\begin{align*}
T_{CR,n} = \sup_{\tau \in \Upsilon} ||r_n(\hat{\beta}(\tau) - \beta_0(\tau))||_{\hat{A}_{CR}(\tau)}
\end{align*}
and 
\begin{align*}
T_{CR,n}^{*}(g) = \sup_{\tau \in \Upsilon} ||r_n(\hat{\beta}_g^*(\tau) -\hat{\beta}(\tau))||_{\hat{A}^*_{CR,g}(\tau)}.
\end{align*}

Following the proof of Theorem \ref{thm:unstudentizedsize}, we have
\begin{align*}
r_n(\hat{\beta}(\tau) - \beta_0(\tau)) = \tilde{\Gamma}(\tau)\left[\sum_{ j \in [J]} \xi_j \mathcal{Z}_j \right] + \mu_\beta(\tau) + o_p(1)
\end{align*}
and 
\begin{align*}
& r_n(\hat{\beta}_g^*(\tau) -\hat{\beta}(\tau))= \tilde{\Gamma}(\tau) \sum_{ j \in [J]}\xi_j g_j\mathcal Z_j + \overline a^*_g(\tau) \mu_\beta(\tau) + o_p(1),
\end{align*}
where the $o_p(1)$ terms in these two displays hold uniformly over $\tau \in \Upsilon$.

Next, we derive the limits of $\hat{A}_{CR}^{-1}(\tau)$ and $(\hat{A}_{CR,g}^*)^{-1}(\tau)$. Note
\begin{align*}
\hat{A}_{CR}^{-1}(\tau)     & =  \hat G^\top (\tau)\hat \Omega(\tau,\tau)\hat G(\tau)  \\
& = \sum_{j \in [J]} \xi_j n_j  \hat G^\top(\tau)\omega \left[\mathbb{P}_{n,j}\hat{f}_\tau(D_{i,j}, \hat{\beta}(\tau),\hat{\gamma}(\tau),0)\right]\left[\mathbb{P}_{n,j}\hat{f}_{\tau}(D_{i,j}, \hat{\beta}(\tau),\hat{\gamma}(\tau),0)\right]^\top \omega^\top \hat G(\tau) .
\end{align*}

Recall $b_j(\tau) = G^\top (\tau)Q_{\Phi,X,j}(\tau)$. Then, by Lemma \ref{lem:stud}, we have
\begin{align*}
& r_n  \hat G^\top(\tau) \omega \left[\mathbb{P}_{n,j}\hat{f}_\tau(D_{i,j}, \hat{\beta}(\tau),\hat{\gamma}(\tau),0)\right] \\
& = \left[G^\top(\tau) r_n\mathbb{P}_{n,j} \tilde{f}_\tau(D_{i,j}, \beta_n(\tau),\gamma_n(\tau),0) - b_j(\tau)\tilde{\Gamma}(\tau) r_n \mathbb{P}_{n} \tilde{f}_\tau(D_{i,j}, \beta_n(\tau),\gamma_n(\tau),0)\right]+ o_p(1)\\
& =  \left[G^\top(\tau) \mathcal{Z}_j(\tau) - b_j(\tau) \tilde \Gamma(\tau) \sum_{\tilde{j} \in J} \xi_{\tilde{j}}\mathcal{Z}_{\tilde{j}}(\tau)\right] + o_p(1), 
\end{align*}
where the $o_p(1)$ term holds uniformly over $\tau \in \Upsilon$. Let $\mathcal Z^\top  = [\mathcal Z_1^\top,\cdots,\mathcal Z_J^\top ]$. Then, by Assumption \ref{ass:J}(ii), we have 
\begin{align*}
& \inf_{\tau \in \Upsilon}
\sum_{j \in [J]} \xi_j^2 \left[G^\top(\tau) \mathcal{Z}_j(\tau) - b_j(\tau) \tilde \Gamma(\tau) \sum_{\tilde{j} \in J}\xi_{\tilde{j}}\mathcal{Z}_{\tilde{j}}(\tau)\right]^2  = \inf_{\tau \in \Upsilon}
\sum_{j \in [J]} \xi_j^2 \left[v_{j}(\tau)^\top \mathcal Z\right]^2>0
\end{align*}
with probability one and  
\begin{align*}
\frac{n \hat{A}_{CR}(\tau)}{r_n^2} & = \left\{ \sum_{j \in [J]} \xi_j^2 \left[G^\top(\tau) \mathcal{Z}_j(\tau) - b_j(\tau) \tilde \Gamma(\tau) \sum_{\tilde{j} \in J}\xi_{\tilde{j}}\mathcal{Z}_{\tilde{j}}(\tau)\right]^2 \right\}^{-1} + o_p(1) \\
& \equiv A_{CR}(\tau) + o_p(1). 
\end{align*}

Similarly, we have
\begin{align*}
(\hat{A}_{CR,g}^*)^{-1}(\tau) & =  \hat G^\top(\tau)\hat{\Omega}_g^*(\tau,\tau)\hat G(\tau)  \\
& = \sum_{j \in [J]} \xi_j n_j  \hat G^\top(\tau) \omega \left[\mathbb{P}_{n,j}\hat f_{\tau,g}^*(D_{i,j}) \right]\left[\mathbb{P}_{n,j}\hat f_{\tau,g}^*(D_{i,j}) \right]^\top\omega^\top \hat G(\tau) .
\end{align*}
Furthermore, Lemmas \ref{lem:expand} and \ref{lem:stud} imply 
\begin{align*}
& r_n  \hat G^\top(\tau) \omega \left[\mathbb{P}_{n,j}\hat f_{\tau,g}^*(D_{i,j}) \right] \\
& = r_n \biggl[g_j G^\top(\tau)\mathbb{P}_{n,j} \tilde{f}_\tau(D_{i,j}, \beta_n(\tau),\gamma_n(\tau),0) - (g_j - \overline{a}_g^*(\tau)) b_j (\tau) (\beta_0(\tau) - \beta_n(\tau)) \\
& - b_j (\tau)\tilde \Gamma(\tau)\sum_{\tilde{j} \in J} g_{\tilde{j}} \xi_{\tilde{j}}  \mathbb{P}_{n,\tilde{j}} \tilde{f}_\tau(D_{i,\tilde{j}}, \beta_n(\tau),\gamma_n(\tau),0)\biggr]+ o_p(1)\\
& =  \left[g_j G^\top(\tau) \mathcal{Z}_j(\tau) - b_j(\tau)\tilde \Gamma(\tau) \sum_{ \tilde{j} \in J} g_{\tilde{j}}\xi_{\tilde{j}} \mathcal{Z}_{\tilde{j}}(\tau)\right] +   (g_j - \overline{a}_g^*(\tau)) b_j(\tau) \mu_\beta(\tau) + o_p(1).
\end{align*}


Then, as $J > 1$,  we have
\begin{align*}
\frac{n \hat{A}_{CR,g}^*(\tau)}{r_n^2} = A_{CR,g}^*(\tau) + o_p(1), 
\end{align*}
where the $o_p(1)$ term holds uniformly over $\tau \in \Upsilon$ and
\begin{align*}
A_{CR,g}^*(\tau) = \left\{ 
\sum_{j \in [J]} \xi_j^2\left[  g_j G^\top(\tau) \mathcal{Z}_j(\tau) - b_j(\tau) \tilde \Gamma(\tau)\sum_{ \tilde{j} \in J} g_{\tilde{j}}\xi_{\tilde{j}} \mathcal{Z}_{\tilde{j}}(\tau) +  (g_j - \overline{a}_g^*(\tau)) b_j(\tau)  \mu_\beta(\tau)\right]^2\right\}^{-1}.
\end{align*}

Let 
\begin{align*}
T_{CR,\infty}(g) = \sup_{\tau \in \Upsilon} \left[\left\Vert \tilde{\Gamma}(\tau) \sum_{j \in [J]} g_j \xi_j \mathcal{Z}_j(\tau) + \overline a_{g}^*(\tau) \mu_\beta(\tau)\right\Vert_{A_{CR,g}^*(\tau)}\right]. 
\end{align*}
Because $\overline{a}_{\iota_J}^*(\tau) = 1$ and $\overline{a}_{-\iota_J}^*(\tau) = -1$, we have,
\begin{align*}
\{ \sqrt{n} T_{CR,n}, \{\sqrt{n} T_{CR,n}^{*}(g)\}_{g \in \textbf{G}} \convD \{T_{CR,\infty}(\iota_J), \{T_{CR,\infty}(g)\}_{g \in \textbf{G}} \}.
\end{align*}
In addition, under the null, $T_{CR,\infty}(g) \stackrel{d}{=}T_{CR,\infty}(g')$ for any $g,g' \in G$ and $T_{CR,\infty}(g) =T_{CR,\infty}(g')$ if and only if $g \in \{g',-g'\}$. Then following the exact same argument in the proof of Theorem \ref{thm:unstudentizedsize}, we have
\begin{align*}
\alpha - \frac{1}{2^{J-1}} \leq \liminf_{n \rightarrow \infty} \mathbb{P} \{ T_{CR,n} > \hat{c}_{CR,n}(1-\alpha) \} \leq \limsup_{n \rightarrow \infty} \mathbb{P} \{ T_{CR,n} > \hat{c}_{CR,n}(1-\alpha) \} \leq \alpha + \frac{1}{2^{J-1}}. 
\end{align*}

For the power analysis, we still have 
\begin{align*}
\liminf_{n \rightarrow \infty}\mathbb{P} \{ T_{CR,n} > \hat{c}_n \} & \geq  \mathbb{P}\left\{ T_{CR,\infty}(\iota_J) > (T_{CR,\infty})^{(k^*)}  \right\}.
\end{align*}
In addition, we aim to show that, as $\sup_{\tau \in \Upsilon}||\mu_\beta(\tau)||_2 \rightarrow \infty$, we have
\begin{align}
\mathbb{P}\left\{ T_{CR,\infty}(\iota_J) > \max_{g \in \textbf{G}_s}T_{CR,\infty}(g)\right\} \rightarrow 1.
\label{eq:GsCR}
\end{align}
Then, given $|\textbf{G}_s| = |\textbf{G}|-2^{J-J_s+1}$ and $\lceil |\textbf{G}|(1-\alpha) \rceil \leq |\textbf{G}|-2^{J-J_s+1}$, \eqref{eq:GsCR} implies,  $\sup_{\tau \in \Upsilon}||\mu_\beta(\tau)||_2 \rightarrow \infty$,
\begin{align*}
\mathbb{P}\left\{ T_{CR,\infty}(\iota_J) > (T_{CR,\infty})^{(k^*)}  \right\} \geq \mathbb{P}\left\{ T_{CR,\infty}(\iota_J) > \max_{g \in \textbf{G}_s}T_{CR,\infty}(g) \right\} \rightarrow 1.
\end{align*}
Therefore, it suffices to establish \eqref{eq:GsCR}. Note that 
$$T_{CR,\infty}(\iota_J) = \sup_{\tau \in \Upsilon} \left[\left\Vert \tilde{\Gamma}(\tau)\left[\sum_{j \in [J]} \xi_j \mathcal{Z}_j(\tau)\right] + \mu_\beta(\tau)\right\Vert_{A_{CR,\iota_J}^*(\tau)}\right],$$
where 
\begin{align*}
A_{CR,\iota_J}^*(\tau) = \left\{\sum_{j \in [J]} \xi_j^2\left[G^\top(\tau) \mathcal{Z}_j(\tau) - b_j(\tau) \tilde{\Gamma}(\tau) \sum_{\tilde{j} \in J}\xi_{\tilde{j}} \mathcal{Z}_{\tilde{j}}(\tau)\right]^2\right\}^{-1}.
\end{align*}
We see that $A_{CR,\iota_J}^*(\tau)$ is independent of $\mu_\beta(\tau)$. For any $\delta>0$, we can find a constant $c>0$ such that with probability greater than $1-\delta$, 
\begin{align*}
\inf_{\tau \in \Upsilon} A_{CR,\iota_J}^*(\tau) \geq c>0,
\end{align*}
and thus, 
\begin{align}
T_{CR,\infty}^2(\iota_J) \geq  c\sup_{\tau \in \Upsilon}\mu_\beta^2(\tau) - O_p(1). 
\label{eq:Tiota}
\end{align}
On the other hand, for $g \in \textbf{G}_s$, we can write  $T_{CR,\infty}(g)$ as
\begin{align*}
T_{CR,\infty}^2(g)&=\sup_{\tau \in \Upsilon} \frac{(N_{0,g}(\tau) + c_{0,g}(\tau)\mu_\beta(\tau))^2}{\sum_{j \in [J]}\xi_j^2 (N_{j,g}(\tau) + c_{j,g}(\tau) \mu_\beta(\tau))^2}, 
\end{align*}
where 
\begin{align*}
& N_{0,g}(\tau) =     \tilde{\Gamma}(\tau)\left[\sum_{j \in [J]} g_j \xi_j \mathcal{Z}_j(\tau)\right], \\
& N_{j,g}(\tau) =    \left[g_j G^\top (\tau) \mathcal{Z}_j(\tau) - b_j(\tau)\tilde{\Gamma}(\tau) \sum_{\tilde{j} \in J}\xi_{\tilde{j}} g_{\tilde{j}}\mathcal{Z}_{\tilde{j}}(\tau)\right],\quad j \in [J],\\
& c_{0,g}(\tau) =  \overline{a}_{g}^*(\tau), \quad \text{and} \quad c_{j,g}(\tau) = (g_j - \overline{a}_g^*(\tau))b_j(\tau), \quad j \in [J].
\end{align*}
We claim that for $g \in \textbf{G}_s$, $c_{j,g}(\tau) \neq 0$ for some $j \in \mathcal J_s$. To see this claim, suppose it does not hold. Then, it implies $g_j = \overline{a}^*_g(\tau)$ for all $j \in \mathcal J_s$, i.e., for all $j \in \mathcal J_s$, $g_j$ shares the same sign. This contradicts with the definition of $\textbf{G}_s$. This claim and the fact that $\inf_{\tau \in \Upsilon, j \in \mathcal J_s} \min(|a_j(\tau)|,|b_j(\tau)|)\geq c_0>0$ further imply that 
\begin{align*}
&\inf_{\tau \in \Upsilon, g \in \textbf{G}_s}\sum_{j \in [J]} \xi_j^2 c_{j,g}^2(\tau) \\
&\geq (c_0^2 \min_{j \in [J]}\xi_j^2) \min((1 -  \overline{a}^*_g(\tau))^2,(1 +  \overline{a}^*_g(\tau))^2) \geq c>0,
\end{align*}
for some constant $c>0$.

In addition, we have
\begin{align*}
& \sum_{j \in [J]}\xi_j^2 (N_{j,g}(\tau) + c_{j,g}(\tau) \mu_\beta(\tau))^2 \\
& = \sum_{j \in [J]}\xi_j^2 N_{j,g}^2(\tau) + \sum_{j \in [J]} 2\xi_j^2 c_{j,g}(\tau)N_{j,g}(\tau) \mu_\beta(\tau)+ (\sum_{j \in [J]} \xi_j^2 c_{j,g}^2(\tau)) \mu_\beta^2(\tau) \\
& = M_1 + 2M_2 \mu_\beta(\tau) + \overline{c}^2 \mu_\beta^2(\tau) \\
& = M_1 -\left(M_2/\overline c\right)^2 + (\overline c \mu_\beta(\tau) + M_2/\overline c)^2,
\end{align*}
where we denote
\begin{align*}
M_1 =     \sum_{j \in [J]}\xi_j^2 N_{j,g}^2(\tau), \quad M_2 = \sum_{j \in [J]} \xi_j^2 c_{j,g}(\tau)N_{j,g}(\tau), \quad  \text{and} \quad \overline{c}^2 = \sum_{j \in [J]} \xi_j^2 c_{j,g}^2(\tau)>0. 
\end{align*}
For notation ease, we suppress the dependence of $(M_1,M_2,\overline{c})$ on $(g,\tau)$. In addition, we note that 
\begin{align*}
M_1 - \frac{M_2^2}{\overline{c}^2} = \sum_{j \in [J]} \xi_j^2  N_{j,g}^2(\tau) - \frac{\left(\sum_{j \in [J]}\xi_j^2  N_{j,g}(\tau) c_{j,g}(\tau)\right)^2}{\sum_{j \in [J]} \xi_j^2 c_{j,g}^2(\tau)} \geq 0,
\end{align*}
where the equal sign holds if and only if there exist $(g,c, \tau) \in \textbf{G}_s \times \Re \times \Upsilon$ such that 
\begin{align*}
N_{j,g}(\tau) = c c_{j,g}(\tau),~ \forall j \in [J], 
\end{align*}
or equivalently, 
\begin{align*}
[v_{1}(\tau),\cdots,v_{J}(\tau)]^\top   \begin{pmatrix}
g_1 \mathcal Z_1(\tau) \\
\vdots \\
g_J \mathcal Z_J(\tau)
\end{pmatrix} = c \begin{pmatrix}
c_{1,g}(\tau) \\
\vdots \\
c_{J,g}(\tau)
\end{pmatrix}.
\end{align*}
Given that the rank of $[v_{1}(\tau),\cdots,v_{J}(\tau)]$ is greater than 1, the RHS is a linear space with rank 1, and $\begin{pmatrix}
g_1 \mathcal Z_1(\tau) \\
\vdots \\
g_J \mathcal Z_J(\tau)
\end{pmatrix}$ is a non-degenerate $(d_\phi J)$ vector of normal random variables, the equality holds with probability zero.



Therefore, $\mathbb{M} \equiv M_1 - \frac{M_2^2 }{\overline{c}^2}$ is invertible with probability one. In addition, denote $\frac{M_2}{\overline{c}} + \overline{c}\mu_\beta(\tau)$ as $\mathbb{V}$. Then,  we have
\begin{align*}
& \left[\sum_{j \in [J]}\xi_j^2 (N_{j,g}(\tau) + c_{j,g}(\tau) \mu_\beta(\tau))^2\right]^{-1} = [\mathbb{M} + \mathbb{V}^2]^{-1} = \mathbb{M}^{-1} - \frac{\mathbb{V}^2\mathbb M^{-2}}{1+ \mathbb{V}^2 \mathbb{M}^{-1}}.
\end{align*}

Next, we note that 
\begin{align*}
N_{0,g}(\tau) + c_{0,g}(\tau)\mu_\beta(\tau) & = N_{0,g}(\tau) + c_{0,g}(\tau)\left(\frac{\mathbb{V}}{\overline{c}} - \frac{M_2}{\overline{c}^2}\right)\\
& \equiv \mathbb{M}_0 + \frac{c_{0,g}(\tau)}{\overline{c}}\mathbb{V},
\end{align*}
where 
\begin{align*}
\mathbb{M}_0 = N_{0,g}(\tau) - \frac{c_{0,g}(\tau)M_2}{\overline{c}^2} = N_{0,g}(\tau) - \frac{c_{0,g}(\tau) (\sum_{j \in [J]}\xi_j^2 c_{j,q}(\tau)N_{j,g}(\tau))}{\sum_{j \in [J]} \xi_j^2 c_{j,g}^2(\tau)}. 
\end{align*}
With these notations, we have
\begin{align*}
& \frac{(N_{0,g}(\tau) + c_{0,g}(\tau)\mu_\beta(\tau))^2}{ \sum_{j \in [J]}\xi_j^2 (N_{j,g}(\tau) + c_{j,g}(\tau) \mu_\beta(\tau))^2} \\
& = \frac{\left(\mathbb{M}_0 + \frac{c_{0,g}(\tau)}{\overline{c}}\mathbb{V} \right)^2 }{\mathbb{M} + \mathbb{V}^2} \\
& \leq 2\mathbb{M}_0^2 \left( \mathbb{M}^{-1} - \frac{\mathbb{V}^2\mathbb M^{-2}}{1+ \mathbb{V}^2 \mathbb{M}^{-1}}\right) +  \frac{2 c_{0,g}^2(\tau) \mathbb V^2}{\overline{c}^2} \left(\mathbb{M}^{-1} - \frac{\mathbb{V}^2\mathbb M^{-2}}{1+ \mathbb{V}^2 \mathbb{M}^{-1}} \right) \\
& \leq 2\mathbb{M}_0^2 \mathbb{M}^{-1} + \frac{2 c_{0,g}^2(\tau)}{\overline{c}^2}\frac{\mathbb{V}^2 \mathbb{M}^{-1} }{1+ \mathbb{V}^2 \mathbb{M}^{-1}} \\
& \leq 2 \frac{(N_{0,g}(\tau) - \frac{c_{0,g}(\tau) (\sum_{j \in [J]}\xi_j^2 c_{j,q}(\tau)N_{j,g}(\tau))}{\sum_{j \in [J]} \xi_j^2 c_{j,g}^2(\tau)})^2}{\sum_{j \in [J]} \xi_j^2  N_{j,g}^2(\tau) - \frac{\left(\sum_{j \in [J]}\xi_j^2  N_{j,g}(\tau) c_{j,g}(\tau)\right)^2}{\sum_{j \in [J]} \xi_j^2 c_{j,g}^2(\tau)} }  + \frac{2c_{0,g}^2(\tau)}{\sum_{j \in [J]} \xi_j^2 c_{j,g}^2(\tau)} \\
& \equiv C(g,\tau).
\end{align*}
By taking the supremum over $\tau \in \Upsilon$ and $g \in \textbf{G}_s $, we have
\begin{align}
\max_{g \in \textbf{G}_s}T_{CR,\infty}(g)  \leq   \sup_{(g,\tau) \in \textbf{G}_s \times \Upsilon}C(g,\tau) = O_p(1).
\label{eq:Tiota*}
\end{align}

Combining \eqref{eq:Tiota} and \eqref{eq:Tiota*}, we have, as $\sup_{\tau \in \Upsilon}||\mu_\beta(\tau)||_2\rightarrow \infty$, 
\begin{align*}
& \liminf_{n \rightarrow \infty}\mathbb{P} \{ T_{CR,n} > \hat{c}_{CR,n}(1-\alpha) \} 
\geq  \mathbb{P}\left\{ T_{CR,\infty}(\iota_J) > (T_{CR,\infty})^{(k^*)}  \right\} 
\geq \mathbb{P}\left\{ T_{CR,\infty}(\iota_J) > \max_{g \in \textbf{G}_s} T_{CR,\infty}(g)  \right\} \\
&  = 1- \mathbb{P}\left\{ T_{CR,\infty}(\iota_J) \leq \max_{g \in \textbf{G}_s} T_{CR,\infty}(g)  \right\} 
\geq 1- \mathbb{P}\left\{ c\sup_{\tau \in \Upsilon}\mu_\beta^2(\tau) - M \leq \sup_{(g,\tau) \in \textbf{G}_s \times \Upsilon}C(g,\tau) \right\}-\delta \rightarrow 1-\delta, 
\end{align*}
where $M>0$ is a sufficiently large constant. 
As $\delta$ is arbitrary, we have established \eqref{eq:GsCR}. This concludes the proof. $\blacksquare$



\section{Proof of Theorem \ref{thm:power_comparison}}
Let $\tilde{c}_{CR,n}(1-\alpha)$ denote the $(1-\alpha)$ quantile
of 
\begin{align*}
\left\{|\hat{\beta}^*_{g}(\tau) - \hat\beta(\tau)|\sqrt{\hat A_{CR}(\tau)}: g \in \textbf{G}\right\},
\end{align*}
i.e., the bootstrap statistic $T_n^*(g)$ studentized by the original CRVE instead of the bootstrap CRVE. 
Because $\Upsilon$ is a singleton, we have
\begin{align*}
1\{T_n > \hat{c}_{n}(1-\alpha)\} = 1\{T_{CR,n} >  \tilde{c}_{CR,n}(1-\alpha)\}.
\end{align*}
Therefore, we have
\begin{align*}
\liminf_{n \rightarrow \infty}\mathbb{P}(\phi^{cr}_n \geq \phi_n) & = \liminf_{n \rightarrow \infty}\mathbb{P}(1\{T_{CR,n} >  \hat {c}_{CR,n}(1-\alpha)\} \geq 1\{T_{CR,n} >  \tilde{c}_{CR,n}(1-\alpha)\}) \\
& = 1 - \limsup_{n \rightarrow \infty}\mathbb{P}(1\{T_{CR,n} >  \hat {c}_{CR,n}(1-\alpha)\} < 1\{T_{CR,n} >  \tilde{c}_{CR,n}(1-\alpha)\}) \\
& \geq 1 - \limsup_{n \rightarrow \infty}\mathbb{P}(\hat {c}_{CR,n}(1-\alpha)>  \tilde{c}_{CR,n}(1-\alpha)) \\
& = 1 - \limsup_{n \rightarrow \infty}\mathbb{P}\left( \sqrt{n} \hat {c}_{CR,n}(1-\alpha)  > \frac{ \tilde T_n^{k^*}}{ \left( \frac{r_n^2}{n} \hat G^\top(\tau) \hat \Omega(\tau,\tau) \hat G (\tau)  \right)^{1/2}}  \right),
\end{align*}
where $k^* = \lceil |\textbf{G}|(1-\alpha) \rceil$ and $\tilde T_n^{k^*}$ is the $k^*$-th order statistic of $\left\{|r_n(\hat{\beta}^*_{g}(\tau) - \hat \beta(\tau))|\right\}_{g \in \textbf{G}}$ in ascending order. We collect all $g \in \textbf{G}$ such that 
\begin{align*}
\tilde T_n^{k^*} \geq |r_n(\hat{\beta}^*_{g}(\tau) - \hat \beta(\tau))|
\end{align*}
and denote it as $\textbf{G}(k^*)$. Then, we have $|\textbf{G}(k^*)| = k^*$ because the probability of ties shrinks to zero. Further recall $\textbf{G}(c_0)$ defined in Theorem \ref{thm:power_comparison}. Because $|\textbf{G}(c_0)| > 2^J - k^*$, it means $\textbf{G}(c_0) \cap \textbf{G}(k^*) \neq \emptyset$. Suppose $g' \in \textbf{G}(c_0) \cap \textbf{G}(k^*)$, then we have
\begin{align*}
\tilde T_n^{k^*} & \geq |r_n(\hat{\beta}^*_{g'}(\tau) - \hat \beta(\tau))| \\
& =  \left|\tilde{\Gamma}(\tau)\left[\sum_{j \in [J]} g_j' \xi_j \mathcal{Z}_j(\tau)\right] + \overline{a}_{g'}^*(\tau)\mu_\beta(\tau) + o_p(1)\right| \\
& \geq c_0 |\mu_\beta(\tau)| - \left|\tilde{\Gamma}(\tau)\left[\sum_{j \in [J]} g_j' \xi_j \mathcal{Z}_j(\tau)\right] + o_p(1)\right|,
\end{align*}
where the equality is due to \eqref{eq:hatbeta^*_g-hatbeta}.

Therefore, we have
\begin{align*}
& \limsup_{n \rightarrow \infty}\mathbb{P}\left( \sqrt{n} \hat {c}_{CR,n}(1-\alpha)  > \frac{ \tilde T_n^{k^*}}{ \left( \frac{r_n^2}{n} \hat G^\top(\tau) \hat \Omega(\tau,\tau) \hat G (\tau)  \right)^{1/2}}  \right) \\
& \leq \limsup_{n \rightarrow \infty}\mathbb{P}\left( \sqrt{n} \hat {c}_{CR,n}(1-\alpha)  > \frac{ c_0 |\mu_\beta(\tau)| - |\tilde{\Gamma}(\tau)\left[\sum_{j \in [J]} g_j' \xi_j \mathcal{Z}_j(\tau)\right] + o_p(1)|}{ \left( \frac{r_n^2}{n} \hat G^\top(\tau) \hat \Omega(\tau,\tau) \hat G (\tau)  \right)^{1/2}}  \right).
\end{align*}

Further note that $\max_{g \in \textbf{G}}|\tilde{\Gamma}(\tau)\left[\sum_{j \in [J]} g_j\xi_j \mathcal{Z}_j(\tau)\right]| = O_P(1)$ and does not depend on $\mu_\beta(\tau)$,
\begin{align*}
\left(\frac{r_n^2}{n} \hat G^\top(\tau) \hat \Omega(\tau,\tau) \hat G (\tau)\right)^{-1}    \convD A_{CR}(\tau),
\end{align*}
and $A_{CR}(\tau)$ does not depend on $\mu_\beta(\tau)$ either. 

Last, although $\hat{c}_{CR,n}(1-\alpha)$ depends on $\mu_\beta(\tau)$, we have 
\begin{align*}
\sqrt{n}\hat{c}_{CR,n}(1-\alpha) \leq \sqrt{n}\max_{g \in \textbf{G}_s} T_{CR,n}^*(g) \convD \max_{g \in \textbf{G}_s} T_{CR,\infty}(g) \leq\sup_{(g,\tau) \in \textbf{G}_s \times \Upsilon}C(g,\tau) = O_P(1)
\end{align*}
for some $C(g)$ that does not depend on $\mu_\beta(\tau)$ as has been proved in the last section. Therefore, for any $\delta>0$, there exists a constant $c_{\mu}>0$, such that when $|\mu_\beta (\tau)| > c_\mu$, 
\begin{align*}
&  \limsup_{n \rightarrow \infty}\mathbb{P}\left( \sqrt{n} \hat {c}_{CR,n}(1-\alpha)  > \frac{ c_0 |\mu_\beta(\tau)| - |\tilde{\Gamma}(\tau)\left[\sum_{j \in [J]} g_j' \xi_j \mathcal{Z}_j(\tau)\right] + o_p(1)|}{ \left( \frac{r_n^2}{n} \hat G^\top(\tau) \hat \Omega(\tau,\tau) \hat G (\tau)  \right)^{1/2}}  \right)\\
& \leq \limsup_{n \rightarrow \infty}\mathbb{P}\left(\left( \frac{r_n^2}{n} \hat G^\top(\tau) \hat \Omega(\tau,\tau) \hat G (\tau)  \right)^{1/2} \sqrt{n}\max_{g \in \textbf{G}_s} T_{CR,n}^*(g)  + C > c_0 c_\mu \right) + \delta/2  \\
& \leq \mathbb P\left(A_{CR}^{-1/2}(\tau) \max_{g \in \textbf{G}_s} T_{CR,\infty}(g)  + C \geq c_0 c_\mu \right)+ \delta/2 \\
& \leq \mathbb P\left(A_{CR}^{-1/2}(\tau) \max_{g \in \textbf{G}_s} C(g)  + C \geq c_0 c_\mu \right)+ \delta/2
\leq \delta,
\end{align*}
where $C$ in the first inequality is a constant such that for $n$ being sufficiently large, 
\begin{align*}
\mathbb P \left(\max_{g \in \textbf{G}}\left|\tilde{\Gamma}(\tau)\left[\sum_{j \in [J]} g_j \xi_j \mathcal{Z}_j(\tau)\right] + o_p(1) \right| \geq C  \right) \leq \delta/2,
\end{align*}
the second inequality is by Portmanteau theorem, and the last inequality holds if $c_\mu$ is sufficiently large. This concludes the proof.
$\blacksquare$

\section{Proof of Theorem \ref{thm:full}}
\label{sec:pf_sec3end}
We focus on the case when Assumption \ref{ass:full}(ii) holds. The proof for the case with Assumption \ref{ass:full}(i) is similar but simpler, and thus, is omitted for brevity. We divide the proof into three steps. In the first step, we derive the limit distribution of $AR_{CR,n}$. In the second step, we derive the limit distribution of $AR_{CR,n}^{*}(g)$. In the third step, we prove the desired result.  Throughout the proof, we impose the null that $\beta_n(\tau) = \beta_0(\tau)$ for $\tau \in \Upsilon$. 

\textbf{Step 1.} By Lemma \ref{lem:basic} with $b_n(\tau) = \beta_n(\tau)=\beta_0(\tau)$, we have
\begin{align}
r_n\hat{\theta}(\beta_0(\tau),\tau) & = \omega \left[\hat F_1(\beta_0(\tau),\tau)\right]^{-1} r_n \left[ I_n(\tau) + II_n(\beta_0(\tau),\tau) + o_p(1) - \hat F_2(\beta_0(\tau),\tau)(\beta_0(\tau) - \beta_n(\tau)) \right] \notag \\
& = \omega \left[Q_{\Psi,\Psi}(\tau)\right]^{-1} r_n  I_n(\tau)  + o_p(1),
\label{eq:thetafull}
\end{align}
where the $o_p(1)$ terms hold uniformly over $\tau \in \Upsilon$ and 
\begin{align*}
I_n(\tau) = (\mathbb{P}_n - \overline{ \mathbb{P}}_n)f_\tau(D,\beta_0(\tau),\gamma_n(\tau),0).
\end{align*}

Therefore, we have
\begin{align}
r_n\hat{\theta}(\beta_0(\tau),\tau) = Q_{\Phi,\Phi}^{-1}(\tau) \sum_{ j \in [J]} \xi_j\mathcal{Z}_j(\tau) + o_p(1),
\label{eq:thetafull'}
\end{align}
where the $o_P(1)$ term holds uniformly over $\tau \in \Upsilon$. In addition, note that
\begin{align*}
\tilde{A}_{CR}^{-1}(\tau) & = \sum_{ j \in [J]} \xi_j n_j \hat H (\tau)\omega \left[ \mathbb{P}_{n,j}\hat{f}_\tau(D_{i,j},\beta_0(\tau),\hat{\gamma}(\beta_0(\tau),\tau))\right] \\
& \times \left[ \mathbb{P}_{n,j}\hat{f}_\tau(D_{i,j},\beta_0(\tau),\hat{\gamma}(\beta_0(\tau),\tau))\right]^\top \omega^\top \hat H(\tau)  
\end{align*}
and
\begin{align*}
& \mathbb{P}_{n,j}\hat{f}_\tau(D_{i,j},\beta_0(\tau),\hat{\gamma}(\beta_0(\tau),\tau),0) \\
& = (\mathbb{P}_{n,j} - \overline{\mathbb{P}}_{n,j})\hat{f}_\tau(D_{i,j},\beta_0(\tau),\hat{\gamma}(\beta_0(\tau),\tau),0) + \overline{\mathbb{P}}_{n,j}\hat{f}_\tau(D_{i,j},\beta_0(\tau),\hat{\gamma}(\beta_0(\tau),\tau),0) \\
& = (\mathbb{P}_{n,j} - \overline{\mathbb{P}}_{n,j})f_\tau(D_{i,j},\beta_0(\tau),\gamma_n(\tau),0) + \overline{\mathbb{P}}_{n,j}\hat{f}_\tau(D_{i,j},\beta_0(\tau),\hat{\gamma}(\beta_0(\tau),\tau),0) + o_p(r_n^{-1})\\
& = \mathbb{P}_{n,j}f_\tau(D_{i,j},\beta_0(\tau),\gamma_n(\tau),0) - Q_{\Psi,\Psi,j}(\tau) \begin{pmatrix}
\hat{\gamma}(\beta_0(\tau),\tau) - \gamma_n(\tau) \\
0
\end{pmatrix} + o_p(r_n^{-1}) \\
& = \mathbb{P}_{n,j}f_\tau(D_{i,j},\beta_0(\tau),\gamma_n(\tau),0) + o_p(r_n^{-1}),
\end{align*}
where the $o_p(r_n^{-1})$ term holds uniformly over $\tau \in \Upsilon$, the second equality is due to Assumption \ref{ass:asym}(ii), the third equality is due to Assumption \ref{ass:asym}(iii) and the fact that 
$$\sup_{\tau \in \Upsilon}||\hat{\gamma}(\beta_0(\tau),\tau) - \gamma_n(\tau)||_2 = o_P(r_n^{-1})$$ as shown in Lemma \ref{lem:expand'}, and the last equality is due to Lemma \ref{lem:expand'}.

Then, we have
\begin{align*}
& \hat H (\tau) \omega \left[ \mathbb{P}_{n,j}\hat{f}_\tau(D_{i,j},\beta_0(\tau),\hat{\gamma}(\beta_0(\tau),\tau))\right] \\
& = \hat H (\tau) \omega \left[ \mathbb{P}_{n,j}f_\tau(D_{i,j},\beta_0(\tau),\gamma_n(\tau),0) - Q_{\Psi,\Psi,j}(\tau) \begin{pmatrix}
\hat{\gamma}(\beta_0(\tau),\tau) - \gamma_n(\tau) \\
0
\end{pmatrix}\right] + o_P(r_n^{-1}) \\
& = H (\tau) \left[ \mathbb{P}_{n,j}\tilde{f}_\tau(D_{i,j},\beta_0(\tau),\gamma_n(\tau),0)\right] + o_P(r_n^{-1}),
\end{align*}
where the $o_p(r_n^{-1})$ term holds uniformly over $\tau \in \Upsilon$ and the second inequality holds because under Assumption \ref{ass:boot}(i),
\begin{align*}
\omega Q_{\Psi,\Psi,j}(\tau) \begin{pmatrix}
\hat{\gamma}(\beta_n(\tau),\tau) - \gamma_n(\tau) \\
0
\end{pmatrix} = 0. 
\end{align*}

Therefore, we have
\begin{align*}
\frac{r_n^2}{n}\hat{A}_{CR}^{-1}(\tau) = \sum_{ j \in [J]}\xi_j^2 H (\tau) \mathcal{Z}_j(\tau)\mathcal{Z}_j^\top(\tau) H(\tau) + o_p(1),
\end{align*}
where the  $o_p(1)$ term hold uniformly over $\tau \in \Upsilon$. When $J>d_\phi$, $\sum_{ j \in [J]}\xi_j^2 H (\tau) \mathcal{Z}_j(\tau)\mathcal{Z}_j^\top(\tau) H(\tau)$ is invertible with probability one, we have
\begin{align}
\frac{n}{r_n^2}\hat{A}_{CR}(\tau) = \left[\sum_{ j \in [J]}\xi_j^2 H (\tau) \mathcal{Z}_j(\tau)\mathcal{Z}_j^\top(\tau) H(\tau)\right]^{-1} + o_p(1),
\label{eq:A3}
\end{align}
where the  $o_p(1)$ term hold uniformly over $\tau \in \Upsilon$. Then, combining  \eqref{eq:thetafull'} and \eqref{eq:A3}, we have
\begin{align*}
n(AR_{CR,n})^2 & = \sup_{\tau \in \Upsilon}\left[\sum_{ j \in [J]} \xi_j \mathcal{Z}_j^\top(\tau)\right]Q_{\Phi,\Phi}^{-1}(\tau) \left[\sum_{ j \in [J]}\xi_j^2 H (\tau) \mathcal{Z}_j(\tau)\mathcal{Z}_j^\top(\tau)H(\tau)\right]^{-1}Q_{\Phi,\Phi}^{-1}(\tau)\left[\sum_{ j \in [J]} \xi_j \mathcal{Z}_j(\tau)\right] \\
& + o_p(1).
\end{align*}

\textbf{Step 2}. Next, we consider the limit distribution of the bootstrap test statistic. By Lemma \ref{lem:bootexpand}, under the null, we have 
\begin{align*}
r_n  \hat{\theta}_g^*(\beta_0(\tau),\tau) & = \omega Q_{\Psi,\Psi}^{-1}(\tau) r_n \sum_{j \in [J]}(1+g_j) \xi_j \mathbb P_{n,j}f_{\tau}(D_{i,j},\beta_n(\tau),\gamma_n(\tau),0)  + o_p(1) \\
& = Q_{\Phi,\Phi}^{-1}(\tau)\sum_{j \in [J]}(1+g_j) \xi_j \mathcal Z_j(\tau) + o_p(1),
\end{align*}
and thus, 
\begin{align*}
r_n  (\hat{\theta}_g^*(\beta_0(\tau),\tau) - \hat{\theta}(\beta_0(\tau),\tau))& = \omega Q_{\Psi,\Psi}^{-1}(\tau) r_n \sum_{j \in [J]} g_j \xi_j \mathbb P_{n,j}f_{\tau}(D_{i,j},\beta_n(\tau),\gamma_n(\tau),0)  + o_p(1) \\
& = Q_{\Phi,\Phi}^{-1}(\tau)\sum_{j \in [J]}g_j \xi_j \mathcal Z_j(\tau) + o_p(1),
\end{align*}
where the $o_P(1)$ term holds uniformly over $\tau \in \Upsilon$.

\textbf{Step 3.} We further define 
\begin{align*}
& AR_{CR,\infty}(g) \\
&  = \sup_{\tau \in \Upsilon}\left\{\left[\sum_{ j \in [J]} \xi_j g_j \mathcal{Z}_j^\top(\tau)\right]Q_{\Phi,\Phi}^{-1}(\tau)\left[\sum_{ j \in [J]}\xi_j^2 H (\tau) \mathcal{Z}_j(\tau)\mathcal{Z}_j^\top(\tau)H(\tau)\right]^{-1} Q_{\Phi,\Phi}^{-1}(\tau)\left[\sum_{ j \in [J]} \xi_j g_j\mathcal{Z}_j(\tau)\right]\right\}^{1/2}.
\end{align*} 
Then, we have, under the null, 
\begin{align*}
(\sqrt{n}AR_{CR,n},\{\sqrt{n}AR_{CR,n}^{*}(g)\}_{g\in G}) \convD (AR_{CR,\infty}(\iota_J),\{AR_{CR,\infty}(g)\}_{g \in \textbf{G}}).
\end{align*}
The distribution of $AR_{CR,\infty}(g)$ is invariant in $g$, and $AR_{CR,\infty}(g) = AR_{CR,\infty}(g')$ if and only if $g \in \{g',-g'\}$. Then, by the same argument in the proofs of Theorem \ref{thm:unstudentizedsize}, we have
\begin{align*}
\alpha - \frac{1}{2^{J-1}} & \leq \liminf_{n \rightarrow \infty} \mathbb{P} \{ AR_{CR,n}> \hat{c}_{AR,CR,n}(1-\alpha)\} \\
&\leq \limsup_{n \rightarrow \infty} \mathbb{P} \{ AR_{CR,n}> \hat{c}_{AR,CR,n}(1-\alpha)\} \leq \alpha + \frac{1}{2^{J-1}}. 
\end{align*}
Similarly, we can show that 
\begin{align*}
\alpha - \frac{1}{2^{J-1}} & \leq \liminf_{n \rightarrow \infty} \mathbb{P} \{ AR_{n}> \hat{c}_{AR,n}(1-\alpha)\} \\
&\leq \limsup_{n \rightarrow \infty} \mathbb{P} \{ AR_{n}> \hat{c}_{AR,n}(1-\alpha)\} \leq \alpha + \frac{1}{2^{J-1}}. \quad  \blacksquare
\end{align*}

\section{Proof of Propositions \ref{prop:boot_full} and \ref{prop:boot}}
\label{sec:proj}
For Proposition \ref{prop:boot_full}, we have
\begin{align*}
& \sup\left\Vert \overline{\mathbb{P}}_{n,j} f_{\eps_{i,j}(\tau)}(\delta_{i,j}(v,\tau)|W_{i,j},Z_{i,j})V_{i,j}(\tau)W_{i,j}\hat{\Phi}_{i,j}^\top(\tau) \right\Vert_{op} \\
& \leq \sup\left\Vert \overline{\mathbb{P}}_{n,j} f_{\eps_{i,j}(\tau)}(\delta_{i,j}(v,\tau)|W_{i,j},Z_{i,j})V_{i,j}(\tau)W_{i,j}\hat Z_{i,j}^\top-Q_{W,Z,j}(\tau) \right\Vert_{op} \\
& +  \sup\left\Vert \left[\overline{\mathbb{P}}_{n,j} f_{\eps_{i,j}(\tau)}(\delta_{i,j}(v,\tau)|W_{i,j},Z_{i,j})V_{i,j}(\tau)W_{i,j}W_{i,j}^\top(\tau) - Q_{W,W,j}(\tau)\right]\hat{\chi}(\tau) \right\Vert_{op}\\
& +\sup_{\tau \in \Upsilon}\left\Vert Q_{W,W,j}(\tau)(\hat{\chi}(\tau)-\chi(\tau))\right\Vert_{op} + \sup_{\tau \in \Upsilon}||Q_{W,W,j}(\tau)\chi(\tau) - Q_{W,Z,j}(\tau)||_{op}
\end{align*}
where the suprema in the first three lines are taken over $\{j \in J,||v||_2 \leq \delta,\tau \in \Upsilon\}$ for $v = (v_b^\top,v_r^\top,v_t^\top)^\top$. We note that, based on Assumption \ref{ass:boot'}, by letting $n \rightarrow \infty$ followed by $\delta \rightarrow 0$,  
\begin{align*}
& \sup\left\Vert \overline{\mathbb{P}}_{n,j} f_{\eps_{i,j}(\tau)}(\delta_{i,j}(v,\tau)|W_{i,j},Z_{i,j})V_{i,j}(\tau)W_{i,j}\hat Z_{i,j}^\top-Q_{W,Z,j}(\tau) \right\Vert_{op} \convP 0 \quad \text{and}\\
& \sup\left\Vert \left[\overline{\mathbb{P}}_{n,j} f_{\eps_{i,j}(\tau)}(\delta_{i,j}(v,\tau)|W_{i,j},Z_{i,j})V_{i,j}(\tau)W_{i,j}W_{i,j}^\top(\tau) - Q_{W,W,j}(\tau)\right]\hat{\chi}(\tau) \right\Vert_{op}\convP 0.
\end{align*}
In addition, Assumption \ref{ass:boot'} implies 
$\sup_{\tau \in \Upsilon}\left\Vert \hat{\chi}(\tau)-\chi(\tau)\right\Vert = o_p(1)$. Therefore, in order to show the result, it suffices to show 
$$\sup_{\tau \in \Upsilon}||Q_{W,Z,j}(\tau)- Q_{W,W,j}(\tau)\chi(\tau)||_{op} = o(1).$$
We note that 
\begin{align*}
& \sup_{\tau \in \Upsilon}||Q_{W,Z,j}(\tau)- Q_{W,W,j}(\tau)\chi(\tau)||_{op} \\
& = \sup_{\tau \in \Upsilon}|| Q_{W,W,j}(\tau)(\chi_{n,j}(\tau)-\chi(\tau))||_{op} \\
& = \sup_{\tau \in \Upsilon}||\lim_{n\rightarrow \infty}\overline{\mathbb{P}}_{n,j}f_{\eps_{i,j}(\tau)}(0|W_{i,j},Z_{i,j})W_{i,j} W_{i,j}^\top(\chi_{n,j}(\tau)-\chi(\tau))||_{op} \\
& \leq \limsup_{n\rightarrow \infty}\sup_{\tau \in \Upsilon}\left[\overline{\mathbb{P}}_{n,j}||f_{\eps_{i,j}(\tau)}(0|W_{i,j},Z_{i,j})W_{i,j}||_2^2\overline{\mathbb{P}}_{n,j} ||W_{i,j}^\top(\chi_{n,j}(\tau)-\chi(\tau))||_{op}^2  \right]^{1/2} = o(1). 
\end{align*}
\vspace{0.1in}

For Proposition \ref{prop:boot}, we note that
\begin{align*}
& \sup\left\Vert \overline{\mathbb{P}}_{n,j} f_{\eps_{i,j}(\tau)}(\delta_{i,j}(v,\tau)|W_{i,j},Z_{i,j})V_{i,j}(\tau)W_{i,j}\hat{\Phi}_{i,j}^\top(\tau) \right\Vert_{op} \\
& \leq \sup\left\Vert \overline{\mathbb{P}}_{n,j} f_{\eps_{i,j}(\tau)}(\delta_{i,j}(v,\tau)|W_{i,j},Z_{i,j})V_{i,j}(\tau)W_{i,j}\hat Z_{i,j}^\top-Q_{W,Z,j}(\tau) \right\Vert_{op} \\
& +  \sup\left\Vert \left[\overline{\mathbb{P}}_{n,j} f_{\eps_{i,j}(\tau)}(\delta_{i,j}(v,\tau)|W_{i,j},Z_{i,j})V_{i,j}(\tau)W_{i,j}W_{i,j}^\top(\tau) - Q_{W,W,j}(\tau)\right]\hat{\chi}_j(\tau) \right\Vert_{op}\\
& + \sup_{\tau \in \Upsilon}||\hat{\underline Q}_{W,W,j}(\tau)\hat{\chi}_j(\tau) - \hat{ \underline{Q}}_{W,Z,j}(\tau)||_{op}  +  \sup_{\tau \in \Upsilon}|| Q_{W,Z,j}(\tau)-\hat{ \underline{Q}}_{W,Z,j}(\tau)||_{op} \\
& +\sup_{\tau \in \Upsilon}\left\Vert \left[Q_{W,W,j}(\tau) - \underline{\hat{Q}}_{W,W,j}(\tau)\right]\hat{\chi}_j(\tau)  \right\Vert_{op} 
\end{align*}
where the suprema in the first three lines are taken over $\{j \in [J],||v||_2 \leq \delta,\tau \in \Upsilon\}$ for $v = (v_b^\top,v_r^\top,v_t^\top)^\top$. We note that, by Assumption \ref{ass:boot''}, by letting $n \rightarrow \infty$ followed by $\delta \rightarrow 0$,  
\begin{align*}
& \sup\left\Vert \overline{\mathbb{P}}_{n,j} f_{\eps_{i,j}(\tau)}(\delta_{i,j}(v,\tau)|W_{i,j},Z_{i,j})V_{i,j}(\tau)W_{i,j}\hat Z_{i,j}^\top-Q_{W,Z,j}(\tau) \right\Vert_{op} \convP 0 \quad \text{and}\\
& \sup\left\Vert \left[\overline{\mathbb{P}}_{n,j} f_{\eps_{i,j}(\tau)}(\delta_{i,j}(v,\tau)|W_{i,j},Z_{i,j})V_{i,j}(\tau)W_{i,j}W_{i,j}^\top(\tau) - Q_{W,W,j}(\tau)\right]\hat{\chi}_j(\tau) \right\Vert_{op}\convP 0.
\end{align*}
Assumption \ref{ass:boot''} also implies 
\begin{align*}
\sup_{\tau \in \Upsilon}|| Q_{W,Z,j}(\tau)-\hat{ \underline{Q}}_{W,Z,j}(\tau)||_{op} + \sup_{\tau \in \Upsilon}\left\Vert Q_{W,W,j}(\tau) - \underline{\hat{Q}}_{W,W,j}(\tau)\right\Vert_{op}  = o_p(1).
\end{align*}
Therefore, in order to show the result, it suffices to show 
\begin{align}
\sup_{\tau \in \Upsilon}||\hat{\chi}_j(\tau)||_{op} = O_p(1)
\label{eq:p1}
\end{align}
and 
\begin{align}
\underline{\hat{Q}}_{W,W,j}(\tau)\hat{\chi}_j(\tau) - \hat{ \underline{Q}}_{W,Z,j}(\tau)=0.
\label{eq:p2}
\end{align}

To see \eqref{eq:p1}, we note that Assumption \ref{ass:boot''}(ii) implies the generalized  inverse is continuous at $Q_{W,W,j}(\tau)$ uniformly over $\tau \in \Upsilon$. Therefore, by the continuous mapping theorem, we have
\begin{align*}
\sup_{\tau \in \Upsilon
}||\underline{\hat{Q}}^-_{W,W,j}(\tau)-Q_{W,W,j}^-(\tau)||_{op} = o_p(1),
\end{align*}
and thus
\begin{align*}
\sup_{\tau \in \Upsilon}||\hat{ \chi}_j(\tau)||_{op} & \leq \sup_{\tau \in \Upsilon}||\hat{ \chi}_j(\tau) - Q_{W,W,j}(\tau)Q_{W,W,j}^-(\tau)Q_{W,W,j}^-(\tau)Q_{W,Z,j}(\tau)||_{op} \\
& + \sup_{\tau \in \Upsilon}||Q_{W,W,j}(\tau)Q_{W,W,j}^-(\tau)Q_{W,W,j}^-(\tau)Q_{W,Z,j}(\tau)||_{op}= O_p(1). 
\end{align*}

To show \eqref{eq:p2}, we define 
$$\mathcal{X}_j = \left(W_{1,j}\hat V_{1,j}^{1/2}(\tau)K^{1/2}\left( \frac{\hat{\underline{\eps}}_{i,j}(\tau) }{h_{3,j}}\right)h_{3,j}^{-1/2},\cdots, W_{n_j,j} \hat V_{n_j,j}^{1/2}(\tau) K^{1/2}\left( \frac{\hat{\underline{\eps}}_{n_j,j}(\tau)}{h_{3,j}}\right)h_{3,j}^{-1/2}\right)^\top$$ 
and 
$$\mathcal Z_j = \left(Z_{1,j} \hat V_{1,j}^{1/2}(\tau) K^{1/2}\left( \frac{\hat{\underline{\eps}}_{i,j}(\tau)}{h_{4,j}}\right)h_{4,j}^{-1/2},\cdots,Z_{n_j,j} \hat V_{n_j,j}^{1/2}(\tau) K^{1/2}\left( \frac{\hat{\underline{\eps}}_{n_j,j}(\tau)}{h_{4,j}}\right)h_{4,j}^{-1/2}\right)^\top.$$ 
Then, we have
\begin{align*}
\underline{\hat{Q}}_{W,W,j}(\tau) = \frac{1}{n_j}\mathcal{X}_j^\top\mathcal{X}_j \quad \text{and} \quad  \underline{\hat{Q}}_{W,Z,j}(\tau) = \frac{1}{n_j}\mathcal{X}_j^\top\mathcal{Z}_j. 
\end{align*}
Define the singular value decomposition of $\mathcal{X}_j$ as $\mathcal{X}_j = U_j^\top \Sigma_j V_j$, where $U_j$ and $V_j$ are $n_j \times n_j$ and $d_w \times d_w$ orthonormal matrices and $\Sigma_j$ is a $n_j \times d_w$ matrix with the first $R$ diagonal elements being positive and all the rest entries in the matrix being zero. Then, we have
\begin{align*}
\underline{\hat{Q}}_{W,W,j}(\tau)\hat{\chi}_j(\tau) & = \frac{1}{n_j}\mathcal{X}_j^T \mathcal{X}_j\mathcal{X}_j^T \mathcal{X}_j\left(\mathcal{X}_j^T \mathcal{X}_j\right)^-\left(\mathcal{X}_j^T \mathcal{X}_j\right)^- \mathcal{X}_j^\top \mathcal{Z}_j  \\
& = \frac{1}{n_j}V_j^\top (\Sigma_j^\top \Sigma_j)(\Sigma_j^\top \Sigma_j)  (\Sigma_j^\top \Sigma_j)^-  (\Sigma_j^\top \Sigma_j)^- \Sigma_j^\top U_j \mathcal{Z}_j\\
& = \frac{1}{n_j}V_j^\top \begin{pmatrix}
I_R & 0_{R \times (d_w-R)} \\
0_{(d_w-R) \times R} & 0_{(d_w-R) \times (d_w-R)}
\end{pmatrix} \Sigma_j^\top U_j \mathcal{Z}_j \\
& = \frac{1}{n_j}V_j^\top \Sigma_j^\top U_j \mathcal{Z}_j \\
& = \underline{\hat{Q}}_{W,Z,j}(\tau),
\end{align*}
where we use the fact that 
\begin{align*}
\begin{pmatrix}
I_R & 0_{R \times (d_w-R)} \\
0_{(d_w-R) \times R} & 0_{(d_w-R) \times (d_w-R)}
\end{pmatrix} \Sigma_j^\top = \Sigma_j^\top.
\end{align*}
This concludes the proof. $\blacksquare$

\section{Technical Lemmas used in the Proofs of Results in Section \ref{sec:IVQR}}
\label{sec:lemma_sec3}
\subsection{Linear Expansion of $\hat{\gamma}(b_n(\tau),\tau)$}
\label{sec:pf_lem3}
\begin{lemma}
	Let $\mathcal{B}(\delta) = \{b(\cdot) \in \ell^\infty(\Upsilon): \sup_{\tau \in \Upsilon}||b(\tau)-\beta_n(\tau)||_2 \leq \delta\}$. Suppose Assumptions \ref{ass:id} and \ref{ass:asym} hold. Let $b_n(\tau)$ be a generic point in $\mathcal{B}(\delta)$. Then, for any $\eps>0$ and $\eps'>0$, there exist constants $\overline{ \delta}>0$ and $c>0$ that are independent of $(n,\delta,\eps, \overline{ \delta})$ such that for $\delta \leq \overline{ \delta}$, with probability greater than $1-c\eps$, 
	\begin{align}
	& \begin{pmatrix}
	r_n\left(\hat{\gamma}(b_n(\tau),\tau) -\gamma_n(\tau)\right) \\
	r_n\hat{\theta}(b_n(\tau),\tau)
	\end{pmatrix} \notag \\
	& = \left[\hat F_1(b_n(\tau),\tau)\right]^{-1} r_n \left[ I_n(\tau) + II_n(b_n(\tau),\tau) + o_p(1/r_n) - \hat F_2(b_n(\tau),\tau)(b_n(\tau) - \beta_n(\tau)) \right],
	\label{eq:betagammarate}
	\end{align}
	where the $o_p(1/r_n)$ term on the RHS of the above display holds uniformly over $\tau \in \Upsilon,b_n(\cdot) \in B(\delta)$, 
	\begin{align*}
	I_n(\tau) = (\mathbb{P}_n - \overline{ \mathbb{P}}_n)f_\tau(D,\beta_n(\tau),\gamma_n(\tau),0), \quad \sup_{\tau \in \Upsilon}||r_nI_n(\tau)||_2 = O_p(1),
	\end{align*}
	\begin{align*}
	II_n(b_n(\tau),\tau) = (\mathbb{P}_n - \overline{ \mathbb{P}}_n)\left(\hat{f}_\tau(D,b_n(\tau),\hat{\gamma}(b_n(\tau),\tau),\hat{\theta}_n(b_n(\tau),\tau))-f_\tau(D,\beta_n(\tau),\gamma_n(\tau),0)\right), 
	\end{align*}
	\begin{align*}
	\sup_{b_n(\cdot) \in B(\delta),\tau \in \Upsilon}||r_nII_n(b_n(\tau),\tau)||_2 \leq \eps',
	\end{align*}
	\begin{align*}
	&	\hat F_1(b_n(\tau),\tau) =  \overline{\mathbb{P}}_{n}f_{\eps_{i,j}(\tau)}(\hat{\delta}_{i,j}(\tau)|W_{i,j},Z_{i,j})\hat{\Psi}_{i,j}(\tau)\hat{\Psi}_{i,j}^\top(\tau)\hat V_{i,j}(\tau),\\
	& \sup_{b_n(\cdot) \in B(\delta),\tau \in \Upsilon}||\hat F_1(b_n(\tau),\tau) - Q_{\Psi,\Psi}(\tau)||_{op} \leq \eps',
	\end{align*}
	\begin{align*}
	&	\hat F_2(b_n(\tau),\tau) = \overline{\mathbb{P}}_{n}f_{\eps_{i,j}(\tau)}(\hat{\delta}_{i,j}(\tau)|W_{i,j},Z_{i,j})\hat{\Psi}_{i,j}(\tau)X_{i,j} \hat V_{i,j}(\tau),\\
	& \sup_{b_n(\cdot) \in B(\delta),\tau \in \Upsilon}||\hat F_2(b_n(\tau),\tau) - Q_{\Psi,X}(\tau)||_{op} \leq \eps',
	\end{align*}
	and 
	$\hat{\delta}_{i,j}(\tau) \in (0, X_{i,j}(b_n(\tau) - \beta_n(\tau)) + W_{i,j}^\top (\hat{\gamma}(b_n(\tau),\tau)-\gamma_n(\tau)) + \hat{\Phi}_{i,j}^\top(\tau) \hat{\theta}(b_n(\tau),\tau))$.
	\label{lem:basic}
\end{lemma}
\begin{proof}
	By Assumption \ref{ass:id}, the sub-gradient condition for $(\hat{\gamma}(b_n(\tau),\tau),\hat{\theta}(b_n(\tau),\tau))$ implies 
	\begin{align}
	o_p(1/r_n) & = \mathbb{P}_n\hat{f}_{\tau}(D_{i,j},b_n(\tau), \hat{\gamma}(b_n(\tau),\tau),\hat{\theta}(b_n(\tau),\tau)) \notag \\
	& = (\mathbb{P}_{n}-\overline{\mathbb{P}}_n)\hat{f}_{\tau}(D_{i,j},b_n(\tau), \hat{\gamma}(b_n(\tau),\tau),\hat{\theta}(b_n(\tau),\tau)) + \overline{\mathbb{P}}_n\hat{f}_{\tau}(D_{i,j},b_n(\tau), \hat{\gamma}(b_n(\tau),\tau),\hat{\theta}(b_n(\tau),\tau)) \notag \\
	& = (\mathbb{P}_{n}-\overline{\mathbb{P}}_n)f_{\tau}(D_{i,j},\beta_n(\tau),\gamma_n(\tau),0)  \notag \\
	& + (\mathbb{P}_{n}-\overline{\mathbb{P}}_n)\left(\hat{f}_{\tau}(D_{i,j},b_n(\tau), \hat{\gamma}(b_n(\tau),\tau),\hat{\theta}(b_n(\tau),\tau)) - f_{\tau}(D_{i,j},\beta_n(\tau),\gamma_n(\tau),0)  \right) \notag \\
	& +  \overline{\mathbb{P}}_n\hat{f}_{\tau}(D_{i,j},b_n(\tau), \hat{\gamma}(b_n(\tau),\tau),\hat{\theta}(b_n(\tau),\tau)) \notag \\
	& = I_n(\tau) + II_n(b_n(\tau),\tau) + III_n(b_n(\tau),\tau),
	\label{eq:expand1}
	\end{align}
	where the $o_p(1/r_n)$  term on the LHS of the above display holds uniformly over $\tau \in \Upsilon,b_n(\cdot) \in B(\delta)$. For the first term, we note that $\overline{\mathbb{P}}_nf_{\tau}(D_{i,j},\beta_n(\tau),\gamma_n(\tau),0)=0$. Then, by Assumption \ref{ass:asym}(iv), we have, 
	\begin{align*}
	\sup_{\tau \in \Upsilon} ||r_n I_n(\tau)||_2 = O_p(1).
	\end{align*}
	
	For the second term, note
	\begin{align*}
	& \sup_{\tau \in \Upsilon}||\hat{\gamma}(b_n(\tau),\tau) -\gamma_n(\tau)||_2 \\
	& \leq \sup_{\tau \in \Upsilon,b \in \mathcal{B}}||\hat{\gamma}(b,\tau) -\gamma_n(b,\tau)||_2  + \sup_{\tau \in \Upsilon, b ,b' \in \mathcal{B},||b-b'||_2 \leq \delta}||\gamma_n(b,\tau) -\gamma_n(b',\tau)||_2 \\
	& \leq \sup_{\tau \in \Upsilon,b \in \mathcal{B}}||\hat{\gamma}(b,\tau) -\gamma_n(b,\tau)||_2 + \sup_{\tau \in \Upsilon, b ,b' \in \mathcal{B},||b-b'||_2 \leq \delta}||\gamma_\infty(b,\tau) -\gamma_\infty(b',\tau)||_2 \\
	&  + 2 \sup_{\tau \in \Upsilon,b \in \mathcal{B}}||\gamma_\infty(b,\tau) -\gamma_n(b,\tau)||_2.
	\end{align*}
	
	By \citet[Theorem 1]{K09} and the fact that both $\hat{\mathcal Q}_n(b,r,t,\tau)$ and $\mathcal Q_n(b,r,t,\tau)$ are convex, we have 
	\begin{align*}
	\sup_{(b,\tau) \in \mathcal{B} \times \Upsilon}\left(||\hat{\gamma}(b,\tau)-\gamma_n(b,\tau)||_2 + ||\hat{\theta}(b,\tau)-\theta_n(b,\tau)||_2 \right) = o_p(1) 
	\end{align*}	
	and
	\begin{align*}
	\sup_{(b,\tau) \in \mathcal{B} \times \Upsilon}\left(||\gamma_\infty(b,\tau)-\gamma_n(b,\tau)||_2 + ||\theta_\infty(b,\tau)-\theta_n(b,\tau)||_2 \right) = o(1).
	\end{align*}
	In addition, by Assumption \ref{ass:asym}(i),  $\gamma_\infty(b,\tau)$ is continuous in $b \in \mathcal B$ uniformly over $\tau \in \Upsilon$. Therefore, for any $\delta'>0$ and $\eps>0$, there exist $\underline{n}$ and $\overline{\delta}$ such that for $n \geq \underline{n}$ and $\delta' \leq \overline{\delta}$, with probability greater than $1-\eps$, 
	\begin{align*}
	\sup_{\tau \in \Upsilon}||\hat{\gamma}(b_n(\tau),\tau) -\gamma_n(\tau)||_2 \leq \delta'.
	\end{align*}
	Similarly, we have
	\begin{align*}
	\sup_{\tau \in \Upsilon}||\hat{\theta}(b_n(\tau),\tau) - 0||_2 \leq \delta'.
	\end{align*}
	
	Then,  for any $\eps>0$, there exist $\underline{n}$ and $\overline{\delta}$ such that for $n \geq \underline{n}$ and $\delta',\delta \leq \overline{\delta}$, with probability greater than $1-\eps$, we have,
	\begin{align*}
	&\sup_{b_n(\cdot) \in B(\delta),\tau \in \Upsilon}  ||r_n II_n(b_n(\tau),\tau)||_2 \\
	& \leq \sup_{||v||_2  \leq 2\delta'+\delta, \tau \in \Upsilon} \biggl\Vert r_n(\mathbb{P}_{n}-\overline{\mathbb{P}}_n)\biggl(\hat{f}_{\tau}(D_{i,j},\beta_n(\tau)+v_b,\gamma_n(\tau)+v_r,v_t) - f_{\tau}(D_{i,j},\beta_n(\tau),\gamma_n(\tau),0)  \biggr)\biggr\Vert_2,
	\end{align*}
	where $v = (v_b^\top, v_r^\top, v_t^\top)^\top$. 
	Then, by Assumption \ref{ass:asym}(ii), for any $\eps>0$, there exist $\underline{n}$ and $\overline{\delta}$ such that for $n \geq \underline{n}$ and $\delta,\delta' \leq \overline{\delta}$, we have, with probability greater than $1-2\eps$, 
	\begin{align*}
	\sup_{b_n(\cdot) \in B(\delta),\tau \in \Upsilon} ||r_n II_n(b_n(\tau),\tau)||_2 \leq \eps'.
	\end{align*}
	
	For the third term in \eqref{eq:expand1}, we have 
	\begin{align}
	& \overline{\mathbb{P}}_n\hat{f}_{\tau}(D_{i,j},b_n(\tau), \hat{\gamma}(b_n(\tau),\tau),\hat{\theta}(b_n(\tau),\tau)) \notag \\
	& = \sum_{j \in [J]} \xi_j \overline{\mathbb{P}}_{n,j}(\tau - 1\{y_{i,j} - X_{i,j} b_n(\tau) - W_{i,j}^\top \hat{\gamma}(b_n(\tau),\tau) - \hat{\Phi}_{i,j}^\top(\tau) \hat{\theta}(b_n(\tau),\tau) \leq 0 \})\hat{\Psi}_{i,j}(\tau) \hat V_{i,j}(\tau)\notag \\
	& =-\overline{\mathbb{P}}_{n,j}f_{\eps_{i,j}(\tau)}(\hat{\delta}_{i,j}(\tau)|W_{i,j},Z_{i,j})\hat{\Psi}_{i,j}(\tau)X_{i,j}(b_n(\tau) - \beta_n(\tau)) \hat V_{i,j}(\tau) \notag \\
	& -  \overline{\mathbb{P}}_{n,j}f_{\eps_{i,j}(\tau)}(\hat{\delta}_{i,j}(\tau)|W_{i,j},Z_{i,j})\hat{\Psi}_{i,j}(\tau)\hat{\Psi}_{i,j}^\top(\tau) \hat V_{i,j}(\tau)
	\begin{pmatrix}
	\hat{\gamma}(b_n(\tau),\tau) -\gamma_n(\tau)\\
	\hat{\theta}(b_n(\tau),\tau)
	\end{pmatrix},
	\label{eq:Pjf}
	\end{align}
	where $\hat{\delta}_{i,j}(\tau) \in (0, X_{i,j}(b_n(\tau) - \beta_n(\tau)) + W_{i,j}^\top (\hat{\gamma}(b_n(\tau),\tau)-\gamma_n(\tau)) + \hat{\Phi}_{i,j}^\top(\tau) \hat{\theta}(b_n(\tau),\tau))$.  For any $\eps>0$, there exist $\underline{n}$ and $\overline{\delta}$ such that for $n \geq \underline{n}$ and $\delta,\delta' \leq \overline{\delta}$, we have, with probability greater than $1-\eps$, 
	\begin{align*}
	\sup_{\tau \in \Upsilon}\left( ||b_n(\tau) - \beta_n(\tau)||_2 + ||\hat{\gamma}(b_n(\tau),\tau)-\gamma_n(\tau)||_2 + ||\hat{\theta}(b_n(\tau),\tau)||_2\right) \leq \delta+2\delta'.
	\end{align*}
	This implies, with probability greater than $1-\eps$, 
	\begin{align*}
	& \sup_{b_n(\cdot) \in B(\delta),\tau \in \Upsilon,j \in [J]}||\hat F_2(b_n(\tau),\tau) - Q_{\Psi,X,j}(\tau)||_{op} \\
	& \leq \left(\max_{j \in [J]} \xi_j\right) \sup_{b_n(\cdot) \in B(\delta),\tau \in \Upsilon,j \in [J]}\left\Vert\overline{\mathbb{P}}_{n,j} f_{\eps_{i,j}(\tau)}(\hat{\delta}_{i,j}(\tau)|W_{i,j},Z_{i,j})\hat{\Psi}_{i,j}(\tau)X_{i,j}\hat V_{i,j}(\tau) - Q_{\Psi,X,j}(\tau) \right\Vert_{op} \\
	& \leq \left(\max_{j \in [J]} \xi_j\right) \sup \left\Vert \overline{\mathbb{P}}_{n,j}f_{\eps_{i,j}(\tau)}(\delta_{i,j}(v,\tau)|W_{i,j},Z_{i,j})\hat{\Psi}_{i,j}(\tau)X_{i,j} \hat V_{i,j}(\tau) - Q_{\Psi,X,j}(\tau) \right\Vert_{op},
	\end{align*}
	where the supremum in the second inequality is taken over $\{(j,v,\tau): j \in [J], ||v||_2 \leq  \delta+2\delta', \tau \in \Upsilon\}$, $v = (v_b^\top,v_r^\top,v_t^\top)^\top$, and 
	\begin{align*}
	\delta_{ij}(v,\tau) = X_{i,j} v_b + W_{i,j}^\top v_r + \hat \Phi_{i,j}^\top v_t.
	\end{align*}
	
	Then, Assumption \ref{ass:asym}(iii) implies, with probability greater than $1-2\eps$, 
	\begin{align*}
	\sup_{b_n(\cdot) \in B(\delta),\tau \in \Upsilon}||\hat F_2(b_n(\tau),\tau) - Q_{\Psi,X}(\tau)||_{op} \leq \eps'. 
	\end{align*}
	Similarly, we have, with probability greater than $1-2\eps$, 
	\begin{align*}
	\sup_{b_n(\cdot) \in B(\delta),\tau \in \Upsilon}||\hat F_1(b_n(\tau),\tau) - Q_{\Psi,\Psi}(\tau)||_{op} \leq \eps'. 
	\end{align*}

	Then, by Assumption \ref{ass:asym} and the fact that $\eps$ can be made arbitrarily small, we have, with probability greater than $1-2\eps$, $\hat F_1(b_n(\tau),\tau) $ is invertible. Therefore, \eqref{eq:expand1} implies, with probability greater than $1-2\eps$, 
	\begin{align*}
	& \begin{pmatrix}
	r_n \left(\hat{\gamma}(b_n(\tau),\tau) -\gamma_n(\tau)\right) \\
	r_n \hat{\theta}(b_n(\tau),\tau)
	\end{pmatrix} \notag \\
	& = \left[\hat F_1(b_n(\tau),\tau)\right]^{-1} r_n \left[ I_n(\tau) + II_n(b_n(\tau),\tau) + o_p(1/r_n)  - \hat F_2(b_n(\tau),\tau)(b_n(\tau) - \beta_n(\tau)) \right]
	\end{align*}
	where the $o_p(1/r_n) $ term on the RHS of the above display holds uniformly over $\tau \in \Upsilon,b_n(\cdot) \in B(\delta)$. 
\end{proof}

\subsection{Technical Results for the IVQR Estimator}
\begin{lemma}
	Suppose Assumptions \ref{ass:id}, \ref{ass:asym}, and \ref{ass:id2} hold. Then, 
	\begin{align*}
	& \begin{pmatrix}
	r_n \left(\hat{\beta}(\tau) - \beta_n(\tau)\right) \\
	r_n\left(\hat{\gamma}(\tau) -\gamma_n(\tau)\right) \\
	r_n\hat{\theta}(\tau)
	\end{pmatrix}\\
	& = \begin{pmatrix}
	&	\left[Q_{\Psi,X}^\top(\tau) Q_{\Psi,\Psi}^{-1}(\tau) \omega^\top A_1(\tau)  \omega Q^{-1}_{\Psi,\Psi}(\tau) Q_{\Psi,X}(\tau)\right]^{-1}Q_{\Psi,X}^\top(\tau) Q_{\Psi,\Psi}^{-1}(\tau) \omega^\top A_1(\tau) \omega Q^{-1}_{\Psi,\Psi}(\tau) \\
	&	Q_{\Psi,\Psi}^{-1}(\tau)\biggl[ \mathbb{I}_{d_w+d_\phi} -  Q_{\Psi,X}(\tau)\left[Q_{\Psi,X}^\top(\tau) Q_{\Psi,\Psi}^{-1}(\tau) \omega^\top A_1(\tau)  \omega Q^{-1}_{\Psi,\Psi}(\tau) Q_{\Psi,X}(\tau)\right]^{-1} \\
	& \times Q_{\Psi,X}^\top(\tau) Q_{\Psi,\Psi}^{-1}(\tau) \omega^\top A_1(\tau) \omega Q^{-1}_{\Psi,\Psi}(\tau) \biggr] 
	\end{pmatrix}\\
	& \times r_n\mathbb{P}_n f_{\tau}(D_{i,j},\beta_n(\tau),\gamma_n(\tau),0) + o_p(1),
	\end{align*}
	where $o_p(1)$ term holds uniformly over $\tau \in \Upsilon$. 
	\label{lem:expand}
\end{lemma}

\begin{proof}
	We divide the proof into three steps. In the first step, we show $(\hat{\beta}(\tau),\hat{\gamma}(\tau),\hat{\theta}(\tau))$ are consistent. In the second step, we derive convergence rates of $(\hat{\beta}(\tau),\hat{\gamma}(\tau),\hat{\theta}(\tau))$. In the third step, we derive linear expansions for $(\hat{\beta}(\tau),\hat{\gamma}(\tau),\hat{\theta}(\tau))$.

	\textbf{Step 1.} We first show the consistency of $(\hat{\beta}(\tau),\hat{\gamma}(\tau),\hat{\theta}(\tau))$. Note by construction, we have $\gamma(\beta_n(\tau),\tau) = \gamma_n(\tau)$, $\theta_n(\beta_n(\tau),\tau) = 0$, $\hat{\gamma}(\tau) = \hat{\gamma}(\hat{\beta}(\tau),\tau)$, and $\hat{\theta}(\tau) = \hat{\theta}(\hat{\beta}(\tau),\tau)$. By \citet[Theorem 1]{K09} and the fact that both $\hat{\mathcal Q}_n(b,r,t,\tau)$ and $\mathcal Q_n(b,r,t,\tau)$ are convex, we have 
	\begin{align*}
	\sup_{(b,\tau) \in \mathcal{B} \times \Upsilon}\left(||\hat{\gamma}(b,\tau)-\gamma_n(b,\tau)||_2 + ||\hat{\theta}(b,\tau)-\theta_n(b,\tau)||_2 \right) = o_p(1). 
	\end{align*}	
	Similarly, we have
	\begin{align*}
	\sup_{(b,\tau) \in \mathcal{B} \times \Upsilon}\left(||\gamma_\infty(b,\tau)-\gamma_n(b,\tau)||_2 + ||\theta_\infty(b,\tau)-\theta_n(b,\tau)||_2 \right) = o(1).
	\end{align*}
	This implies 
	\begin{align*}
	\sup_{(b,\tau) \in \mathcal{B} \times \Upsilon}\left|||\hat{\theta}(b,\tau)||_{\hat{A}_1(\tau)} - ||\theta_\infty(b,\tau)||_{A_1(\tau)} \right| = o_p(1),
	\end{align*}	
	and $0 = \lim_{n \rightarrow \infty}\theta_n(\beta_n(\tau),\tau) = \theta_\infty(\beta_0(\tau),\tau)$. In addition, under Assumptions \ref{ass:id2}(i) and \ref{ass:id2}(ii), \citet[Proof of Theorem 3]{Chernozhukov-Hansen(2005)} showed  $\theta_\infty(b,\tau)$ has a unique root for $\tau \in \Upsilon$,  which implies 
	$||\theta_\infty(b,\tau)||_{A_1(\tau)}$ is uniquely minimized at $b = \beta_0(\tau)$. Then, \citet[Lemma B.1]{Chernozhukov-Hansen(2006)} implies 
	\begin{align*}
	\sup_{\tau \in \Upsilon}||\hat{\beta}(\tau) - \beta_0(\tau)||_2 = o_p(1), 
	\end{align*}
	and thus, 
	\begin{align*}
	\sup_{\tau \in \Upsilon}||\hat{\beta}(\tau) - \beta_n(\tau)||_2 = o_p(1).
	\end{align*}
	Then, we have
	\begin{align*}
	& \sup_{\tau \in \Upsilon}||\hat{\gamma}(\hat{\beta}(\tau),\tau) -\gamma_n(\tau)||_2 \\
	&\leq \sup_{\tau \in \Upsilon,b \in \mathcal{B}}||\hat{\gamma}(b,\tau) -\gamma_\infty(b,\tau)||_2 + \sup_{\tau \in \Upsilon}||\gamma_\infty(\hat{\beta}(\tau),\tau) -\gamma_\infty(\beta_n(\tau),\tau)||_2 = o_p(1),
	\end{align*}
	and similarly, 
	\begin{align*}
	\sup_{\tau \in \Upsilon}||\hat{\theta}(\hat{\beta}(\tau),\tau) - 0||_2 =o_p(1).
	\end{align*}

	\textbf{Step 2.} We derive the convergence rates of $\hat{\beta}(\tau)$, $\hat{\gamma}(\tau)$, and $\hat{\theta}(\tau)$. Let $\mathcal{B}(\delta) = \{b(\cdot) \in \ell^\infty(\Upsilon): \sup_{\tau \in \Upsilon}||b(\tau)-\beta_n(\tau)||_2 \leq \delta\}$. For any $\delta>0$, we have $\hat{\beta}(\cdot) \in \mathcal{B}(\delta)$ w.p.a.1. Let $b_n(\tau)$ be a generic point in $\mathcal{B}(\delta)$. Recall $\omega = (0_{d_w \times d_\phi},\mathbb{I}_{d_\phi})$. Then, Lemma \ref{lem:basic} implies, with probability greater than $1-c\eps$,
	\begin{align*}
	& ||r_n\hat{\theta}(b_n(\tau),\tau)||_{\hat{A}_1(\tau)}^2 \\
	& = r_n^2 \left[ I_n(\tau) + II_n(b_n(\tau),\tau) + o_p(1/r_n) -\hat F_2(b_n(\tau),\tau)(b_n(\tau) - \beta_n(\tau)) \right]^\top \\
	& \times \left[\hat F_1(b_n(\tau),\tau)\right]^{-1} \omega^\top \hat{A}_1(\tau) \omega  \left[\hat F_1(b_n(\tau),\tau)\right]^{-1} \\
	& \times \left[ I_n(\tau) + II_n(b_n(\tau),\tau) + o_p(1/r_n) - \hat F_2(b_n(\tau),\tau)(b_n(\tau) - \beta_n(\tau)) \right] \\
	& \geq  r_n^2 (b_n(\tau) - \beta_n(\tau))^\top \biggl\{\hat F_2^\top(b_n(\tau),\tau)\left[\hat F_1(b_n(\tau),\tau)\right]^{-1} \omega^\top \hat{A}_1(\tau) \omega  \left[\hat F_1(b_n(\tau),\tau)\right]^{-1}\hat F_2(b_n(\tau),\tau)\biggr\}\\
	& \times (b_n(\tau) - \beta_n(\tau)) -  O_p(1),
	\end{align*}
	where the $O_p(1)$ term on the RHS of the above display holds uniformly over $\tau \in \Upsilon,b_n(\cdot) \in B(\delta)$. In addition, Lemma \ref{lem:basic} implies
	\begin{align*}
	& \sup_{b_n(\cdot) \in B(\delta),\tau \in \Upsilon}\biggl\Vert \hat F_2^\top(b_n(\tau),\tau)\left[\hat F_1(b_n(\tau),\tau)\right]^{-1} \omega^\top \hat{A}_1(\tau) \omega  \left[\hat F_1(b_n(\tau),\tau)\right]^{-1}\hat F_2(b_n(\tau),\tau) \\
	& - Q_{\Psi,X}^\top(\tau) Q_{\Psi,\Psi}^{-1}(\tau) \omega^T A_1(\tau) \omega Q_{\Psi,\Psi}^{-1}(\tau) Q_{\Psi,X}(\tau)\biggr\Vert_{op} = o_p(1). 
	\end{align*}
	Therefore, by Assumption \ref{ass:asym}, there exists a constant $\underline{c}$ independent of $\tau$ and $b_n(\cdot)$ such that 
	\begin{align*}
	& r_n^2 (b_n(\tau) - \beta_n(\tau))^\top \biggl\{\hat F_2^\top(b_n(\tau),\tau)\left[\hat F_1(b_n(\tau),\tau)\right]^{-1} \omega^\top \hat{A}_1(\tau) \omega  \left[\hat F_1(b_n(\tau),\tau)\right]^{-1}\hat F_2(b_n(\tau),\tau)\biggr\} \\
	& \times (b_n(\tau) - \beta_n(\tau)) - O_p(1) \\
	& \geq \underline{c}||r_n(b_n(\tau) - \beta_n(\tau))||_2^2 - O_p(1),
	\end{align*}
	where the $O_p(1)$ term on the RHS of the above display holds uniformly over $\tau \in \Upsilon,b_n(\cdot) \in B(\delta)$. 
	
	On the other hand, we have $\hat{\beta}(\tau) \in B(\delta)$ w.p.a.1 for any $\delta>0$ and 
	\begin{align*}
	\sup_{\tau \in \Upsilon}	||r_n\hat{\theta}(\hat{\beta}(\tau),\tau)||_{\hat{A}_1(\tau)} \leq \sup_{\tau \in \Upsilon} ||r_n\hat{\theta}(\beta_n(\tau),\tau)||_{\hat{A}_1(\tau)} = O_p(1), 
	\end{align*}
	where the last equality holds by Lemma \ref{lem:basic}. This implies
	\begin{align*}
	\underline{c}\sup_{\tau \in \Upsilon}||r_n(\hat{\beta}(\tau) - \beta_n(\tau))||_2^2 - O_p(1) \leq \sup_{\tau \in \Upsilon}||r_n\hat{\theta}(\hat{\beta}(\tau),\tau)||_{\hat{A}_1(\tau)} = O_p(1),
	\end{align*}
	and thus, 
	\begin{align}
	\sup_{\tau \in \Upsilon}||r_n(\hat{\beta}(\tau) - \beta_n(\tau))||_2 = O_p(1).
	\label{eq:alpharate}
	\end{align}
	Plugging \eqref{eq:alpharate} into \eqref{eq:betagammarate}, we obtain that 
	\begin{align*}
	\sup_{\tau \in \Upsilon}||r_n(\hat{\gamma}(\tau) -\gamma_n(\tau))||_2 = O_p(1) \quad \text{and} \quad \sup_{\tau \in \Upsilon}||r_n(\hat{\theta}(\tau) - 0)||_2 = O_p(1). 
	\end{align*}

	\textbf{Step 3.} Next, we derive the linear expansions for $\hat{\beta}(\tau)$ and $\hat{\gamma}(\tau)$. Let $\hat{u}(\tau) = r_n(\hat{\beta}(\tau) - \beta_n(\tau))$. Then, Step 2 shows $\sup_{\tau \in \Upsilon}||\hat{u}(\tau)||_2 = O_p(1)$. For any $\eps>0$, there exists a constant $C>0$ such that with probability greater than $1-\eps$, we have, for all $\tau \in \Upsilon$,
	\begin{align*}
	\hat{u}(\tau) = \arg \inf_{u: ||u||_2 \leq C}||\hat{\theta}(\beta_n(\tau) + u/r_n,\tau )||_{\hat{A}_1(\tau)}.
	\end{align*}
	Denote $b_n(\tau) = \beta_n(\tau) + u/r_n$ for $||u||_2 \leq C$. Then, by Lemma \ref{lem:basic}, we have 
	\begin{align*}
	& ||\hat{\theta}(\beta_n(\tau) + u/r_n,\tau )||_{\hat{A}_1(\tau)} \\
	& = \biggl[r_n\mathbb{P}_nf_{\tau}(D_{i,j},\beta_n(\tau),\gamma_n(\tau),0) - \hat F_2(\beta_n(\tau) + u/r_n,\tau) u + r_n II_n(\beta_n(\tau)+ u/r_n,\tau) + o_p(1)\biggr]^\top \\
	& \times \left[ \hat F_1^{-1}(\beta_n(\tau) + u/r_n,\tau)\omega^\top A_1(\tau) \omega \hat F_1^{-1}(\beta_n(\tau) + u/r_n,\tau)\right] \\
	& \times \biggl[r_n\mathbb{P}_nf_{\tau}(D_{i,j},\beta_n(\tau),\gamma_n(\tau),0) - \hat F_2(\beta_n(\tau) + u/r_n,\tau) u + r_n II_n(\beta_n(\tau)+ u/r_n,\tau) + o_p(1)\biggr],
	\end{align*}
	where the $o_p(1)$ term is uniform over $\tau \in \Upsilon$ and $|u| \leq C$. 
	In addition, by Assumption \ref{ass:asym}, we have
	\begin{align}
	\sup_{\tau \in \Upsilon,||u||_2 \leq C} \left\Vert \hat F_1^{-1}(\beta_n(\tau) + u/r_n,\tau) - Q_{\Psi,\Psi}^{-1}(\tau)\right\Vert_{op} = o_p(1),
	\label{eq:Jrb1}
	\end{align}
	\begin{align*}
	\sup_{\tau \in \Upsilon,||u||_2 \leq C} \left\Vert \hat F_2(\beta_n(\tau) + u/r_n,\tau) - Q_{\Psi,X}(\tau)\right\Vert_{op} = o_p(1),
	\end{align*}
	and 
	\begin{align}
	\sup_{\tau \in \Upsilon,||u||_2 \leq C}r_n||II_n(\beta_n(\tau) + u/r_n,\tau)||_2 = o_p(1). 
	\label{eq:II2}
	\end{align}
	
	Then, we have
	\begin{align}\label{eq:thetahat_limit}
	\left||| \hat{\theta}(\beta_n(\tau) + u/r_n,\tau )||_{\hat{A}_1(\tau)} - || \omega Q^{-1}_{\Psi,\Psi}(\tau)\left[r_n\mathbb{P}_n f_{\tau}(D_{i,j},\beta_n(\tau),\gamma_n(\tau),0) - Q_{\Psi,X}(\tau) u \right]  ||_{A_1(\tau)}\right| = o_p(1).
	\end{align}
	Then, \citet[Lemma B.1]{Chernozhukov-Hansen(2006)} implies 
	\begin{align}
	\hat{u}(\tau) & = \left[Q_{\Psi,X}^\top(\tau) Q_{\Psi,\Psi}^{-1}(\tau) \omega^\top A_1(\tau)  \omega Q^{-1}_{\Psi,\Psi}(\tau) Q_{\Psi,X}(\tau)\right]^{-1} \notag \\
	& \times Q_{\Psi,X}^\top(\tau) Q_{\Psi,\Psi}^{-1}(\tau) \omega^\top A_1(\tau) \omega Q^{-1}_{\Psi,\Psi}(\tau) r_n\mathbb{P}_nf_{\tau}(D_{i,j},\beta_n(\tau),\gamma_n(\tau),0) + o_p(1),
	\label{eq:alpha_expand2}
	\end{align}
	where $o_p(1)$ term holds uniformly over $\tau \in \Upsilon$. Plugging \eqref{eq:alpha_expand2} into \eqref{eq:betagammarate}, we have
	\begin{align*}
	& \begin{pmatrix}
	r_n\left(\hat{\gamma}(\tau) -\gamma_n(\tau)\right) \\
	r_n\hat{\theta}(\tau)
	\end{pmatrix}\\
	& = \left[\hat F_1(\hat{\beta}(\tau),\tau)\right]^{-1} r_n \left[ I_n(\tau) + II_n(\hat{\beta}(\tau),\tau) + o_p(1/r_n) - \hat F_2(\hat{\beta}(\tau),\tau)(\hat{\beta}(\tau) - \beta_n(\tau)) \right] \\
	& = Q_{\Psi,\Psi}^{-1}(\tau)\biggl[ \mathbb{I}_{d_w+d_\phi} -  Q_{\Psi,X}(\tau)\left[Q_{\Psi,X}^\top(\tau) Q_{\Psi,\Psi}^{-1}(\tau) \omega^\top A_1(\tau)  \omega Q^{-1}_{\Psi,\Psi}(\tau) Q_{\Psi,X}(\tau)\right]^{-1}\\
	& \times Q_{\Psi,X}^\top(\tau) Q_{\Psi,\Psi}^{-1}(\tau) \omega^\top A_1(\tau) \omega Q^{-1}_{\Psi,\Psi}(\tau) \biggr] \times r_n\mathbb{P}_nf_{\tau}(D_{i,j},\beta_n(\tau),\gamma_n(\tau),0) + o_p(1),
	\end{align*}
	where both $o_p(1/r_n)$ and $o_p(1)$ terms hold uniformly over $\tau \in \Upsilon$. This concludes the proof. 
	
\end{proof}

\begin{lemma}
	If Assumptions \ref{ass:id} and \ref{ass:asym} hold, then 
	\begin{align*}
	\begin{pmatrix}
	r_n\left(\hat{\gamma}(\beta_n(\tau),\tau) -\gamma_n(\tau)\right) \\
	r_n\hat{\theta}(\beta_n(\tau),\tau)
	\end{pmatrix}=  Q_{\Psi,\Psi}^{-1}(\tau)r_n (\mathbb{P}_n - \overline{\mathbb{P}}_n)f_\tau(D,\beta_n(\tau),\gamma_n(\tau),0) + o_p(1),
	\end{align*}
	where the $o_p(1)$ terms hold uniformly over $\tau \in \Upsilon$.
	\label{lem:expand'}
\end{lemma}
\begin{proof}
	By Lemma \ref{lem:basic} with $b_n(\tau) = \beta_n(\tau)$, we have
	\begin{align*}
	& \begin{pmatrix}
	r_n\left(\hat{\gamma}(\beta_n(\tau),\tau) -\gamma_n(\tau)\right) \\
	r_n\hat{\theta}(\beta_n(\tau),\tau)
	\end{pmatrix} = \left[\hat F_1(\beta_n(\tau),\tau)\right]^{-1} r_n \left[ I_n(\tau) + II_n(\beta_n(\tau),\tau) + o_p(1/r_n)\right],
	\end{align*}
	where the $o_p(1/r_n)$ term holds uniformly over $\tau \in \Upsilon$. Then, by \eqref{eq:Jrb1} and \eqref{eq:II2}, we have
	\begin{align*}
	\begin{pmatrix}
	r_n\left(\hat{\gamma}(\beta_n(\tau),\tau) -\gamma_n(\tau)\right) \\
	r_n\hat{\theta}(\beta_n(\tau),\tau)
	\end{pmatrix} =  Q_{\Psi,\Psi}^{-1}(\tau) r_n  I_n(\tau) + o_p(1),
	\end{align*}
	where the $o_p(1)$ terms hold uniformly over $\tau \in \Upsilon$. This concludes the proof. 
\end{proof}

\subsection{Technical Results for the Bootstrap IVQR Estimator}
\begin{lemma}
	Suppose Assumptions  \ref{ass:id} and \ref{ass:asym} hold. Then, 
	\begin{align*}
	\sup_{\tau \in \Upsilon} \left(  || r_n \hat \theta_g^*(\beta_0(\tau),\tau) ||_2 + || r_n (\hat \gamma_g^*(\beta_0(\tau),\tau)-\gamma_n(\tau)) ||_2 \right)= O_p(1),
	\end{align*}
	and 
	\begin{align*}
	r_n  \hat{\theta}_g^*(\beta_0(\tau),\tau) & = \omega Q_{\Psi,\Psi}^{-1}(\tau) r_n\sum_{j \in [J]}(1+g_j) \xi_j \mathbb P_{n,j}f_{\tau}(D_{i,j},\beta_n(\tau),\gamma_n(\tau),0) \\
	& + \omega Q_{\Psi,\Psi}^{-1}(\tau) \sum_{j \in [J]} \xi_j (1+g_j) Q_{\Psi,X,j}\mu_\beta(\tau) + o_p(1),
	\end{align*}
	where the $o_p(1)$ term holds uniformly over $\tau \in \Upsilon$. 
	If we further assume Assumption  \ref{ass:id2}, then we have 
	\begin{align*}
	r_n(\hat{\beta}_g^*(\tau) - \hat{\beta}(\tau)) & =  \tilde{\Gamma}(\tau) \sum_{ j \in [J]}\xi_j g_j\mathcal Z_j + \overline a^*_g(\tau) \mu_\beta(\tau) + o_p(1),
	\end{align*}
	where $o_p(1)$ term holds uniformly over $\tau \in \Upsilon$ and $\tilde \Gamma(\tau)$  and $\Gamma(\tau)$ are defined in \eqref{eq:Gamma}.
	\label{lem:bootexpand}
\end{lemma}

\begin{proof}
	We divide the proof into three steps. In the first step, we show the consistency of $(\hat{\gamma}_g^*(\beta_0(\tau),\tau),\hat{\theta}_g^*(\beta_0(\tau), \tau))$. In the second step, we show the first desired result. These two steps do not require Assumption \ref{ass:id2} as $\beta_0$ is assumed to be in the local neighborhood of $\beta_n(\tau)$. 	In the third step, we show the second desired result. 
	
	\textbf{Step 1}. By Assumption \ref{ass:asym}(iv), we have 
	\begin{align}\label{eq:0}
	\sup_{\tau \in \Upsilon}     \left\Vert\frac{1}{n}\sum_{j \in [J]} g_j\sum_{i \in I_{n,j}}  f_{\tau}(D_{i,j},\beta_n(\tau),\gamma_n(\tau),0)\right\Vert_2 = o_p(1).
	\end{align}
	
	In addition, note that 
	\begin{align}
	& \left\Vert\frac{1}{n}\sum_{j \in [J]} g_j\sum_{i \in I_{n,j}} \left( \hat{f}_{\tau}(D_{i,j},\beta_0(\tau),\hat{\gamma}^r(\tau),0) - f_{\tau}(D_{i,j},\beta_n(\tau),\gamma_n(\tau),0)\right)\right\Vert_2 \notag \\
	& \leq \sum_{j \in [J]} \xi_j \left\Vert  (\mathbb{P}_{n,j} - \overline{\mathbb{P}}_{n,j}) ( \hat{f}_{\tau}(D_{i,j},\beta_0(\tau),\hat{\gamma}^r(\tau),0) - f_{\tau}(D_{i,j},\beta_n(\tau),\gamma_n(\tau),0))\right\Vert_2 \notag \\
	& + \sum_{j \in [J]} \xi_j \left\Vert  \overline{\mathbb{P}}_{n,j}  \hat{f}_{\tau}(D_{i,j},\beta_0(\tau),\hat{\gamma}^r(\tau),0)\right\Vert_2.
	\label{eq:1}
	\end{align}
	Because $\hat{\gamma}^r(\tau)$ is consistent as shown in Lemma \ref{lem:restricted} and $\sup_{\tau \in \Upsilon}|\beta_0(\tau) - \beta_n(\tau)| = o(1)$, Assumption \ref{ass:asym} implies 
	\begin{align}
	\sup_{\tau \in \Upsilon}\sum_{j \in [J]} \xi_j \left\Vert  (\mathbb{P}_{n,j} - \overline{\mathbb{P}}_{n,j}) ( \hat{f}_{\tau}(D_{i,j},\beta_0(\tau),\hat{\gamma}^r(\tau),0) - f_{\tau}(D_{i,j},\beta_n(\tau),\gamma_n(\tau),0))\right\Vert_2 = o_p(r_n^{-1}).
	\label{eq:II}
	\end{align}
	In addition, due to  Lemma \ref{lem:restricted}, we have
	\begin{align*}
	\sup_{\tau \in \Upsilon}\left\Vert   \hat{\gamma}^r(\tau) -\gamma_n(\tau)  \right\Vert_2 = O_p(r_n^{-1}).
	\end{align*}
	
	Then, by Assumption \ref{ass:asym} and
	following the same argument in \eqref{eq:Pjf}, we can show that 
	\begin{align}
	& \sup_{\tau \in \Upsilon}\left\Vert  \overline{\mathbb{P}}_{n,j}  \hat{f}_{\tau}(D_{i,j},\beta_0(\tau),\hat{\gamma}^r(\tau),0) + Q_{\Psi,X,j}(\tau)(\beta_0(\tau) - \beta_n(\tau)) + Q_{\Psi,\Psi,j}(\tau) \begin{pmatrix}
	\hat{\gamma}^r(\tau) -\gamma_n(\tau) \\
	0 
	\end{pmatrix}  \right\Vert_2 \notag \\
	& = o_p(r_n^{-1}).
	\label{eq:III}
	\end{align}
	
	This further implies 
	\begin{align}
	\sup_{\tau \in \Upsilon}\sum_{j \in [J]} \xi_j\left\Vert  \overline{\mathbb{P}}_{n,j}  \hat{f}_{\tau}(D_{i,j},\beta_0(\tau),\hat{\gamma}^r(\tau),0)\right\Vert_2 = o_p(1).
	\label{eq:III'}
	\end{align}
	
	Combining \eqref{eq:0}--\eqref{eq:III'}, we have
	\begin{align}\label{eq:2}
	\sup_{\tau \in \Upsilon} \left\Vert\frac{1}{n}\sum_{j \in [J]} g_j\sum_{i \in I_{n,j}}  \hat{f}_{\tau}(D_{i,j},\beta_0(\tau),\hat{\gamma}^r(\tau),0)\right\Vert_2 = o_p(1).
	\end{align}
	Let 
	\begin{align*}
	\tilde{\mathcal Q}_n(b,r,t,\tau) & = \sum_{j \in [J]}\sum_{i \in I_{n,j}}\rho_\tau(y_{i,j} - X_{i,j} b - W_{i,j}^\top r - \hat{\Phi}_{i,j}^\top(\tau) t) \hat{V}_{i,j}(\tau) \\
	& - \sum_{j \in [J]} g_j\sum_{i \in I_{n,j}} \hat{f}^\top_\tau(D_{i,j},\beta_0(\tau),\hat{\gamma}^r(\tau),0) \begin{pmatrix}
	r \\
	t
	\end{pmatrix} \\
	& = \hat{\mathcal Q}_n(b,r,t,\tau) - \sum_{j \in [J]} g_j\sum_{i \in I_{n,j}} \hat{f}^\top_\tau(D_{i,j},\beta_0(\tau),\hat{\gamma}^r(\tau),0) \begin{pmatrix}
	r \\
	t
	\end{pmatrix}.
	\end{align*}
	Then, we have
	\begin{align*}
	(\hat{\gamma}_g^*(b,\tau),\hat{\theta}_g^*(b,\tau)) & = \arg \inf_{r,t} \tilde{\mathcal Q}_n(b,r,t,\tau).
	\end{align*}
	By \eqref{eq:2}, uniformly over $(\tau,b) \in \Upsilon \times \mathcal{B}$,  
	\begin{align*}
	\tilde{\mathcal Q}_n(b,r,t,\tau) \convP \mathcal Q_\infty(b,r,t,\tau).
	\end{align*}
	Then, because $\tilde{\mathcal Q}_n(b,r,t,\tau) $ is convex in $(r,t)$, by \citet[Theorem 1]{K09}, we have
	\begin{align}\label{eq:theta*uniform}
	\sup_{(b,\tau) \in \mathcal{B} \times \Upsilon}\left(||\hat{\gamma}_g^*(b,\tau)-\gamma_\infty(b,\tau)||_2 + ||\hat{\theta}_g^*(b,\tau)-\theta_\infty(b,\tau)||_2 \right) = o_p(1). 
	\end{align}	
	This implies 
	$\sup_{\tau \in \Upsilon}||\hat{\theta}_g^*(\beta_0(\tau),\tau)||_2 = o_p(1)$.

	\textbf{Step 2.} For any $b_n( \cdot) \in \mathcal{B}(\delta)$, the sub-gradient condition for $(\hat{\gamma}_g^*(b_n(\tau),\tau), \hat{\theta}_g^*(b_n(\tau),\tau))$ is 
	\begin{align}
	o_p(1/r_n) & = \mathbb{P}_n\hat{f}_{\tau}(D_{i,j},b_n(\tau), \hat{\gamma}_g^*(b_n(\tau),\tau),\hat{\theta}_g^*(b_n(\tau),\tau),\tau) + \frac{1}{n} \sum_{j \in [J]} g_j\sum_{i \in I_{n,j}}\hat{f}_{\tau}(D_{i,j},\beta_0(\tau), \hat{\gamma}^r(\tau),0),
	\label{eq:expand2}
	\end{align}
	where the $o_p(1/r_n)$ term on the LHS of the above display holds uniformly over $\tau \in \Upsilon,b_n(\cdot) \in B(\delta)$.
	
	Following the same argument in Lemma \ref{lem:basic}, for any $\eps>0$, there exists $\overline{\delta}$ such that for $\delta,\delta' \leq \overline{\delta}$, we have, with probability greater than $1-c \eps$, 
	\begin{align}
	& \mathbb{P}_n\hat{f}_{\tau}(D_{i,j},b_n(\tau), \hat{\gamma}_g^*(b_n(\tau),\tau),\hat{\theta}_g^*(b_n(\tau),\tau),\tau) \notag \\
	& = I_n(\tau) + II_{n,g}(b_n(\tau),\tau)  - \hat F_{2,g}(b_n(\tau),\tau)(b_n(\tau) - \beta_n(\tau)) - \hat F_{1,g}(b_n(\tau),\tau)\begin{pmatrix}
	\hat{\gamma}_g^*(b_n(\tau),\tau) -\gamma_n(\tau) \\
	\hat{\theta}_g^*(b_n(\tau),\tau) 
	\end{pmatrix},
	\label{eq:subgradient*}
	\end{align}
	where 
	\begin{align*}
	I_n(\tau) = (\mathbb{P}_{n}-\overline{\mathbb{P}}_{n})f_{\tau}(D_{i,j},\beta_n(\tau),\gamma_n(\tau),0), 
	\end{align*}
	\begin{align*}
	II_{n,g}(b_n(\tau),\tau) = \sum_{j \in [J]} \xi_j (\mathbb{P}_{n,j}-\overline{\mathbb{P}}_{n,j})\left(\hat{f}_{\tau}(D_{i,j},b_n(\tau), \hat{\gamma}_g^*(b_n(\tau),\tau),\hat{\theta}_g^*(b_n(\tau),\tau),\tau) - f_{\tau}(D_{i,j},\beta_n(\tau),\gamma_n(\tau),0)  \right)
	\end{align*}
	such that $\sup_{b_n(\cdot) \in B(\delta),\tau \in \Upsilon} r_n||II_{n,g}(b_n(\tau),\tau) ||_2 \leq \eps$,
	\begin{align*}
	\hat F_{1,g}(b_n(\tau),\tau) = \sum_{j \in [J]}\frac{\xi_j}{n_j} \sum_{i \in I_{n,j}} \mathbb{E}f_{\eps_{i,j}(\tau)}(\hat{\delta}_{i,j,g}(\tau)|W_{i,j},Z_{i,j})\hat{\Psi}_{i,j}(\tau)\hat{\Psi}_{i,j}^\top(\tau)\hat V_{i,j}(\tau),
	\end{align*}
	\begin{align*}
	\hat F_{2,g}(b_n(\tau),\tau) = \sum_{j \in [J]}\frac{\xi_j}{n_j} \sum_{i \in I_{n,j}} \mathbb{E}f_{\eps_{i,j}(\tau)}(\hat{\delta}_{i,j,g}(\tau)|W_{i,j},Z_{i,j})\hat{\Psi}_{i,j}(\tau)X_{i,j}\hat V_{i,j}(\tau),
	\end{align*}
	\begin{align*}
	\hat{\delta}_{i,j,g}(\tau) \in (0, X_{i,j}(b_n(\tau) -\beta_n(\tau)) + W_{i,j}^\top (\hat{\gamma}_g^*(b_n(\tau),\tau)-\gamma_n(\tau)) + \hat{\Phi}_{i,j}^\top(\tau) \hat{\theta}_g^*(b_n(\tau),\tau)),
	\end{align*}
	\begin{align*}
	\sup_{b_n(\cdot) \in B(\delta),\tau \in \Upsilon,j \in [J]}\left\Vert \frac{1}{n_j} \sum_{i \in I_{n,j}} \mathbb{E}f_{\eps_{i,j}(\tau)}(\hat{\delta}_{i,j,g}(\tau)|W_{i,j},Z_{i,j})\hat{\Psi}_{i,j}(\tau)\hat{\Psi}_{i,j}^\top(\tau)\hat V_{i,j}(\tau) - Q_{\Psi,\Psi,j}(\tau)\right\Vert_{op} \leq \eps,
	\end{align*}
	and 
	\begin{align*}
	\sup_{b_n(\cdot) \in B(\delta),\tau \in \Upsilon,j \in [J]}\left\Vert \frac{1}{n_j} \sum_{i \in I_{n,j}} \mathbb{E}f_{\eps_{i,j}(\tau)}(\hat{\delta}_{i,j,g}(\tau)|W_{i,j},Z_{i,j})\hat{\Psi}_{i,j}(\tau)X_{i,j}\hat V_{i,j}(\tau) - Q_{\Psi,X,j}(\tau)\right\Vert_{op} \leq \eps.
	\end{align*}
	
	In addition, by \eqref{eq:1}, \eqref{eq:II}, and \eqref{eq:III}, we have
	\begin{align*}
	&  \frac{1}{n}\sum_{j \in [J]} g_j\sum_{i \in I_{n,j}}  \hat{f}_{\tau}(D_{i,j},\beta_0(\tau),\hat{\gamma}^r(\tau),0) \\
	& = \frac{1}{n}\sum_{j \in [J]} g_j\sum_{i \in I_{n,j}} \left( \hat{f}_{\tau}(D_{i,j},\beta_0(\tau),\hat{\gamma}^r(\tau),0) - f_{\tau}(D_{i,j},\beta_n(\tau),\gamma_n(\tau),0)\right) \\
	& +  \frac{1}{n}\sum_{j \in [J]} g_j\sum_{i \in I_{n,j}} f_{\tau}(D_{i,j},\beta_n(\tau),\gamma_n(\tau),0) \\
	& =  \frac{1}{n}\sum_{j \in [J]} g_j\left[\sum_{i \in I_{n,j}}\left( f_{\tau}(D_{i,j},\beta_n(\tau),\gamma_n(\tau),0)  - Q_{\Psi,X,j}(\tau)(\beta_0(\tau) - \beta_n(\tau)) - Q_{\Psi,\Psi,j}(\tau) \begin{pmatrix}
	\hat{\gamma}^r(\tau) -\gamma_n(\tau) \\
	0
	\end{pmatrix} \right)\right] \\
	& + o_p(1/r_n),
	\end{align*}
	where the $o_p(1/r_n)$ term holds uniformly over $\tau \in \Upsilon$. Combining this with Assumption \ref{ass:asym} and Lemma \ref{lem:restricted}, we have  
	\begin{align}\label{eq:expand3}
	\sup_{\tau \in \Upsilon}\left\Vert \frac{1}{n}\sum_{j \in [J]} g_j\sum_{i \in I_{n,j}}  \hat{f}_{\tau}(D_{i,j},\beta_0(\tau),\hat{\gamma}^r(\tau),0) \right\Vert_2 = O_p(1/r_n).
	\end{align}
	
	Combining \eqref{eq:expand2}, \eqref{eq:subgradient*}, and \eqref{eq:expand3} implies 
	\begin{align}
	o_p(1) =  -\hat F_{2,g}(b_n(\tau),\tau)r_n(b_n(\tau) - \beta_n(\tau)) - \hat F_{1,g}(b_n(\tau),\tau)\begin{pmatrix}
	r_n(\hat{\gamma}_g^*(b_n(\tau),\tau) -\gamma_n(\tau)) \\
	r_n\hat{\theta}_g^*(b_n(\tau),\tau) 
	\end{pmatrix} + O_p(1),
	\label{eq:betastar}
	\end{align}
	where the $o_p(1)$ and $O_p(1)$ terms hold uniformly over $\{b_n(\cdot) \in B(\delta),\tau \in \Upsilon\}$. 
	
	By letting $b_n(\tau) = \beta_0(\tau)$ in the above display and noting that 
	\begin{align*}
	\sup_{\tau \in \Upsilon}\left\Vert \hat F_{1,g}(\beta_0(\tau),\tau) - Q_{\Psi,\Psi}(\tau)\right\Vert_{op} = o_p(1),
	\end{align*}
	we have
	\begin{align*}
	\sup_{\tau \in \Upsilon} \left(  || r_n \hat \theta_g^*(\beta_0(\tau),\tau) ||_2 + || r_n (\hat \gamma_g^*(\beta_0(\tau),\tau)-\gamma_n(\tau)) ||_2 \right)= O_p(1).
	\end{align*}
	
	In addition, by letting $b_n(\tau) = \beta_0(\tau)$ in \eqref{eq:expand2}, \eqref{eq:subgradient*}, and \eqref{eq:expand3}, we have
	\begin{align*}
	o_p(1/r_n) & =  I_n(\tau) + II_{n,g}(\beta_0(\tau),\tau)  + \hat F_{2,g}(\beta_0(\tau),\tau)\frac{\mu_\beta(\tau)}{r_n} - \hat F_{1,g}(\beta_0(\tau),\tau)\begin{pmatrix}
	\hat{\gamma}_g^*(\beta_0(\tau),\tau) -\gamma_n(\tau) \\
	\hat{\theta}_g^*(\beta_0(\tau),\tau) 
	\end{pmatrix} \\
	& + \frac{1}{n}\sum_{j \in [J]} g_j\left[\sum_{i \in I_{n,j}}\left( f_{\tau}(D_{i,j},\beta_n(\tau),\gamma_n(\tau),0) + Q_{\Psi,X,j}(\tau)\frac{\mu_\beta(\tau)}{r_n} - Q_{\Psi,\Psi,j}(\tau) \begin{pmatrix}
	\hat{\gamma}^r(\tau) -\gamma_n(\tau) \\
	0
	\end{pmatrix}\right)\right] \\
	& + o_p(1/r_n), 
	\end{align*}
	where the $o_p(1/r_n)$ holds uniformly over $\tau \in \Upsilon$. This further implies 
	\begin{align*}
	r_n \begin{pmatrix}
	\hat{\gamma}_g^*(\beta_0(\tau),\tau) -\gamma_n(\tau) \\
	\hat{\theta}_g^*(\beta_0(\tau),\tau) 
	\end{pmatrix} & = Q_{\Psi,\Psi}^{-1}(\tau) r_n\sum_{j \in [J]}(1+g_j) \xi_j \mathbb P_{n,j}f_{\tau}(D_{i,j},\beta_n(\tau),\gamma_n(\tau),0) \\
	& + Q_{\Psi,\Psi}^{-1}(\tau) \sum_{j \in [J]} \xi_j (1+g_j) Q_{\Psi,X,j}\mu_\beta(\tau) \\
	& - Q_{\Psi,\Psi}^{-1}(\tau) \sum_{j \in [J]}\xi_j g_j Q_{\Psi,\Psi,j} (\tau)\begin{pmatrix}
	r_n(\hat{\gamma}^r(\tau) -\gamma_n(\tau)) \\
	0
	\end{pmatrix} + o_p(1),
	\end{align*}
	where the $o_p(1)$ term holds uniformly over $\tau \in \Upsilon$. In addition, we note that 
	\begin{align*}
	\omega Q_{\Psi,\Psi}^{-1}(\tau) \sum_{j \in [J]}\xi_j g_j Q_{\Psi,\Psi,j} (\tau)\begin{pmatrix}
	r_n(\hat{\gamma}^r(\tau) -\gamma_n(\tau)) \\
	0
	\end{pmatrix} = 0 
	\end{align*}
	because $\omega Q_{\Psi,\Psi}^{-1}(\tau)$ and $Q_{\Psi,\Psi,j}$ are all block diagonal matrices. 
	This implies 
	\begin{align*}
	r_n  \hat{\theta}_g^*(\beta_0(\tau),\tau) & = \omega Q_{\Psi,\Psi}^{-1}(\tau) r_n\sum_{j \in [J]}(1+g_j) \xi_j \mathbb P_{n,j}f_{\tau}(D_{i,j},\beta_n(\tau),\gamma_n(\tau),0) \\
	& + \omega Q_{\Psi,\Psi}^{-1}(\tau) \sum_{j \in [J]} \xi_j (1+g_j) Q_{\Psi,X,j}\mu_\beta(\tau) + o_p(1),
	\end{align*}
	where the $o_p(1)$ term holds uniformly over $\tau \in \Upsilon$.

	\textbf{Step 3.} Last, we show the second result in the Lemma. Suppose, in addition, Assumption \ref{ass:id2} holds. Then, by \eqref{eq:theta*uniform} and the same argument in Step 1 of the proof of Lemma \ref{lem:expand}, we can show 
	\begin{align*}
	& \sup_{\tau \in \Upsilon}||\hat{\beta}_g^*(\tau) - \beta_n(\tau)||_2 = o_p(1), 
	\end{align*}
	which further implies 
	\begin{align*}
	& \sup_{\tau \in \Upsilon}||\hat{\gamma}_g^*(\tau) -\gamma_n(\tau)||_2 = o_p(1) \quad \text{and} \quad \sup_{\tau \in \Upsilon}||\hat{\theta}_g^*(\tau) - 0||_2 = o_p(1). 
	\end{align*}
	
	By letting $b_n(\tau) = \hat{\beta}_g^*(\tau)$ in \eqref{eq:betastar}, we have $\hat{\beta}_g^*(\tau) \in \mathcal{B}(\delta)$ w.p.a.1 for any $\delta>0$, and thus, 
	\begin{align*}
	||\omega \left[\hat F_{1,g}(\hat{\beta}_g^*(\tau),\tau)\right]^{-1} \left[\hat F_{2,g}(\hat{\beta}_g^*(\tau),\tau)r_n(\hat{\beta}_g^*(\tau) - \beta_n(\tau))+O_p(1)\right] ||_{\hat{A}_1(\tau)}^2 \leq ||r_n\hat{\theta}_g^*(\beta_n(\tau),\tau) ||_{\hat{A}_1(\tau)}^2. 
	\end{align*}
	In addition, note that w.p.a.1,
	\begin{align*}
	\inf_{\tau \in \Upsilon}\lambda_{\min}([\hat F_{2,g}(\hat{\beta}_g^*(\tau),\tau)]^\top\left[\hat F_{1,g}(\hat{\beta}_g^*(\tau),\tau)\right]^{-1}\omega^T \hat{A}_1(\tau)\omega \left[\hat F_{1,g}(\hat{\beta}_g^*(\tau),\tau)\right]^{-1} \hat F_{2,g}(\hat{\beta}_g^*(\tau),\tau) ) \geq \underline{c}>0,
	\end{align*}
	so that we have
	\begin{align*}
	\sup_{\tau \in \Upsilon}	\underline{c}r_n^2||\hat{\beta}_g^*(\tau) - \beta_n(\tau)||_2^2 - O_p(1) \leq \sup_{\tau \in \Upsilon}||r_n\hat{\theta}_g^*(\beta_n(\tau),\tau) ||_{\hat{A}_1(\tau)}^2 \leq O_p(1).
	\end{align*}
	Therefore, we have
	\begin{align*}
	\sup_{\tau \in \Upsilon}r_n||\hat{\beta}_g^*(\tau) - \beta_n(\tau)||_2 = O_p(1). 
	\end{align*}
	Plugging this  into \eqref{eq:betastar}, we have
	\begin{align*}
	\sup_{\tau \in \Upsilon}r_n||\hat{\gamma}_g^*(\tau) -\gamma_n(\tau)||_2 = O_p(1) \quad \text{and} \quad \sup_{\tau \in \Upsilon}r_n||\hat{\theta}_g^*(\tau)||_2 = O_p(1).
	\end{align*}
	
	Then, let $b_n(\tau) = \beta_n(\tau) + u/r_n$ in \eqref{eq:expand2}, \eqref{eq:subgradient*}, and \eqref{eq:expand3}, we have 
	\begin{align*}
	o_p(1) & = r_nI_n(\tau) - Q_{\Psi,X}(\tau)u - Q_{\Psi,\Psi}(\tau)\begin{pmatrix}
	r_n(\hat{\gamma}_g^*(\beta_n(\tau) + u/r_n,\tau) -\gamma_n(\tau)) \\
	r_n\hat{\theta}_g^*(\beta_n(\tau) + u/r_n,\tau) 
	\end{pmatrix} \\
	& + r_nI_{n,g}(\tau)- \sum_{j \in [J]} g_j \xi_j Q_{\Psi,X,j}(\tau)r_n (\beta_0(\tau) - \beta_n(\tau)) \\
	& -\sum_{j \in [J]} g_j \xi_j Q_{\Psi,\Psi,j}(\tau) \begin{pmatrix}
	r_n(\hat{\gamma}^r(\tau) -\gamma_n(\tau)) \\
	0 
	\end{pmatrix} -  o_p(1),
	\end{align*}
	where 
	\begin{align*}
	I_{n,g}(\tau) = \sum_{j \in [J]}\xi_j g_j(\mathbb{P}_{n,j} - \overline{\mathbb{P}}_{n,j})f_{\tau}(D_{i,j},\beta_n(\tau),\gamma_n(\tau),0)
	\end{align*}
	such that $r_n I_{n,g}(\tau) = \sum_{j \in [J]} \xi_j g_j r_n(\mathbb{P}_{n,j} - \overline{\mathbb{P}}_{n,j})f_{\tau}(D_{i,j},\beta_n(\tau),\gamma_n(\tau),0)$
	and the $o_p(1)$ term holds uniformly over $\tau \in \Upsilon, |u|\leq M$. This implies 
	\begin{align*}
	& r_n\hat{\theta}_g^*(\beta_n(\tau) + u/r_n,\tau) \\
	& = \omega Q_{\Psi,\Psi}^{-1}(\tau)\biggl[r_n I_n(\tau) +  r_n I_{n,g}(\tau) - Q_{\Psi,X}(\tau) u - \sum_{j \in [J]} g_j \xi_j Q_{\Psi,X,j}(\tau)r_n (\beta_0(\tau) - \beta_n(\tau)) \\
	& -\sum_{j \in [J]} g_j \xi_j Q_{\Psi,\Psi,j}(\tau) \begin{pmatrix}
	r_n(\hat{\gamma}^r(\tau) -\gamma_n(\tau)) \\
	0 
	\end{pmatrix} -  o_p(1)\biggr]\\ 
	& = G_n(u,\tau) + o_p(1),
	\end{align*}
	where 
	\begin{align*}
	G_n(u,\tau) & = \omega Q_{\Psi,\Psi}^{-1}(\tau)\biggl[r_nI_n(\tau) + r_n I_{n,g}(\tau) - Q_{\Psi,X}(\tau) u - \sum_{j \in [J]} g_j \xi_j Q_{\Psi,X,j}(\tau)r_n (\beta_0(\tau) - \beta_n(\tau)) \\
	& -\sum_{j \in [J]} g_j \xi_j Q_{\Psi,\Psi,j}(\tau) \begin{pmatrix}
	r_n(\hat{\gamma}^r(\tau) -\gamma_n(\tau)) \\
	0 
	\end{pmatrix}\biggr]
	\end{align*}
	and the $o_p(1)$ term holds uniformly over $\tau \in \Upsilon, |u|\leq M$. Let $\hat{u}_g^*(\tau) = r_n(\hat{\beta}_g^*(\tau) - \beta_n(\tau))$. 
	Because $\sup_{\tau \in \Upsilon}||\hat{u}^*_g(\tau)||_2 = O_p(1)$, for any $\eps>0$, there exists an integer $\underline{n}$ such that for $n \geq \underline{n}$, there exists a sufficiently large constant $M>0$ such that 
	\begin{align*}
	\hat{u}_g^*(\tau) = \arg \inf_{||u||_2 \leq M}||r_n\hat{\theta}_g^*(\beta_n(\tau) + u/r_n,\tau)||^2_{\hat{A}_1(\tau)}. 
	\end{align*}
	In addition, because
	\begin{align*}
	\sup_{\tau \in \Upsilon, ||u||_2 \leq M}\biggl| ||r_n\hat{\theta}_g^*(\beta_n(\tau) + u/r_n,\tau)||^2_{\hat{A}_1(\tau)}- ||G_n(u,\tau)||^2_{A_1(\tau)} \biggr| = o_p(1),
	\end{align*}
	\citet[Lemma B.1]{Chernozhukov-Hansen(2006)} implies 
	\begin{align}
	\hat{u}_g^*(\tau) & = \left[Q_{\Psi,X}^\top(\tau) Q_{\Psi,\Psi}^{-1}(\tau) \omega^\top A_1(\tau)  \omega Q^{-1}_{\Psi,\Psi}(\tau) Q_{\Psi,X}(\tau)\right]^{-1}Q_{\Psi,X}^\top(\tau) Q_{\Psi,\Psi}^{-1}(\tau) \omega^\top A_1(\tau) \omega Q^{-1}_{\Psi,\Psi}(\tau) \notag \\
	& \times \biggl[r_nI_n(\tau) + r_nI_{n,g}(\tau) - \sum_{j \in [J]} g_j \xi_j Q_{\Psi,X,j}(\tau)r_n (\beta_0(\tau) - \beta_n(\tau)) \notag \\
	& -\sum_{j \in [J]} g_j \xi_j Q_{\Psi,\Psi,j}(\tau) \begin{pmatrix}
	r_n(\hat{\gamma}^r(\tau) -\gamma_n(\tau)) \\
	0 
	\end{pmatrix}\biggr]   + o_p(1), 
	\label{eq:uhatstar}
	\end{align}
	where the $o_p(1)$ term holds uniformly over $\tau \in \Upsilon$. Subtracting \eqref{eq:alpha_expand2} from \eqref{eq:uhatstar}, we have
	\begin{align*}
	& r_n(\hat{\beta}_g^*(\tau) - \hat{\beta}(\tau))\\
	& = \Gamma (\tau) \biggl[r_nI_{n,g}(\tau) - \sum_{j \in [J]} g_j \xi_j Q_{\Psi,X,j}(\tau)r_n (\beta_0(\tau) - \beta_n(\tau)) -\sum_{j \in [J]} g_j \xi_j Q_{\Psi,\Psi,j}(\tau) \begin{pmatrix}
	r_n(\hat{\gamma}^r(\tau) -\gamma_n(\tau)) \\
	0 
	\end{pmatrix}\biggr] \\
	& + o_p(1).
	\end{align*}
	Assumption \ref{ass:boot}(i) implies  $Q_{\Psi,\Psi,j}(\tau)$, and thus, $[Q_{\Psi,\Psi}(\tau)]^{-1} Q_{\Psi,\Psi,j}(\tau)$ are block diagonal, i.e., 
	\begin{align*}
	[Q_{\Psi,\Psi}(\tau)]^{-1} Q_{\Psi,\Psi,j}(\tau) = \begin{pmatrix}
	Q_{d_w \times d_w} & 0_{d_w \times d_\phi} \\
	0_{ d_\phi \times d_w} &  Q_{d_\phi \times d_\phi}
	\end{pmatrix}.
	\end{align*}
	Then, 
	\begin{align*}
	\omega Q^{-1}_{\Psi,\Psi}(\tau)Q_{\Psi,\Psi,j}(\tau)  \begin{pmatrix}
	r_n(\hat{\gamma}^r(\tau) -\gamma_n(\tau)) \\
	0 
	\end{pmatrix}& = \omega [ Q_{\Psi,\Psi}(\tau)]^{-1} Q_{\Psi,\Psi,j}(\tau)\begin{pmatrix}
	r_n(\hat{\gamma}^r(\tau) -\gamma_n(\tau)) \\
	0 
	\end{pmatrix} \\
	& = [0_{d_\phi \times d_w},\mathbb{I}_{d_\phi}] 
	\begin{pmatrix}
	Q_{d_w \times d_w} & 0_{d_w \times d_\phi} \\
	0_{ d_\phi \times d_w} &  Q_{d_\phi \times d_\phi}
	\end{pmatrix}
	\begin{pmatrix}
	r_n(\hat{\gamma}^r(\tau) -\gamma_n(\tau)) \\
	0_{d_\phi \times 1} 
	\end{pmatrix} = 0.
	\end{align*}
	In addition, for the same reason, we have 
	\begin{align*}
	\Gamma (\tau)f_\tau(D_{i,j},\beta_n(\tau),\gamma_n(\tau),0) = \tilde{\Gamma}(\tau)\tilde{f}_\tau(D_{i,j},\beta_n(\tau),\gamma_n(\tau),0).
	\end{align*}
	
	Therefore, we have
	\begin{align*}
	& r_n(\hat{\beta}_g^*(\tau) - \hat{\beta}(\tau))\\
	& = \Gamma (\tau) r_n I_{n,g}(\tau) - \sum_{ j \in [J]}g_j \xi_j \Gamma (\tau)Q_{\Psi,\Psi,j}(\tau)r_n (\beta_0(\tau) - \beta_n(\tau)) + o_p(1),\\
	& = \tilde{\Gamma}(\tau) \sum_{ j \in [J]}\xi_j g_jr_n (\mathbb{P}_{n,j} - \overline{\mathbb{P}}_{n,j})\tilde{f}_\tau(D_{i,j},\beta_n(\tau),\gamma_n(\tau),0) \\
	& - \sum_{ j \in [J]}g_j \xi_j \Gamma (\tau)Q_{\Psi,X,j}(\tau)r_n (\beta_0(\tau) - \beta_n(\tau)) + o_p(1) \\
	& = \tilde{\Gamma}(\tau) \sum_{ j \in [J]}\xi_j g_j\mathcal Z_j + \overline a^*_g(\tau) \mu_\beta(\tau) + o_p(1),
	\end{align*}
	where the $o_p(1)$ term holds uniformly over $\tau \in \Upsilon$.
	
\end{proof}

\subsection{Technical Results for the Restricted Estimator}
\begin{lemma}
	Suppose Assumptions  \ref{ass:id}--\ref{ass:id2} hold. Then, 
	\begin{align*}
	\sup_{\tau \in \Upsilon} ||\hat{\gamma}^r(\tau)  -\gamma_n(\tau)||_2 = O_p(r_n^{-1}). 
	\end{align*}
	\label{lem:restricted}
\end{lemma}
\begin{proof}
	We note that 
	\begin{align*}
	\sup_{\tau \in \Upsilon}||\beta_0(\tau) - \beta_n(\tau)||_2 = O(r_n^{-1}). 
	\end{align*}
	Plugging this into  \eqref{eq:betagammarate}, Lemma \ref{lem:basic} with $b_n(\tau) = \beta_0(\tau)$ implies 
	\begin{align*}
	\sup_{\tau \in \Upsilon}||\hat{\gamma}^r(\tau) - \gamma_n(\tau)||_2 = O_p(r_n^{-1}). 
	\end{align*}
	
	
\end{proof}

\subsection{Lemma used in the Proof of Theorem \ref{thm:studentized}}
\label{sec:pf_lem_sec3end}


\begin{lemma}
	Suppose Assumptions \ref{ass:id}--\ref{ass:id2} and \ref{ass:J} hold.  Then, we have
	\begin{align*}
	\hat G^\top(\tau)\omega \mathbb{P}_{n,j}\hat{f}_\tau(D_{i,j}, \hat{\beta}(\tau),\hat{\gamma}(\tau),0) & = G^\top(\tau)\mathbb{P}_{n,j} \tilde{f}_\tau(D_{i,j}, \beta_n(\tau),\gamma_n(\tau),0) \\
	&- b_j(\tau)\tilde{\Gamma}(\tau) \mathbb{P}_{n} \tilde{f}_\tau(D_{i,j}, \beta_n(\tau),\gamma_n(\tau),0) + o_p(r_n^{-1}),
	\end{align*}
	\begin{align*}
	\hat G^\top(\tau) \omega \mathbb{P}_{n,j} \hat{f}_\tau(D_{i,j}, \hat{\beta}_g^*(\tau),\hat{\gamma}_g^*(\tau),0) & = G^\top(\tau)\mathbb{P}_{n,j} \tilde{f}_\tau(D_{i,j}, \beta_n(\tau),\gamma_n(\tau),0) \\
	& - b_j(\tau) \tilde{\Gamma}(\tau)\mathbb{P}_n \tilde{f}_\tau(D_{i,j}, \beta_n(\tau),\gamma_n(\tau),0) \\
	& - b_j(\tau) \tilde{\Gamma}(\tau)\sum_{\tilde{j} \in J} g_{\tilde{j}} \xi_{\tilde{j}}  \mathbb{P}_{n,\tilde{j}} \tilde{f}_\tau(D_{i,\tilde{j}}, \beta_n(\tau),\gamma_n(\tau),0) \\
	& + \overline{a}_g^*(\tau)  b_j(\tau) (\beta_0(\tau) - \beta_n(\tau)) + o_p(r_n^{-1}), 
	\end{align*}
	\begin{align*}
	\hat G^\top(\tau)\omega \mathbb{P}_{n,j}\hat{f}_\tau(D_{i,j}, \beta_0(\tau),\hat{\gamma}^r(\tau),0) & = G^\top(\tau)\mathbb{P}_{n,j} \tilde{f}_\tau(D_{i,j}, \beta_n(\tau),\gamma_n(\tau),0)  \\
	&- b_j(\tau)(\beta_0(\tau) - \beta_n(\tau)) + o_p(r_n^{-1}),
	\end{align*}
	and
	\begin{align*}
	\hat G^\top(\tau) \omega \mathbb{P}_{n,j}\hat f_{\tau,g}^*(D_{i,j}) & = g_j G^\top(\tau)\mathbb{P}_{n,j} \tilde{f}_\tau(D_{i,j}, \beta_n(\tau),\gamma_n(\tau),0) - (g_j - \overline{a}_g^*(\tau)) b_j(\tau) (\beta_0(\tau) - \beta_n(\tau)) \\
	& - b_j(\tau)\tilde{\Gamma}(\tau)\sum_{\tilde{j} \in J} g_{\tilde{j}} \xi_{\tilde{j}}  \mathbb{P}_{n,\tilde{j}} \tilde{f}_\tau(D_{i,\tilde{j}}, \beta_n(\tau),\gamma_n(\tau),0) +  o_p(r_n^{-1}),
	\end{align*}
	where all the $o_p(r_n^{-1})$ terms hold uniformly over $\tau \in \Upsilon$ and $\overline{a}_g^*(\tau) = \sum_{ j \in [J]} \xi_j g_j a_j(\tau)$. 
	\label{lem:stud}
\end{lemma}
\begin{proof}
	
	For the first result, we have
	\begin{align}
	& \mathbb{P}_{n,j}\hat{f}_\tau(D_{i,j}, \hat{\beta}(\tau),\hat{\gamma}(\tau),0) \notag  \\
	& =  \overline{\mathbb{P}}_{n,j}\hat{f}_\tau(D_{i,j}, \hat{\beta}(\tau),\hat{\gamma}(\tau),0) + (\mathbb{P}_{n,j} - \overline{\mathbb{P}}_{n,j})\hat{f}_\tau(D_{i,j}, \hat{\beta}(\tau),\hat{\gamma}(\tau),0) \notag \\
	& = \overline{\mathbb{P}}_{n,j}\hat{f}_\tau(D_{i,j}, \hat{\beta}(\tau),\hat{\gamma}(\tau),0) +  (\mathbb{P}_{n,j} - \overline{\mathbb{P}}_{n,j})f_\tau(D_{i,j}, \beta_n(\tau),\gamma_n(\tau),0) + o_p(r_n^{-1}) \notag \\
	& = \mathbb{P}_{n,j} f_\tau(D_{i,j}, \beta_n(\tau),\gamma_n(\tau),0) - Q_{\Psi,X,j}(\tau)(\hat{\beta}(\tau) - \beta_n(\tau)) - Q_{\Psi,\Psi,j}(\tau) \begin{pmatrix}
	\hat{\gamma}(\tau) - \gamma_n(\tau)\\
	0
	\end{pmatrix} + o_p(r_n^{-1})\notag \\
	& = \mathbb{P}_{n,j} f_\tau(D_{i,j}, \beta_n(\tau),\gamma_n(\tau),0)  - Q_{\Psi,X,j}(\tau)\tilde{\Gamma}(\tau) \mathbb{P}_{n} \tilde{f}_\tau(D_{i,j}, \beta_n(\tau),\gamma_n(\tau),0) \notag \\
	& - Q_{\Psi,\Psi,j}(\tau) \begin{pmatrix}
	\hat{\gamma}(\tau) - \gamma_n(\tau)\\
	0
	\end{pmatrix} + o_p(r_n^{-1}), 
	\label{eq:t1stu}
	\end{align}
	where the $o_p(r_n^{-1})$ term holds uniformly over $\tau \in \Upsilon$, the second equality is by Assumption \ref{ass:asym}(iii), the third equality is by Assumption \ref{ass:asym}(ii) and the fact that 
	$$\sup_{\tau \in \Upsilon}\left(||\hat{\beta}(\tau) - \beta_n(\tau)||_2 + ||\hat{\gamma}(\tau) - \gamma_n(\tau)||_2\right) = o_p(r_n^{-1}),$$ as shown in Lemma \ref{lem:expand}, and the last equality is by Lemma \ref{lem:expand} and the fact that, by Assumption \ref{ass:boot}(i),  
	\begin{align*}
	\Gamma (\tau)f_\tau(D_{i,j}, \beta_n(\tau),\gamma_n(\tau),0) = \tilde{\Gamma}(\tau)\tilde{f}_\tau(D_{i,j}, \beta_n(\tau),\gamma_n(\tau),0). 
	\end{align*}
	In addition, under Assumption \ref{ass:boot}(i), we have
	\begin{align*}
	\omega Q_{\Psi,\Psi,j}(\tau)
	\begin{pmatrix}
	\hat{\gamma}(\tau) -\gamma_n(\tau) \\
	0_{d_\phi \times 1} 
	\end{pmatrix} = 0,
	\end{align*}
	and thus, 
	\begin{align*}
	& \hat G^\top(\tau)\omega \mathbb{P}_{n,j}\hat{f}_\tau(D_{i,j}, \hat{\beta}(\tau),\hat{\gamma}(\tau),0) \\
	& =  G^\top(\tau) \omega \mathbb{P}_{n,j}  f_\tau(D_{i,j}, \beta_n(\tau),\gamma_n(\tau),0) - b_j(\tau)\tilde{\Gamma}(\tau) \mathbb{P}_{n} \tilde{f}_\tau(D_{i,j}, \beta_n(\tau),\gamma_n(\tau),0) + o_p(r_n^{-1}) \\
	& = G^\top(\tau)\mathbb{P}_{n,j} \tilde{f}_\tau(D_{i,j}, \beta_n(\tau),\gamma_n(\tau),0) - b_j(\tau)\tilde{\Gamma}(\tau) \mathbb{P}_{n} \tilde{f}_\tau(D_{i,j}, \beta_n(\tau),\gamma_n(\tau),0) + o_p(r_n^{-1}),
	\end{align*}
	where the $o_p(r_n^{-1})$ term holds uniformly over $\tau \in \Upsilon$.
	
	For the second result, we have
	\begin{align*}
	& \mathbb{P}_{n,j}\hat{f}_\tau(D_{i,j}, \hat{\beta}_g^*(\tau),\hat{\gamma}_g^*(\tau),0) \\
	& = \mathbb{P}_{n,j} f_\tau(D_{i,j}, \beta_n(\tau),\gamma_n(\tau),0) - Q_{\Psi,X,j}(\tau)(\hat{\beta}_g^*(\tau) - \beta_n(\tau)) - Q_{\Psi,\Psi,j}(\tau) \begin{pmatrix}
	\hat{\gamma}_g^*(\tau) - \gamma_n(\tau)\\
	0
	\end{pmatrix} + o_p(r_n^{-1})\\
	& = \mathbb{P}_{n,j} f_\tau(D_{i,j}, \beta_n(\tau),\gamma_n(\tau),0)  - Q_{\Psi,X,j}(\tau) \tilde \Gamma (\tau)\sum_{\tilde{j} \in J} g_{\tilde{j}} \xi_{\tilde{j}}  \mathbb{P}_{n,\tilde j} \tilde f_\tau(D_{i, \tilde j}, \beta_n(\tau),\gamma_n(\tau),0) \\
	& - Q_{\Psi,X,j}(\tau)(\hat{\beta}(\tau) - \beta_n(\tau))+ Q_{\Psi,X,j}(\tau)  \overline{a}_g^*(\tau)(\beta_0(\tau) - \beta_n(\tau)) \\
	& - Q_{\Psi,\Psi,j}(\tau) \begin{pmatrix}
	\hat{\gamma}_g^*(\tau) - \gamma_n(\tau)\\
	0
	\end{pmatrix} + o_p(r_n^{-1}), 
	\end{align*}
	where the first equality is due to the same argument in \eqref{eq:t1stu} and the second equality is due to Lemma \ref{lem:bootexpand}. Then, by Lemma \ref{lem:expand}, we have
	\begin{align*}
	& \hat G^\top(\tau)\omega \mathbb{P}_{n,j}\hat{f}_\tau(D_{i,j}, \hat{\beta}_g^*(\tau),\hat{\gamma}_g^*(\tau),0) \\
	& = G^\top(\tau)\mathbb{P}_{n,j} \tilde{f}_\tau(D_{i,j}, \beta_n(\tau),\gamma_n(\tau),0) - b_j(\tau) \tilde{\Gamma}(\tau)\mathbb{P}_n \tilde{f}_\tau(D_{i,j}, \beta_n(\tau),\gamma_n(\tau),0) \\
	& - b_j(\tau) \tilde{\Gamma}(\tau)\sum_{\tilde{j} \in J} g_{\tilde{j}} \xi_{\tilde{j}}  \mathbb{P}_{n,\tilde{j}} \tilde{f}_\tau(D_{i,\tilde{j}}, \beta_n(\tau),\gamma_n(\tau),0) + \overline{a}_g^*(\tau)  b_j(\tau) (\beta_0(\tau) - \beta_n(\tau)) + o_p(r_n^{-1}). 
	\end{align*}
	
	For the third result, by the same argument in \eqref{eq:t1stu},  we have
	\begin{align*}
	& \mathbb{P}_{n,j}\hat{f}_\tau(D_{i,j}, \beta_0(\tau),\hat{\gamma}^r(\tau),0) \\
	& = \mathbb{P}_{n,j} f_\tau(D_{i,j}, \beta_n(\tau),\gamma_n(\tau),0) - Q_{\Psi,X,j}(\tau)(\beta_0(\tau) - \beta_n(\tau)) - Q_{\Psi,\Psi,j}(\tau) \begin{pmatrix}
	\hat{\gamma}^r(\tau) - \gamma_n(\tau)\\
	0
	\end{pmatrix} + o_p(r_n^{-1}).
	\end{align*}
	Then, we have
	\begin{align*}
	& \hat G^\top(\tau) \omega \mathbb{P}_{n,j}\hat{f}_\tau(D_{i,j}, \beta_0(\tau),\hat{\gamma}^r(\tau),0) \\
	& = G^\top(\tau)\mathbb{P}_{n,j} \tilde{f}_\tau(D_{i,j}, \beta_n(\tau),\gamma_n(\tau),0)  - b_j(\tau)(\beta_0(\tau) - \beta_n(\tau)) + o_p(r_n^{-1}). 
	\end{align*}
	
	Combining the previous three results, we have
	\begin{align*}
	\hat G^\top(\tau) \omega \mathbb{P}_{n,j}\hat f_{\tau,g}^*(D_{i,j}) & = g_j G^\top(\tau)\mathbb{P}_{n,j} \tilde{f}_\tau(D_{i,j}, \beta_n(\tau),\gamma_n(\tau),0) - (g_j - \overline{a}_g^*(\tau)) b_j(\tau) (\beta_0(\tau) - \beta_n(\tau)) \\
	& - b_j(\tau)\tilde{\Gamma}(\tau)\sum_{\tilde{j} \in J} g_{\tilde{j}} \xi_{\tilde{j}}  \mathbb{P}_{n,\tilde{j}} \tilde{f}_\tau(D_{i,\tilde{j}}, \beta_n(\tau),\gamma_n(\tau),0) +  o_p(r_n^{-1}).
	\end{align*}
	All the $o_p(r_n^{-1})$ terms in this proof hold uniformly over $\tau \in \Upsilon$. 
\end{proof}

\bibliographystyle{ecta}
\bibliography{Biblio_boot_few_clusters}

\end{document}